\documentclass[11pt]{article}
\usepackage{amsfonts}
\usepackage{amssymb,amsmath,amsthm}
\usepackage{graphicx}
\usepackage[dvips, paper=letterpaper, top=0.95in, bottom=.6in, left=0.9in, right=0.95in, nohead, includefoot, footskip=.4in]{geometry}
\usepackage[authoryear,round]{natbib}
\usepackage{environ}
\usepackage{color}
\usepackage{epsfig}
\usepackage{caption}
\usepackage{verbatim}
\usepackage{subcaption}
\usepackage{float}
\usepackage{setspace}
\usepackage{enumitem}
\setlist[enumerate]{itemsep=0mm}
\usepackage{calc}
\usepackage[pdfpagemode=UseOutlines,hyperfootnotes=false, plainpages=false]{hyperref}
\usepackage{mathtools}
\usepackage{dsfont}
\usepackage[english]{babel}
\usepackage{amsthm}
\usepackage{bm}
\usepackage{breqn}
\usepackage{bigints}
\usepackage{tikz}
\usepackage{nicefrac}
\usepackage{xurl}

\usepackage[font=normalsize,labelfont=bf]{caption}
\usepackage[font=footnotesize,labelfont=bf]{subcaption}
\usepackage{sectsty}
\allsectionsfont{\normalfont\sffamily\bfseries}
\sffamily

\newtheoremstyle{theoremsansserif} 
    {\topsep}                    
    {\topsep}                    
    {\itshape}                   
    {}                           
    {\sffamily\bfseries }        
    {.}                          
    {.5em}                       
    {}  

\theoremstyle{theoremsansserif}

\newtheorem{lemma}{Lemma}
\newtheorem{corollary}{Corollary}
\newtheorem{proposition}{Proposition}
\newtheorem{remark}{Remark}
\newtheorem{definition}{Definition}
\newtheorem{claim}{Claim}

\newtheorem*{example*}{Example}

\newtheorem{theorem}{Theorem}

\newcommand{\R}{\mathbb{R}}
\newcommand{\X}{\mathcal{X}}
\newcommand{\A}{\mathcal{A}}
\newcommand{\w}{\hat{w}}
\newcommand{\E}{\mathbb{E}}

\DeclareMathOperator{\argmin}{argmin}
\DeclareMathOperator{\argmax}{argmax}

\newenvironment{newsanti}{\color{black}}{}
\newcommand{\rk}[1]{{\color{black} #1}}

\begin{document}

\title{\sf\textbf{Contextual Standard Auctions with Budgets:}\\
\sf\textbf{Revenue Equivalence and Efficiency Guarantees}}

\author{
\sf Santiago R. Balseiro\\
\sf Columbia University \\
\small\texttt{srb2155@columbia.edu}
\and
\sf Christian Kroer \\
\sf Columbia University\\
\small\texttt{christian.kroer@columbia.edu}
\and
\sf Rachitesh Kumar\\
\sf Columbia University \\
\small\texttt{rk3068@columbia.edu}}

\date{\vspace{1em}
\sf This version: \today\\
\sf First version: February 20, 2021}

\maketitle

\allowdisplaybreaks

\begin{abstract}
The internet advertising market is a multi-billion dollar industry, in which advertisers buy thousands of ad placements every day by repeatedly participating in auctions. An important and ubiquitous feature of these auctions is the presence of campaign budgets, which specify the maximum amount the advertisers are willing to pay over a specified time period. In this paper, we present a new model to study the equilibrium bidding strategies in standard auctions, a large class of auctions that includes first- and second-price auctions, for advertisers who satisfy budget constraints on average. Our model dispenses with the common, yet unrealistic assumption that advertisers' values are independent and instead assumes a contextual model in which advertisers determine their values using a common feature vector. We show the existence of a natural value-pacing-based Bayes-Nash equilibrium under very mild assumptions. Furthermore, we prove a revenue equivalence showing that all standard auctions yield the same revenue even in the presence of budget constraints. Leveraging this equivalence, we prove Price of Anarchy bounds for liquid welfare and structural properties of pacing-based equilibria that hold for all standard auctions. 

{\newsanti
In recent years, the internet advertising market has adopted first-price auctions as the preferred paradigm for selling advertising slots. Our work thus takes an important step toward understanding the implications of the shift to first-price auctions in internet advertising markets by studying how the choice of the selling mechanism impacts revenues, welfare, and advertisers' bidding strategies.}

\bigskip

\end{abstract}

\thispagestyle{empty}
\pagebreak
\setcounter{page}{1}

\setstretch{1.5}
\pagebreak

\section{Introduction}



In 2019, the revenue from selling internet ads in the US surpassed \textdollar 129 billion.\footnote{See \url{https://www.emarketer.com/content/us-digital-ad-spending-2019}.} A large fraction of these are sold on ad platforms operated by tech-giants like Google, Facebook and Twitter. These platforms facilitate the sale of ads by acting as intermediaries between advertisers and publishers. Millions of ad slots are sold every day using auctions in which advertisers bid based on user-specific information (such as geographical location, cookies, and historical activity, among others). The advertisers repeatedly participate in these auctions, with the aim of using their advertising budget to maximize their reach, through a combination of user-specific targeting and bid-optimization. The presence of budgets introduces significant challenges, as it links different auctions together.

With billions of dollars at stake, the auction format plays a crucial role. In recent years, a major shift has occurred towards using first-price auctions as the preferred mode of selling display ads, as opposed to the earlier standard of using second-price auctions. For example, in 2019, Google, which is one of the industry leaders, announced a shift to the first-price auction format for its ad exchange.\footnote{See \url{https://www.blog.google/products/admanager/rolling-out-first-price-auctions-google-ad-manager-partners/}.} In 2020, Twitter also made the move to first-price auctions for the sale of mobile app advertising slots.\footnote{See \url{https://www.mopub.com/en/blog/first-price-auction}.} First-price auctions typically lead to more complicated bidding behavior, because, unlike second-price auctions, truthful bidding is not an equilibrium in the first-price setting.

This paper attempts to capture the salient features of these display ad auctions, with a focus on the newly adopted first-price auctions. While equilibrium behavior in first-price auctions has been studied extensively, very little attention has been given to the effects of budget constraints and user-specific information.
Budget constraints span the auctions, which means that advertisers must strategize about their bids across all auctions simultaneously.
User-specific information leads to correlation between the valuations that different advertisers have for a particular ad opportunity, whereas the literature largely focuses on i.i.d. valuations.
Our paper aims to shed some light on these aspects by introducing and analyzing a framework for general standard auctions, including first-price auctions, that incorporates budgets and context-based valuations. In particular, the main questions we tackle are: How does the auction format affect the equilibrium strategies of budget-constrained bidders? How does the auction format impact the revenue of the ad platform and {\newsanti the efficiency of the market}?

\rk{\subsection{Main Contributions}} \label{sec:main-contributions}

We incorporate the availability of user-specific information (that is common to all advertisers) via a contextual valuation model, which allows us to capture correlation in values. User information and buyer targeting criteria are modeled as vectors, with the value that an advertiser gains from her ad being shown to the user being given by the inner product of these vectors, or a function thereof.
Each advertiser has a (possibly non-binding) budget which must be satisfied in expectation.
Such budget constraints are well motivated in practice due to the large number of auctions, and have been the subject of study in previous work on budget management~\citep{gummadi2012repeated,abhishek2013optimal,balseiro2015repeated,balseiro2017budget}.
Our main contribution is to introduce a framework that allows for the study of \emph{standard auctions}, which are auctions in which the highest bidder wins, in the presence of budget constraints and contextual valuations. To the best of our knowledge, this is the first analysis of standard auctions in the presence of average budget constraints.

Typically, the targeting criteria and the budget of an advertiser are not known to her competitors. This motivates us to model the participation of advertisers in the market as a non-atomic game of incomplete information in which each advertiser assumes that the other advertisers are being drawn from some common underlying distribution. In this game, the decision problem faced by each advertiser is to optimize her utility while satisfying her budget constraint in expectation. This expectation is taken over all the potential auctions she could end up participating in, i.e., the expectation is over users and competing advertisers. Our non-atomic game allows us to sidestep the possibility of multiple buyers tying in the auction and leads to simple and intuitive equilibrium strategies.


\paragraph{Equilibrium Analysis.} A contribution of this paper is to prove the existence of a remarkably simple Bayes-Nash equilibrium strategy using a novel topological argument. In our non-atomic model, there is a continuum of advertiser types, and a strategy for each advertiser type is a function which maps contexts to bids. Directly proving existence of an equilibrium in this complicated strategy space in the presence of budget constraints turns out to be difficult. We side-step this difficulty by establishing strong duality for the constrained non-convex optimization problem faced by each advertiser type and characterizing the primal optimum in terms of the dual optimum.

We propose a remarkably simple class of strategies, which we dub \emph{value-pacing-based strategies}. These strategies build on the symmetric equilibrium strategies of the standard i.i.d.~setting, inheriting their interpretability in the process. A value-pacing-based strategy recommends that each advertiser should shade her value by a multiplicative factor to manage her budget, and then bid using the symmetric equilibrium strategy from the standard i.i.d.~setting---like she would in the absence of budgets---but assuming that competitors' values are also paced. This naturally extends multiplicative bid pacing/shading, which is one of the several ways budgets are managed in practice, to non-truthful auctions \citep{balseiro2017budget, conitzer2017multiplicative,conitzer2019pacing}.
To the best of our knowledge, our value-pacing approach is the first to show optimal pacing-based strategies outside of truthful auctions.



Our non-atomic game has a pacing (dual) multiplier for each buyer type, which are uncountably many in cardinality. This leads to an infinite-dimensional equilibrium space even after moving to the simpler dual space. In infinite dimensions, establishing even the simple prerequisites of any fixed-point theorem, namely compactness and continuity, can be an ordeal; one which requires careful topological arguments. While other papers have also analyzed equilibrium strategies in the dual space (see, e.g., \citealt{balseiro2015repeated,gummadi2012repeated}), these consider settings with finitely-many pacing multipliers in which establishing compactness is a trivial task. The main technical contribution of this paper is twofold: (i) choosing the right topological space for the pacing multipliers based on their monotonicity properties, (ii) establishing compactness and continuity in this carefully chosen space. As we discuss in Subsection~\ref{sec:fixed-point}, this choice of topology is far from obvious. In fact, to the best of our knowledge, all of the topologies used in standard fixed-point arguments for infinite-dimensional spaces (see \citet{aliprantis2006infinite} for examples) prove insufficient in the setting we consider, which compels us to carefully exploit the structural properties of pacing and work with the topological space of multivariate-functions of bounded variation. We believe that the tools developed in this paper might be useful in other non-atomic games.


\paragraph{Standard Auctions and Revenue Equivalence.} Our framework accommodates anonymous auction formats in which the highest bidder wins, such as the second-price auction and all-pay auction (even in the presence of reserve prices). In its full generality, it acts as a powerful black-box: it takes as input any Bayes-Nash equilibrium for the well-studied standard i.i.d.~setting, composes it with value-pacing, and outputs a Bayes-Nash equilibrium for our model. Surprisingly, we show that, for a fixed distribution over advertisers and users, the same multiplicative factors can be used by the advertisers to shade their values in the equilibrium strategies for all standard auctions.
This fact allows us to compare revenues across auction formats. We prove that, in the presence of in-expectation budget constraints, the revenue generated in a value-pacing-based equilibrium is the same for all standard auctions. This is in sharp contrast to the case when budgets constraints are strict, where revenue equivalence is known not to hold~\citep{che1998standard}. In light of the recent shift from second-price auction to first-price auction by many ad platforms, the ability to compare budget management in both first and second-price auctions is an especially relevant aspect of our framework. A recent paper of \citet{goke2021learning} empirically investigated the revenue impact resulting from this recent switch. \citet{goke2021learning} found that, after a brief adjustment period, publishers' revenues under first-price auctions returned to the same levels as they were under second-price auctions before the change.
Since existing revenue equivalence results do not apply to the display-ad industry (due to budget constraints and dependencies in valuations), our theory offers the first principled justification for this empirical finding by establishing revenue equivalence in the presence of contextual values and in-expectation budget constraints.

\paragraph{Price of Anarchy and Structural Results.} We leverage our revenue equivalence result to establish efficiency guarantees and structural properties that hold for all standard auctions. In particular, we prove a $(1/2)$-lower-bound for the Price of Anarchy of liquid welfare (a notion of efficiency that incorporates budget constraints) for our value-pacing-based equilibria. Our result implies that the liquid welfare of a pacing equilibrium is at most $1/2$ of the liquid welfare of the best possible allocation.
On the structural front, we study how value-pacing-based equilibrium strategies change with buyer type, which consists of a weight vector (representing targeting criteria) and a budget. We show that budget-constrained buyers with identical budgets and co-linear weights for features get paced to the same value in equilibrium. This shows that any enhancement in the ad quality without changing its targeting criteria, which corresponds to scaling up the weight vector, is futile in the absence of an increase in budget. Moreover, we also study how advertisers should change their targeting criteria (as represented by their weight vector) to maximize their utility.

\paragraph{Numerical Experiments.} To test our model, we run numerical experiments after making appropriate discretizations. The outcomes of these experiments are strikingly close to our theoretical predictions. In particular, despite the discontinuities introduced by discretization, budget violations are small, and moreover, the equilibrium strategies are in strong adherence to the structural properties derived theoretically.


\rk{\subsection{Related Work}}\label{sec:related-work}

Beyond the works already mentioned, there is a large literature on online auctions.
We discuss the existing work that is most closely related to ours.
In keeping with previous work on auctions, from now on, we will use the terms buyers and items in place of advertisers and users. 

Auctions with budget-constrained buyers have been modeled in a variety of ways, most of which is focused on second-price auctions. From a technical standpoint, the closest to our work \cite{balseiro2015repeated}, which considers randomly-arriving budget-constrained buyers in a fluid mean field setting. They show equilibrium existence for second-price auctions, in which buyers use pacing-based strategies. Their model assumes a finite type space and independence of the value distributions of the buyers, whereas our context-based model allows for correlation between buyer values. 
Several other works have also studied repeated second-price auctions with budgets under various models that are less related to ours~\citep{gummadi2012repeated,balseiro2017budget,conitzer2017multiplicative,ChenKK21,balseiro2019learning,ciocan2021tractable}.
Beyond second-price auctions, 
\citet{AggarwalBM19} consider affine constraints (which include budget constraints as a special case) in multi-slot truthful auctions; they show existence of a bid-pacing equilibrium under restrictive assumptions. \citet{Babaioff0HIL21} consider a general model of non-quasi-linear buyers participating in mechanisms that are truthful for quasi-linear buyers. Their model also captures the case of budget constraints as a special case. Moreover, they too prove the existence of a pacing-based equilibrium in their model. 
None of the aforementioned existing work addresses strategic bidding in non-truthful auctions such as first-price auctions with budget-constrained buyers. 

\citet{conitzer2019pacing} and \citet{borgs2007dynamics} study pacing in a first-price context, but both disregard strategic behavior on behalf of the buyers.
This is also the case for a long line of work that models repeated auctions among budget-constrained buyers as an online matching problem with capacity constraints (see \citealt{mehta2013online} for a survey).

Another direction of research considers buyers with ex-post budget constraints (also called strict budget constraints). 
There, first price~\citep{kotowski2020first}, standard auctions~\citep{che1998standard}, optimal auctions~\citep{pai2014optimal}, and auctions with combinatorial constraints~\citep{goel2015polyhedral} have been studied.
In contrast to our revenue equivalence results, \citet{che1998standard} show that with strict budget constraints first-price auctions yield higher revenue than second-price auctions. 
These models are different from our setting which only requires budget constraints to hold in expectation at the interim stage. 
In-expectation budget constraints are more appropriate for modeling repeated ad auctions, and yield simpler and more interpretable equilibrium strategies.

Contextual models in which values are based on feature vectors are widely used in the multi-armed bandit literature (for example, see \citealt{langford2007epoch} and \citealt{li2010contextual}), and in pricing~\citep{golrezaei2021dynamic,chen2021statistical,lobel2018multidimensional}. Vector-based valuation models are also connected to low-rank models, which have received attention in previous market design work (see e.g. \citealt{mcmahan2013ad,kroer2019computing}).

Our work also contributes to the literature on equilibrium analysis for non-atomic games. Due to the presence of a continuum of buyer types in our model, the topological arguments we develop bear resemblance to those used in the study of non-atomic games, such as the one addressed in \cite{schmeidler1973equilibrium}, though continuity is by assumption in \cite{schmeidler1973equilibrium}, whereas achieving continuity is at the heart of our proof.

%
%


\section{Model}\label{sec:model}

We consider the setting in which a seller (i.e., the advertising platform) plans to sell an indivisible item to one of $n$ buyers (i.e., the advertisers) using an auction. We adopt a feature-based valuation model for the buyer. More precisely, the item type is represented using a vector $\alpha$ belonging to the set $A \subset \R^d$, where each component of $\alpha$ can be interpreted as a feature. We also refer to $\alpha$ as the \emph{context}. Each buyer type is represented using a vector $(w,B)$ belonging to the set $\Theta \subset \R^{d+1}$ of possible buyer types, where the last component $B$ denotes her budget and the first $d$ components $w$ capture the weights she assigns to each of the $d$ features. The value (maximum willingness to pay) that buyer type $(w,B)$ has for item $\alpha$ is given by the inner product $w^T\alpha$. For simplicity of notation and ease of exposition, we will state our results under this linear relationship between values and the features, but our model and results can be extended to accommodate non-linear response functions (such as the logistic function) that are commonly used in practice (see Appendix~\ref{appendix:non_linear} for a more detailed discussion). We will use $\omega = \max_{(w,B) \in \Theta, \alpha \in A} w^T\alpha$ to denote the maximum value that a buyer can have for an item.

We assume that the context of the item to be auctioned is drawn from some distribution $F$ over the set of possible items types $A$. Furthermore, the type for every buyer is drawn according to some distribution $G$ over the set of possible buyer types $\Theta$, independently of the other buyers and the choice of the item. Note that, by virtue of our context-based valuation model, the values of the $n$ buyers for the item need not be independent. In line with standard models used for Bayesian analysis of auctions, we will assume that both $G$ and $F$ are common knowledge, while maintaining that the realized type vector $(w,B)$ associated with a buyer is her private information. Our model allows budgets to be random and correlated with the buyers' weight vector. In addition, we will assume that buyers are unaware of the type of their competing buyers---this implies budgets are private.

To fix ideas, we first consider the case of a first-price auction with reserve prices and then discuss how our results extend to standard auctions in Section~\ref{sec:standard}. In a first-price auction, the seller allocates the item to the highest bidder whenever her bid is above the reserve price and the winning bidder pays her bid. We assume the seller discloses the item type $\alpha$ to the $n$ buyers before bids are solicited from them. As a result, the bid of a buyer on item $\alpha$ can depend on $\alpha$. We use $r: A \to \R$ to specify the publicly known context-dependent reserve prices, where $r(\alpha)$ denotes the reserve price on item type $\alpha$.

The budget of a buyer represents an upper bound on the amount she would like to pay in the auction. We only require that each buyer satisfy her budget constraint in expectation over the item type and competing buyer types. Similar assumptions have been made in the literature (see, e.g., \citealt{gummadi2012repeated, abhishek2013optimal, balseiro2015repeated, balseiro2017budget, conitzer2017multiplicative}). The motivation behind this modeling choice is that budget constraints are often enforced on average by advertising platforms. For example, Google Ads allows daily budgets to be exceeded by a factor of two in any given day, but, over the course of month, the total expenditure never exceeds the daily budget times the days in the month.\footnote{Google Ads Help page defines ``Average Daily Budget": \url{https://support.google.com/google-ads/answer/6312?hl=en}} In-expectation budget constraints are also motivated by the fact that, in practice, buyers typically participate in a large number of auctions and many buyers use stationary bidding strategies. Thus, by the law of large numbers, our model can be interpreted as collapsing multiple, repeated auctions in which item types are drawn i.i.d.~from $F$ into a single one-shot auction with in-expectation constraints.

\paragraph{Notation.} We will use $\R_+$ and $\R_{\geq 0}$ to denote the set of strictly positive and non-negative real numbers, respectively. We will use $G_w$ to denote the marginal distribution of $w$ when $(w,B) \sim G$, i.e., $G_w(K) \coloneqq G(\{(w,B) \in \Theta \mid w \in K\})$ for all Borel sets $K \subset S$. In a similar vein, we will use $\Theta_w$ to denote the set of $w \in \R^d$ such that $(w,B) \in \Theta$ for some $B \in \R$. (Here we abuse notation by using $w$ both as a weight vector variable and as a subscript to denote the projection of a buyer type onto the first $d$ dimensions). Unless specified otherwise, $\|\cdot\|$ denotes the Euclidean norm.

\paragraph{Assumptions.} We will assume that there exist $U, B_{\min} > 0$ such that the set of possible buyer types $\Theta$ is given by $\Theta = (0,U)^d \times (B_{\min},U)$. In a similar vein, we also assume that the set of possible item types $A$ is a subset of the positive orthant $\R_+^d$. We will restrict our attention to $d \geq 2$, which is the regime in which our feature-vector based valuation model yields interesting insights. To completely specify the aforementioned probability spaces, we endow $A$, $\Theta$ and $\Theta_w$ with the Lebesgue $\sigma-$algebra. Moreover, we will assume that the distributions $F$ and $G$ have density functions. Note that the distribution $G$ can be any distribution on $\Theta$, including one with probability zero on some regions of $\Theta$. Thus we can address any buyer distribution, so long as it has a density and is supported on a bounded subset of the strictly-positive orthant with a positive lower bound on the possible budgets. Similarly, $F$ can capture a wide variety of item distributions. It is worth noting that any distribution that lacks a density can be approximated arbitrarily well by a distribution with a density, thereby extending the reach of our results to arbitrary distributions.
\subsection{Equilibrium Concept}

Consider the decision problem faced by a buyer type $(w,B) \in \Theta$ if we fix the bidding strategies of all competing buyers on all possible item types: She wishes to bid on the items in a way that maximizes her expected utility while satisfying her budget constraint in expectation (where the expectation is taken over items and competing buyers' types). As is true in the well-studied standard budget-free i.i.d. setting (\citealt{krishna2009auction}), her optimal strategy depends on the strategies used by the other buyer types. In the standard setting, the symmetric Bayes-Nash equilibrium is an appealing solution concept for the game formed by these interdependent decision problems faced by the buyers. We adopt a similar approach and define the symmetric Bayes-Nash equilibrium for our model. A strategy $\beta^*: \Theta \times A \to \R_{\geq 0}$ (a mapping that specifies what each buyer type should bid on every item) is a Symmetric First-Price Equilibrium if, almost surely over all buyer types, using $\beta^*$ is an optimal solution to a buyer type's decision problem when all other buyer types also use it.

\begin{definition}\label{definition_equilibrium}
	A strategy $\beta^*:\Theta\times A \to \mathbb{R}_{\geq 0}$ is called a Symmetric First-Price Equilibrium (SFPE) if $\beta^*(w, B, \alpha)$ (as a function of $\alpha$) is an optimal solution to the following optimization problem almost surely w.r.t. $(w,B) \sim G$:
   	\begin{align*}
    	\max_{b: A \to \R_{\geq 0}} \quad &\mathbb{E}_{\alpha,\{\theta_i\}_{i=1}^{n-1}}\left[ (w^T \alpha - b(\alpha))\ \mathds{1}\left\{b(\alpha) \geq \max \left( r(\alpha), \left\{ \beta^*(\theta_i, \alpha) \right\}_i \right) \right\} \right]\\
    	\operatorname{s.t.} \quad &\mathbb{E}_{\alpha,\{\theta_i\}_{i=1}^{n-1}}\left[ b(\alpha)\ \mathds{1}\left\{b(\alpha) \geq \max \left( r(\alpha), \left\{ \beta^*(\theta_i, \alpha) \right\}_i \right) \right\} \right] \leq B.\\
   	\end{align*}
\end{definition}
In the buyer's optimization problem the buyer wins whenever her bid $b(\alpha)$ is higher than the reserve price $r(\alpha)$ and all competiting bids $\beta^*(\theta_i, \alpha)$ for $i=1,\ldots,n-1$. Because of the first-price auction payment rule, each bidder pays her bid whenever she wins. For convenience, in the above definition, we are using an infeasible tie-breaking rule which allocates the entire good to every highest bidder. This is inconsequential, and can be replaced by any arbitrary tie-breaking rule, because we will later show (see part (d) of Lemma \ref{pacing_strategy_properties}) that ties are a zero-probability event under our value-pacing-based equilibria.

In our solution concept, it is sufficient that advertisers have Bayesian priors over the maximum competing bid $\max_i\{ \beta^*(\theta_i, \alpha)\}$ to determine a best response. This is aligned with practice as many advertising platforms provide bidders with historical bidding landscapes, which advertisers can use to optimize their bidding strategies \citep{google_landscape}.\footnote{See, for example, \url{https://www.blog.google/products/admanager/rolling-out-first-price-auctions-google-ad-manager-partners/}.} Additionally, we require that budgets are satisfied in expectation over the contexts and buyer types. Connecting back to our repeated auctions interpretation, one can assume competitors' types to be fixed throughout the horizon while contexts are drawn i.i.d.~in each auction. In this case, our solution concept would be appropriate if buyers cannot observe the types of competitors and, in turn, employ stationary strategies that do not react to the market dynamics. Such stationary strategies are appealing because they deplete budgets smoothly over time and are simple to implement. Moreover, it has been previously established that stationary policies approximate well the performance of dynamic policies in non-strategic settings when the number of auctions is large and the maximum value of each auction is small relative to the budget (see, e.g., \citealt{talluri2006theory}).

When the types of bidders is fixed throughout the horizon, a bidder who employs a dynamic strategy could, in principle, profitably deviate by inferring the competitors' types and using this information to optimally shade her bids. Implementing such strategies in practice is challenging because many platforms do not disclose the identity of the winner nor the bids of competitors in real-time (as we discussed above, they mostly provide historical information that is aggregated over many auctions). Moreover, when the number of bidders is large and each bidder competes with a random subset of bidders, such deviations can be shown to not be profitable using mean-field techniques (see, e.g., \citealt{iyer2014mean, balseiro2015repeated}) in our contextual value model as long as values are independent across time. Therefore, our model can be alternatively interpreted as one in which there is a large population of active buyers and each buyer competes with a random subset of buyers. This assumption is well motivated in the context of internet advertising markets because the number of advertisers actively bidding is typically large and, because of sophisticated targeting technologies, advertisers often participate only in a fraction of all auctions.

\rk{

\subsection{\newsanti Ties and the Role of Contexts}

Before moving onto the proof of existence of SFPE, we would like to shed some light on the role played by contexts in our model and results. The assumption that the feature vectors $\alpha$ are drawn from a distribution $F$ which has a density is necessary for our results to hold. In fact, if there was only one deterministic context, an SFPE may fail to exist: we provide an example in Appendix~\ref{appendix:det-context-counter}. The root cause behind the absence of a well-behaved equilibrium in this example is the tension between the proclivity of budgets to cause ties with positive probability (as we demonstrate in Section~\ref{sec:structural}) and the potential lack of equilibria for first-price auctions under value distributions that cause ties with a positive probability. Our example in Appendix~\ref{appendix:det-context-counter} does admit a symmetric equilibrium for second-price auction, thereby demonstrating the added complexity of dealing with first-price auctions.

Issues of tie-breaking have previously come up in a line of related work on pacing-based equilibria in second-price auctions \citep{borgs2007dynamics, balseiro2015repeated, conitzer2017multiplicative, Babaioff0HIL21}, where they were addressed by methods that are some version of randomly perturbing the value of each buyer and enforcing the budget constraint on average over these perturbations. This causes ties to become zero-probability events. It is possible to prove our existence and revenue equivalence results for the case of one deterministic context with value perturbations. However, unlike second-price auctions where bidding truthfully is a dominant strategy, value perturbation is not well-suited for first-price auctions because, even in the absence of budgets, the first-price auction equilibrium strategy would depend on the perturbations. Moreover, our structural results (Proposition~\ref{structural_prop} and Proposition~\ref{thm:targeting}) may not hold for arbitrary perturbations and would require an unjustifiably-strong assumption that carefully coordinates the perturbations across buyer types. That being said, if one is willing to ignore ties, our results continue to hold for a single deterministic context and the reader can safely continue with that setting in mind.

}

\section{Existence of Symmetric First-Price Equilibrium} \label{section_existence}
In this section, we study the existence of SFPE, and show that this existence is achieved by a compelling solution which is interpretable.
We do so in several steps. First, we define a natural parameterized class of value-pacing-based strategies. Then, assuming that the buyer types are using a strategy from this class, we establish strong duality for the optimization problem faced by each buyer type and characterize the primal optimum in terms of the dual optimum. This leads to a substantial simplification of the analysis because it allows us to work in the much simpler dual space. Finally, we establish the existence of a value-pacing-based SFPE by a fixed-point argument over the space of dual-multipliers.

\subsection{Value-Pacing-Based Strategies}\label{sec:value-pacing}

In this paper, pacing refers to multiplicatively scaling down a quantity.\footnote{We use the term value-pacing-based strategies to differentiate it from bid-pacing/bid-shading, which has previously been studied in the context of truthful auctions~\citep{borgs2007dynamics,balseiro2015repeated,balseiro2017budget,conitzer2017multiplicative,conitzer2019pacing}.\label{footnote:value-pacing}} Consider a function $\mu: \Theta \to \mathbb{R}_{\geq 0}$, which we will refer to as the \emph{pacing function}. We define the \emph{paced weight vector} of a buyer with type $(w,B)$ to be $w/(1 + \mu(w,B))$, which is simply the true weight vector $w$ scaled down by the factor $1/(1 + \mu(w,B))$. Similarly, we define the \emph{paced value} of a buyer type $(w,B)$ for item $\alpha$ as $w^T\alpha/(1 + \mu(w,B))$. We will use pacing to ensure that the budget constraints of all buyer types are satisfied, and at the same time, maintain the best response property at equilibrium. The motivation for using pacing as a budget management strategy will become clear in the next section, where we show that the best response of a buyer to other buyers using a value-pacing-based strategy is to also use a value-pacing-based strategy. Before defining the strategy, we set up some preliminaries.

Consider a pacing function $\mu: \Theta \to \R_{\geq 0}$ and an item $\alpha \in A$. Let $\lambda_\alpha^\mu$ denote the distribution of paced values $w^T \alpha/(1 + \mu(w,B))$ for item $\alpha$ when $(w,B) \sim G$. Let $H_\alpha^\mu$ denote the distribution of the highest value $Y:=\max\{X_1,\dots, X_{n-1}\}$ among $n-1$ buyers, when each $X_i \sim \lambda_\alpha^\mu$ is drawn independently for $i \in \{1,\ldots,n-1\}$.
Observe that $H_\alpha^\mu( (-\infty, x]) = \lambda_\alpha^\mu((-\infty, x])^{n-1}$ for all $\alpha \in A$ because the random variables are i.i.d.

For a given item $\alpha \in A$, when $x \geq r(\alpha)$, define the following bidding function,
\begin{align*}
	\sigma_\alpha^\mu(x) \coloneqq x - \int_{r(\alpha)}^{x}  \frac{H_\alpha^\mu(s)}{H_\alpha^\mu(x)} ds,
\end{align*}
where we interpret $\sigma_\alpha^\mu(x)$ to be $0$ if $H_\alpha^\mu(x) = 0$. Moreover, when $x < r(\alpha)$, define $\sigma_\alpha^\mu(x) \coloneqq x$ (we make this choice to ensure that no value below the reserve price gets mapped to a bid above the reserve price, while maintaining continuity). Note that $\sigma_\alpha^\mu(x) = \mathbb E\left[ \max(Y,r) \mid Y < x\right]$. If $\lambda_\alpha^\mu$ has a density, then $\sigma_\alpha^\mu$ is the same as the single-auction equilibrium strategy for a standard first-price auction without budgets, when the buyer values are drawn i.i.d. from $\lambda_\alpha^\mu$ and the item has a reserve price of $r(\alpha)$ (see, e.g., section 2.5 of \citealt{krishna2009auction}). Our value-pacing-based strategy uses $\sigma_\alpha^\mu(x)$ as a building block, by composing it with value-pacing:

\begin{definition}\label{definition_equilibrium_strategy}
	The value-pacing-based strategy $\beta^\mu: \Theta \times A \to \R_{\geq 0}$ for pacing function $\mu: \Theta \to \R_{\geq 0}$ is given by
	\begin{align*}
		\beta^\mu(w,B,\alpha) \coloneqq \sigma_\alpha^\mu \left(\frac{w^T\alpha}{1+\mu(w,B)}\right) \qquad \forall\ (w,B) \in \Theta, \alpha \in A	
	\end{align*}
\end{definition}

The bid $\beta^\mu(w,B,\alpha)$ is the amount that a non-budget-constrained buyer with type $(w,B)$ would bid on item $\alpha$ if she acted as if her paced value was her true value (this is captured by the use of the paced value as the argument for $\sigma_\alpha^\mu$), and believed that the rest of the buyers were also acting in this way (this is captured by the use of $\sigma_\alpha^\mu$). Therefore, our strategy has a simple interpretation: bidders pace their values and then bid as in a first-price auction in which competitors' values are also paced. Consequently, under our strategy bidders are shading their values twice: first when determining their paced values $w^T\alpha/(1+\mu(w,B))$ to account for budget constraints and then again when adopting the bidding function $\sigma_\alpha^\mu$ for the first-price auction. The bidding strategy $\sigma_\alpha^\mu$ optimally trades off two effects: on the one hand, bidding too close to their paced values leaves no utility to buyers because they pay their bid in case of winning and, on the other hand, bidding too low decreases payments at the expense of also decreasing the chance of winning.

Observe that value-pacing-based strategies greatly reduce the degrees of freedom in the system. Instead of specifying a bidding strategy, which is a function, for each buyer type, we only need to specify a scalar, $\mu(w,B)$ for each buyer type. In addition, our dual characterization allow us to optimize over the space of all bidding strategies without imposing any restriction on the class of admissible functions. Having defined value-pacing-based strategies, we are now ready to state our main existence result.

\begin{theorem} \label{main_existence_result}
	There exists a pacing function $\mu: \Theta \to \R_{\geq 0}$ such that the value-pacing-based strategy $\beta^\mu: \Theta \times A \to \R_{\geq 0}$ is a Symmetric First-Price Equilibrium (SFPE).	
\end{theorem}

Before proceeding with the proof of Theorem~\ref{main_existence_result}, we note some of its practical prescriptions: (i) It recommends that buyers should pace their value to manage their budgets. As we will later show, the equilibrium pacing functions for first-price auctions are identical to the ones for second-price auctions. This suggests that pacing-based-budget-management techniques developed for second-price auctions (like \citealt{balseiro2019learning}) can be used for first-price auctions to compute the paced valued. (ii) Advertising platforms typically provide bidding landscapes to the buyers which allow them to compute the optimal bid for a given value. Given a context $\alpha$, if $\mathbb{P}_\alpha^\mu$ represents the equilibrium bidding landscape (distribution of highest competing bids), then we have
\begin{align*}
	\sigma_\alpha^\mu \left( x \right) \in \argmax_b \left( x - b \right) \mathbb{P}_\alpha^\mu(b)
\end{align*}
Thus, the paced value can be combined with the landscape to compute the optimal bid $\beta^\mu(w, B, \alpha)$.

We provide the proof of Theorem~\ref{main_existence_result} in the remaining subsections. First, in Subsection~\ref{subsection_strong_duality}, we show that, if all of the competing buyers are assumed to employ a value-pacing-based strategy, then strong duality holds for the budget-constrained utility maximization problem faced by each buyer type. This allows us to drastically simplify the equilibrium strategy space of each buyer type from a function (mapping contexts to bids) to a single scalar (the dual variable $\mu(w,B)$). Next, in Subsection~\ref{sec:fixed-point}, we prove the existence of a value-pacing-based equilibrium strategy by proving a fixed-point theorem in the dual space of pacing functions. Despite our simplifying move to the dual space, establishing a fixed point is by no means a straightforward task because we are still left with a dual variable for each buyer type and there are (uncountable) infinitely many of those. This leads to an infinite-dimensional fixed-point problem which requires careful topological analysis. We find that the commonly-employed general-purpose topologies fail for our problem, and this motivates us to carefully exploit the structure of pacing to select the right topology.


\subsection{Strong Duality and Best Response Characterization}\label{subsection_strong_duality}

We start by considering the optimization problem faced by an individual buyer with type $(w,B)$ when all competing buyers use the value-pacing-based strategy with pacing function $\mu: \Theta \to \mathbb{R}_{\geq 0}$. Denoting by $Q^\mu(w,B)$ the optimal expected utility of such a buyer, we have
\begin{align*}
    Q^\mu( w, B) = \max_{b: A \to \R_{\geq 0}} \quad &\mathbb{E}_{\alpha,\{\theta_i\}_{i=1}^{n-1}}\left[ (w^T \alpha - b(\alpha))\ \mathds{1}\left\{b(\alpha) \geq \max\left( r(\alpha), \left\{ \beta^\mu(\theta_i, \alpha) \right\}_i \right) \right\} \right]\\
    \textrm{s.t.} \quad &\mathbb{E}_{\alpha,\{\theta_i\}_{i=1}^{n-1}}\left[ b(\alpha)\ \mathds{1}\left\{b(\alpha) \geq \max\left( r(\alpha), \left\{ \beta^\mu(\theta_i, \alpha) \right\}_i \right) \right\}\right] \leq B.
\end{align*}
Our goal in this section is to show that the value-pacing-based strategy put forward in Definition~\ref{definition_equilibrium_strategy} is a best response when competitors are pacing their bids according to a pacing function $\mu$.

\begin{remark}
	Compare $Q^\mu(w,B)$ to the definition of a SFPE (Definition \ref{definition_equilibrium}), and observe that, if we were able to show that there exists $\mu: \Theta \to \R_{\geq 0}$ such that $\beta^\mu(w, B,\ \cdot)$ is an optimal solution to $Q^\mu(w, B)$ almost surely w.r.t. $(w,B) \sim G$, then $\beta^\mu$ would be an SFPE.
\end{remark}

For $\mu: \Theta \to \mathbb{R}_{\geq 0}$ and $(w,B) \in \Theta$, consider the Lagrangian optimization problem of $Q^\mu(w, B)$ in which we move the budget constraint to the objective using the Lagrange multiplier $t \geq 0$. We use $t$ to denote the multiplier of one buyer in isolation to distinguish from $\mu$, which is a function giving a multiplier for every buyer type. Denoting by $q^\mu(w,B,t)$ the dual function, we have that
\begin{align*}
	 q^\mu(w,B,t) &= \max_{b(\cdot)} \mathbb{E}_{\alpha, \{\theta_i\}_{i=1}^{n-1}}\left[ (w^T \alpha - (1+t)b(\alpha))\ \mathds{1}\left\{b(\alpha) \geq \max\left( r(\alpha), \left\{ \beta^\mu(\theta_i, \alpha) \right\}_i \right) \right\}\right] + t B\\
	 &= (1 + t) \max_{b(\cdot)} \mathbb{E}_{\alpha,\{\theta_i\}_{i=1}^{n-1}}\left[ \left( \frac{w^T \alpha}{1+t} - b(\alpha)\right)\ \mathds{1}\left\{b(\alpha) \geq \max\left( r(\alpha), \left\{ \beta^\mu(\theta_i, \alpha) \right\}_i \right) \right\}\right] + tB.
\end{align*}
The dual problem of $Q^\mu(w, B)$ is given by $\min_{t\ge0} q^\mu(w, B, t)$.

The next lemma states that the optimal solution to the Lagrangian optimization problem is a value-pacing-based strategy. More specifically,  for every pacing function $\mu: \Theta \to  \R_{\geq 0}$, buyer type $(w,B)$ and dual multiplier $t$, the value pacing based strategy $\sigma_\alpha^\mu \left( w^T \alpha / (1+t)\right)$ is an optimal solution to the Langrangian relaxation of $Q^\mu(w, B)$ corresponding to multiplier $t$. Note that, in general, $t$ need not be equal to $\mu(w,B)$.

\begin{lemma}\label{lagrangian_optimal}
	For pacing function $\mu: \Theta \to \mathbb{R}_{\geq 0}$, buyer type $(w,B) \in \Theta$ and dual multiplier $t \geq 0$,
	\begin{align*}
        \sigma_\alpha^\mu \left(\frac{w^T \alpha}{1+t}\right) \in \argmax_{b(\cdot)}\ \mathbb{E}_{\alpha,\{\theta_i\}_{i=1}^{n-1}}\left[ \left( \frac{w^T \alpha}{1+t} - b(\alpha)\right)\ \mathds{1}\left\{b(\alpha) \geq \max\left( r(\alpha), \left\{ \beta^\mu(\theta_i, \alpha) \right\}_i \right) \right\}\right].
    \end{align*}	
\end{lemma}

In the proof of Lemma \ref{lagrangian_optimal}, we actually show something stronger than the statement of Lemma \ref{lagrangian_optimal}: the value-pacing-based strategy is optimal point-wise for each $\alpha$ and not just in expectation over $\alpha$. This follows from the observation that once we fix an item $\alpha$, we are solving the best response optimization problem faced by a buyer with value $w^T \alpha / (1+t)$ in the standard i.i.d.~setting~\citep{krishna2009auction} with competing buyer values being drawn from $\lambda_\alpha^\mu$ and under the assumption that the competing buyers use the strategy $\sigma_\alpha^\mu$. If $\lambda_\alpha^\mu$ had a strictly positive density, then the optimality of $\sigma_\alpha^\mu \left(w^T \alpha / (1+t)\right)$ would be a direct consequence of the definition of a symmetric BNE in the standard i.i.d.~setting. Even though the standard results cannot be used directly because of the potential absence of a density in the situation outlined above, we show that it is possible to adapt the techniques used in the proof of Proposition 2.2 of \cite{krishna2009auction} to show Lemma~\ref{lagrangian_optimal}.

Using Lemma \ref{lagrangian_optimal}, we can simplify the expression for the dual function $q^\mu(w,B,t)$. First, note that because $\sigma_\alpha^\mu$ is non-decreasing the highest competing bid can be written as
\[
    \max_{i=1,\ldots,n-1} \left\{ \beta^\mu(\theta_i, \alpha) \right\}
    = \max_{i=1,\ldots,n-1} \left\{ \sigma_\alpha^\mu\left(\frac{w_i^T\alpha}{1+\mu(\theta_i)} \right) \right\}
    = \sigma_\alpha^\mu\left(Y \right)\,,
\]
where $Y\sim H^\mu_\alpha$ is the maximum of $n-1$ paced values. Therefore, using that $\sigma_\alpha^\mu\left(w^T \alpha / (1+t) \right)$ is an optimal bidding strategy we get that
\begin{align*}
    q^\mu(w, B, t)
    &= (1+t)\ \mathbb{E}_{\alpha} \mathbb{E}_{Y\sim H^\mu_\alpha}\left[ \left(\frac{w^T \alpha}{1+t} - \sigma_\alpha^\mu\left(\frac{w^T \alpha}{1+t}\right)\right)\ \mathds{1}\left\{\sigma_\alpha^\mu\left(\frac{w^T \alpha}{1+t}\right) \geq \max \left( r(\alpha), \sigma_\alpha^\mu(Y) \right) \right\} \right] + t B\\
    &= (1+t)\ \mathbb{E}_{\alpha} \mathbb{E}_{Y\sim H^\mu_\alpha}\left[ \left(\frac{w^T \alpha}{1+t} - \sigma_\alpha^\mu\left(\frac{w^T \alpha}{1+t}\right)\right)\ \mathds{1}\left\{\frac{w^T \alpha}{1+t} \geq \max\left( r(\alpha), Y\right) \right\} \right] + t B\\
    &= (1+t)\ \mathbb{E}_{\alpha}\left[ \left(\frac{w^T \alpha}{1+t} - \sigma_\alpha^\mu\left(\frac{w^T \alpha}{1+t}\right)\right)\ H^\mu_\alpha\left(\frac{w^T \alpha}{1+t}\right) \mathds{1}\left\{\frac{w^T \alpha}{1+t} \geq r(\alpha) \right\} \right] + t B\\
    &= (1+t)\ \mathbb{E}_\alpha\left[\mathds{1}\left\{\frac{w^T \alpha}{1+t} \geq r(\alpha) \right\} \int_{r(\alpha)}^{\frac{w^T\alpha}{1+t}} H^\mu_\alpha(s) ds\right] + tB\,,
\end{align*}
where the second equation follows from part (c) of Lemma \ref{pacing_strategy_properties}, the third from taking expectations with respect to $Y$, and the last from our formula for $\sigma_\alpha^\mu$.

We now present the main result of this subsection, which characterizes the optimal solution of $Q^\mu(w, B)$ in terms of the optimal solution of the dual problem. The idea of using value-pacing-based strategies as candidates for the equilibrium strategy owes its motivation to Proposition~\ref{optimal_solution}.
It establishes that if all the other buyers are using a value-pacing-based strategy, with some pacing function $\mu: \Theta \to \R_{\geq 0}$, then a value-pacing-based strategy is a best response for a given buyer $(w,B)$.

\begin{proposition} \label{optimal_solution}
     There exists $\Theta' \subset \Theta$ such that $G(\Theta') = 1$ and for all  pacing functions $\mu: \Theta \to \R_{\geq 0}$ and buyer types $(w,B) \in \Theta'$, if $t^*$ is an optimal solution to the dual problem, i.e., if $t^* \in \argmin_{t^* \geq 0} q^\mu(w, B, t)$, then $\sigma_\alpha^\mu\left(w^T \alpha / (1+t^*)\right)$ is an optimal solution for the optimization problem $Q^\mu( w,B)$.
\end{proposition}

In Proposition~\ref{optimal_solution}, the pacing parameter $t^*$ used for pacing in the best response can, in general, be different from $\mu(w,B)$. This caveat requires a fixed-point argument to resolve, which will be the subject matter of the next subsection.
\vspace{-0.5em}
\begin{remark}\label{remark_theta'}
	Restricting to the measure-one set $\Theta'$ is without loss. Recall that according to Definition~\ref{definition_equilibrium}, a strategy constitutes a SFPE if, almost surely over $(w,B) \sim G$, using $\beta^*$ is an optimal solution to their optimization problem when all other buyer types also use it. As a consequence of this definition, we will show that it suffices to show strong duality for a subset of buyer types $\Theta' \subset \Theta$ such that $G(\Theta') = 1$. In the absence of reserve prices $r(\alpha)$ for the items, Proposition~\ref{optimal_solution} holds for all $(w,B) \in \Theta$. Reserve prices introduce some discontinuities in the utility and payment term. The subset $\Theta' \subset \Theta$ captures a collection of buyer types for which these discontinuities are inconsequential, while maintaining $G(\Theta') =1$.
\end{remark}

Observe that $Q^\mu(w, B)$ is not a convex optimization problem, so in order to prove the above theorem, we cannot appeal to the well-known strong duality results established for convex optimization. Instead, we will use Theorem 5.1.5 of \citet{bertsekas1998nonlinear}, which states that, to prove optimality of $\sigma_\alpha^\mu\left(w^T \alpha/(1+t^*)\right)$ for $Q^\mu(w, B)$, it suffices to show primal feasibility of $\sigma_\alpha^\mu\left(w^T \alpha/(1+t^*)\right)$, dual feasibility of $t^*$, Lagrange optimality of $\sigma_\alpha^\mu\left(w^T \alpha/(1+t^*)\right)$ for multiplier $t^*$, and complementary slackness. Our approach will be to show these required properties by combining the differentiability of the dual function with first order optimality conditions for one dimensional optimization problems. The key observation here is that the derivative of the dual function is equal to the difference between the budget of the buyer and her expected expenditure. Therefore, at optimality, the first-order conditions of the dual problem imply feasibility of the value-based pacing strategy. To prove differentiability we leverage that in our game the distribution of competing bids is absolutely continuous, which is critical for our results to hold.

For $t^* \in \argmin_{t \geq 0} q^\mu(w, B, t)$, if we apply the first-order optimality conditions for an optimization problem with a differentiable objective function over the domain $[0, \infty)$, we get
\begin{align*}
	\frac{\partial q^\mu(w, B, t^*)}{\partial t} \geq 0, \qquad t^* \geq 0, \qquad t^*\cdot \frac{\partial q^\mu(w, B, t^*)}{\partial t} = 0	\,.
\end{align*}
The first condition can be shown to imply primal feasibility, the second implies dual feasibility, and the third implies complementary slackness. Combining this with Lemma~\ref{lagrangian_optimal}, which establishes Lagrange optimality, and applying Theorem 5.1.5 of \cite{bertsekas1998nonlinear} yields Proposition~\ref{optimal_solution}. The complete proof of Proposition~\ref{optimal_solution} can be found in Appendix~\ref{appendix_proof_existence}.

\subsection{Fixed Point Argument}\label{sec:fixed-point}

In light of Proposition~\ref{optimal_solution}, the proof of Theorem~\ref{main_existence_result} (the existence of a value-pacing-based SFPE) boils down to showing that there exists a pacing function $\mu: \Theta \to \R_{\geq 0}$ such that, almost surely w.r.t. $(w,B) \sim G$, $\mu(w,B)$ is an optimal solution to the dual optimization problem $\min_{t \geq 0} q^\mu(w, B, t)$. In other words, given that everybody else acts according to $\mu$, a buyer $(w,B)$ that wishes to minimize the dual function is best off acting according to $\mu$.
More specifically, in Proposition~\ref{optimal_solution} we showed that, starting from a pacing function $\mu:\Theta \to \R_{\geq 0}$, if $\mu^*(w,B)$ constitutes an optimal solution to the dual problem $\min_{t \geq 0} q^\mu(w, B, t)$ almost surely w.r.t. $(w,B) \sim G$, then $\sigma_\alpha^\mu \left(w^T \alpha / (1+\mu^*(w,B))\right)$ is an optimal solution for the optimization problem $Q^\mu(w,B)$ almost surely w.r.t. $(w,B) \sim G$.
In other words, bidding according to $\sigma_\alpha^\mu$ while pacing according to $\mu^*:\Theta \to \R_{\geq 0}$ is a utility-maximizing strategy for buyer $(w,B) \sim G$ almost surely, given that other buyers bid according to $\sigma_\alpha^\mu$ with paced values obtained from $\mu$.
The following theorem establishes the existence of a pacing function $\mu: \Theta \to \R_{\geq 0}$ for which $\mu$ itself fills the role of $\mu^*$ in the previous statement, thereby implying the optimality of $\sigma_\alpha^\mu \left(w^T \alpha/(1+ \mu(w,B))\right)$  almost surely w.r.t. $(w,B) \sim G$.

\begin{proposition} \label{main_fixed_point}
    There exists $\mu: \Theta \to \R_{\geq 0}$ such that $\mu(w,B) \in \argmin_{t \geq 0} q^\mu( w, B, t)$ almost surely w.r.t. $(w,B) \sim G$.
\end{proposition}

We prove the above statement using an infinite-dimensional fixed-point argument on the space of pacing functions with a carefully chosen topology. Informally, we need to show that the correspondence that maps a pacing function $\mu: \Theta \to \R_{\geq 0}$ to the set of dual-optimal pacing functions, $\mu^*: \Theta \to \R_{\geq 0}$ which satisfy $\mu^*(w,B) \in \argmin_{t \geq 0} q^\mu(w, B, t)$, has a fixed point. However, unlike finite-dimensional fixed-point arguments, establishing the sufficient conditions of convexity and compactness needed to apply infinite-dimensional fixed point theorems requires a careful topological argument.

Lemma \ref{compact_dual_space} in the appendix shows that all dual optimal functions $\mu^*: \Theta \to \R_{\geq 0}$ map to a range that is a subset of $[0, \omega/B_{\min}]$. Therefore, any pacing function $\mu: \Theta \to \R_{\geq 0}$ that is a fixed point, i.e., satisfies $\mu(w,B) \in \argmin_{t \geq 0} q^\mu(w, B, t)$ almost surely w.r.t. $(w,B) \sim G$, must also satisfy $\text{range}(\mu) \subset [0, \omega/ B_{\min}]$. Hence, it suffices to restrict our attention to pacing functions of the form $\mu: \Theta \to [0, \omega/ B_{\min}]$.

Consider the set of all potential pacing functions
\[
  \mathcal{X} = \{\mu \in L_1\left(\Theta \right) \mid \mu(w,B) \in [0, \omega/B_{\min}]\ \forall\ (w,B) \in \Theta\},
\]
where $L_1\left(\Theta \right)$ is the space of functions $f : \Theta \rightarrow \R$ with finite $L_1$ norm w.r.t. the Lebesgue measure. Here, by $L_1$ norm of $f$ w.r.t. the Lebesgue measure, we mean $\|f\|_{L_1} = \int_\Theta |f(\theta)| d\theta$. Our goal is to find a $\mu \in \X$ such that almost surely w.r.t $(w,B) \sim G$ we have
\[
\mu(w,B) \in \argmin_{t \in [0, \omega/ B_{\min}]} q^\mu(w, B, t).
\]
Dealing with infinitely many individual optimization problems $\min_{t \in [0, \omega/ B_{\min}]} q^\mu(w, B, t)$, one for each $(w,B)$, makes the analysis hard. To remedy this issue, we combine these optimization problems by defining the objective $f: \X \times \X \to \R$, for all $\mu, \hat{\mu} \in \X$, as follows
\begin{align*}
        f(\mu, \hat{\mu}) &\coloneqq \mathbb{E}_{(w,B)} [q^\mu(w, B, \hat{\mu}(w,B))].
\end{align*}
For a fixed $\mu \in \X$, we then get a single optimization problem $\min_{\hat{\mu} \in \X} f(\mu, \hat{\mu})$ over functions in $\X$, instead of one optimization problem for each of the infinitely-many buyer types $(w,B) \in \Theta$. Later, in Lemma~\ref{combined_to_individual}, we will show that any optimal solution to the combined optimization problem is also an optimal solution to the individual optimization problems almost surely w.r.t $(w,B) \sim G$. Thus, shifting our attention to the combined optimization problem is without any loss (because sub-optimality on zero-measure sets is tolerable).

 With $f$ as above, we proceed to define the correspondence that is used in our fixed-point argument. The \emph{optimal solution correspondence} $C^*: \mathcal{X} \rightrightarrows \mathcal{X}$ is given by $C^*(\mu) \coloneqq \arg \min_{\hat{\mu} \in \mathcal{X}} f(\mu, \hat{\mu})$ (which could be empty) for all $\mu \in \mathcal{X}$. In Lemma \ref{combined_to_individual}, we will show that the proof of Proposition~\ref{main_fixed_point} boils down to showing that $C^*$ has a fixed point, which will be our next step.

Our proof will culminate with an application of the Kakutani-Glicksberg-Fan theorem, on a suitable version of $C^*$, to show the existence of a fixed point. An application of this result (or any other infinite dimensional fixed point theorem) requires intricate topological considerations. In particular, we need to endow $\X$ with a topology that satisfies the following conditions:
\begin{itemize}
	\item[I.] $\X$ is compact, convex and $C^*(\mu)$ is a non-empty subset of $\X$ for all $\mu \in \X$.
	\item[II.] $C^*$ is a Kakutani map, i.e., it is upper hemicontinuous, and $C^*(\mu)$ is compact and convex for all $\mu \in \X$.
\end{itemize}

In the case of infinite dimensions, bounded sets in many spaces, such as the $L_p(\Omega)$ spaces, are not compact. In particular, $\X$ is not compact as a subset of $L_p(\Omega)$ for any $1 \leq p \leq \infty$. One possible way around it would be to consider the weak* topology on $\X \subset L_\infty(\Omega)$, in which bounded sets are compact. This choice runs into trouble because it is difficult to show the upper hemicontinuity of $C^*$ (property II) under the weak convergence notion of the weak* topology. Alternatively, one could impose structural properties and restrict to a subset of $\X$, such as the space of Lipschitz functions, in which both compactness and continuity can be established. The issue with this approach is that the correspondence operator may, in general, not preserve these properties, i.e., property I might not hold. For example, even if $\mu$ is Lipschitz, $C^*(\mu)$ might not contain any Lipschitz functions.

We would like to strike a delicate balance between properties I and II by picking a space in which we can establish compactness of $\X$ and upper hemicontinuity of $C^*$, while, at the same time, ensuring that $C^*(\mu)$ contains at least one element from this space. It turns out that the right space that works for our proof is the space of bounded variation. To motivate this topology on the space of pacing functions, we state some properties of the ``smallest'' dual optimal pacing function. For $\mu: \Theta \to [0, \omega/ B_{\min}]$, we define $\ell^\mu: \Theta \to [0, \omega/ B_{\min}]$ as
\begin{align*}
    \ell^\mu(w,B) \coloneqq \min \left\{s \in \argmin_{t \in [0, \omega/ B_{\min}]} q^\mu(w, B, t) \right\}
\end{align*}
for all $(w,B) \in \Theta$. The minimum always exists because $q^\mu(w, B, t)$ is continuous as a function of $t$ (see Corollary~\ref{continuity_dual} in the appendix for a proof) and the feasible set of the dual problem is compact.

We first show that $\ell^\mu$ varies nicely with $w$ and $B$ along individual components:

\begin{lemma}\label{monotonicity_pacing_function}
	For $\mu: \Theta \to [0, \omega/ B_{\min}]$, the following statements hold:
	\begin{enumerate}[topsep = 0pt]
		\item $\ell^\mu: \Theta \to [0, \omega/B_{\min}]$ is non-decreasing in each component of $w$.
		\item $\ell^\mu: \Theta \to [0, \omega/B_{\min}]$ is non-increasing as a function of $B$.
	\end{enumerate}
\end{lemma}

The proof applies results from comparative statics, which characterize the way the optimal solutions behave as a function of the parameters, to the family of optimization problems $\min_{t \in [0, \omega/ B_{\min}]} q^\mu(w, B, t)$ parameterized by $(w,B) \in \Theta$.

Now we wish to show bounded variation of $\ell^\mu$. It is a well-known fact that monotonic functions of one variable have finite total variation. Moreover, functions of bounded total variation also form the dual space of the space of continuous functions with the $L_{\infty}$ norm, which allows us to invoke the Banach-Alaoglu Theorem to establish compactness in the weak* topology. These results for single variable functions, although not directly applicable to the multivariable setting, act as a guide in choosing the appropriate topology for our setting.

Since pacing functions take as input several variables, we need to look at multivariable generalizations of total variation. To this end, we state one of the standard definitions (there are multiple equivalent ones) of total variation for functions of several variables (see section 5.1 of \citealt{evans2015measure}) and then follow it up by a lemma which gives a bound on the total variation of the component-wise monotonic function $\ell^\mu$.

\begin{definition} For an open subset $\Omega \subset \mathbb{R}^n$, the total variation of a function $u \in L_1(\Omega)$ is given by
	\begin{align*}
    V(u, \Omega) \coloneqq \sup \left\{ \int_{\Omega} u(\omega) \operatorname{div} \phi(\omega) d\omega\ \biggr\lvert\  \phi \in C_c^1(\Omega, \mathbb{R}^n), \|\phi\|_\infty \leq 1 \right\}
\end{align*}
where $C_c^1(\Omega, \R^n)$ is the space of continuously differentiable vector functions $\phi$ of compact support contained in $\Omega$ and $\operatorname{div} \phi = \sum_{i=1}^n \frac{\partial \phi_i}{\partial x_i}$ is the divergence of $\phi$.
\end{definition}

\begin{lemma}\label{topology_motivating_properties}
	For any pacing function $\mu: \Theta \to [0, \omega/B_{\min}]$, the following statements hold:
	\begin{enumerate}[topsep = 0pt]
		\item $\ell^\mu \in L_1(\Theta)$.
		\item $V(\ell^{\mu}, \Theta) \leq V_0$ where $V_0:= (d+1)U^{d+1} \omega/B_{\min}$ is a fixed constant.	
	\end{enumerate}
\end{lemma}

Motivated by the above lemma, we define the set of pacing functions that will allow us to use our fixed-point argument. Define $\mathcal{X}_0 = \{\mu \in \X \mid V(\mu, \Theta) \leq V_0 \}$
to be the subset of pacing functions with variation at most $V_0$. Note that $\ell^\mu \in \X_0$. Define $C_0^*: \mathcal{X}_0 \rightrightarrows \mathcal{X}_0$ as $C_0^*(\mu) \coloneqq \argmin_{\hat{\mu} \in \mathcal{X}_0} f(\mu, \hat{\mu})$ for all $\mu \in \X_0$. We now state the properties satisfied by $\X_0$ that make it compatible with the Kakutani-Fan-Glicksberg fixed-point theorem.

\begin{lemma}\label{topology_required_properties}
	The following statements hold:
	\begin{enumerate}[topsep = 0 cm]
		\item $\mathcal{X}_0$ is non-empty, compact and convex as a subset of $L_1(\Theta)$.
		\item $f: \X_0 \times \X_0 \to \R$ is continuous when $\X_0 \times \X_0$ is endowed with the product topology.
		\item $C_0^*: \mathcal{X}_0 \rightrightarrows \mathcal{X}_0$ is upper hemi-continuous with non-empty, convex and compact values.
	\end{enumerate}
\end{lemma}

Finally, with the above lemma in place, we can apply the Kakutani-Fan-Glicksberg theorem to establish the existence of a $\mu \in \X_0$ such that $\mu \in C^*_0(\X_0)$. The following lemma completes the proof of Proposition~\ref{main_fixed_point} by showing that the fixed point is also almost surely optimal for each type. It follows from the fact that for $\mu \in \X_0$ that satisfy $\mu \in C_0^*(\mu)$, we have $\ell^\mu \in C_0^*(\mu)$.

\begin{lemma} \label{combined_to_individual}
	If $\mu \in C_0^*(\mu) = \argmin_{\hat{\mu} \in \mathcal{X}_0} f(\mu, \hat{\mu})$, then $\mu(w,B)$ is almost surely optimal for each type, i.e., $\mu(w,B) \in \argmin_{t \in [0, \omega/ B_{\min}]} q^\mu(w, B, t)$ a.s. w.r.t. $(w,B) \sim G$.	
\end{lemma}

As mentioned earlier, Proposition~\ref{main_fixed_point}, combined with Proposition~\ref{optimal_solution}, implies Theorem~\ref{main_existence_result}.

\section{Standard Auctions and Revenue Equivalence}\label{sec:standard}

In this section, we move beyond first-price auctions and generalize our results to anonymous standard auctions with reserve prices. An auction $\A = (Q, M)$, with allocation rule $Q: \R_{\geq 0}^n \to [0,1]^n$, payment rule $M: \R_{\geq 0}^n \to \R_{\geq 0}^n$ and reserve price $r$, is called an \emph{anonymous standard auction} if the following conditions are satisfied:
\begin{itemize}
	\item \emph{Highest bidder wins.} When the buyers bid $(b_1, \dots, b_n)$, the allocation received by buyer $i$ is given by $Q_i(b_1, \dots, b_n) = \mathds{1}(b_i \geq r, b_i \geq b_j\ \forall j \in [n])$, for all $i \in [n]$.
	\item \emph{Anonymity.} The payments made by a buyer do not depend on the identity of the buyer. More formally, if the buyers bid $(b_1, \dots, b_n)$, then for any permutation $\pi$ of $[n]$ and buyer $i \in [n]$, we have $M_i(b_1, \dots, b_n) = M_{\pi(i)}(b_{\pi(1)}, \dots, b_{\pi(n)})$, i.e., the payment made by the $i$th buyer before the bids are permuted equals the payment made by the bidder $\pi(i)$ after the bids have been permuted.
\end{itemize}

As in our definition of SFPE,  we are using an infeasible tie-breaking rule which allocates the entire good to every highest bidder. As with SFPE, ties are a zero-probability event under our value-pacing-based equilibria, and our results hold for arbitrary tie-breaking rules.

For consistency of notation, we will modify the above notation slightly to better match the one used in previous sections. Exploiting the anonymity of auction $\A$, we will denote the payment made by a buyer who bids $b$, when the other $n-1$ buyers bid $\{b_i\}_{i=1}^{n-1}$, by $M\left(b, \{b_i\}_{i=1}^{n-1}\right)$, i.e., we use the first argument for the bid of the buyer under consideration and the other arguments for the competitors' bids. Also, as the reserve price completely determines the allocation rule of a standard auction, in the rest of the section, we will omit the allocation rule while discussing anonymous standard auctions and represent them as a tuple $\A = (r,M)$ of reserve price and payment rule.

To avoid delving into the inner workings of the auction, we assume the existence of an \emph{oracle} that takes as an input an atomless distribution $\mathcal{H}$ over $[0,\omega]$ and outputs a bidding strategy $\psi^{\mathcal{H}}:[0,\omega] \to \mathbb{R}$ satisfying the following properties:
\begin{enumerate}
\item The strategy $\psi^{\mathcal{H}}$ is a single-auction equilibrium for the auction $\mathcal{A}$ when the values are drawn i.i.d.~from $\mathcal{H}$, i.e., $\psi^{\mathcal{H}}(x) \in \argmax_{b \ge 0} \mathbb{E}_{X_i \sim \mathcal H }\big[ x \ \mathds{1}\{b \geq \max(r, \{ \psi^{\mathcal{H}}(X_i)\}_i) \} - M\left(b, \{ \psi^{\mathcal{H}}(X_i)\}_i \right) \big]$.

\item The strategy $\psi^{\mathcal{H}}(x)$ is non-decreasing in $x$, and $\psi^{\mathcal{H}}(x) \geq r$ if and only if $x \geq r$.
\item The payoff for a bidder who has zero value for the object is zero at the single-auction equilibrium.
\item The distribution of $\psi^{\mathcal{H}}(x)$, when $x \sim \mathcal{H}$, is atomless.
\end{enumerate}
Our results will produce a pacing-based equilibrium bidding strategy for budget-constrained buyers by invoking $\psi^{\mathcal{H}}$ as a black box. To make the discussion more concrete, let $\A$ to be a second-price auction with reserve price $r$. For a given atomless distribution $\mathcal{H}$, define $\psi^{\mathcal{H}}(v) = v$ to be the truthful bidding strategy. Then, $\psi^{\mathcal{H}}$ is a single-auction equilibrium because bidding truthfully is a dominant strategy in second-price auctions. Moreover, $\psi^{\mathcal{H}}$ is non-decreasing, $\psi^{\mathcal{H}}(x) \geq r$ if and only if $x \geq r$, a bidder with zero value bids zero to attain a payoff of zero, and finally the distribution of $\psi^{\mathcal{H}}(x)$ when $x \sim \mathcal{H}$ is simply $\mathcal{H}$, which is atomless. Thus, second-price auctions with reserve prices satisfy the above assumptions.

In our analysis, we allow the seller to condition on the feature vector and choose a different mechanism for each context $\alpha \in A$. Let $\{\A_\alpha = (r(\alpha), M_\alpha)\}_{\alpha \in A}$ be a family of anonymous standard auctions such that $\alpha \mapsto r(\alpha)$ is measurable. Moreover, suppose that for any measurable bidding function $\alpha \mapsto b(\alpha)$ and any collection of measurable competing bidding functions $\alpha \mapsto b_i(\alpha)$ for $i \in [n-1]$, the payment function $\alpha \mapsto M_\alpha \left(b(\alpha),\{b_i(\alpha)\}_{i=1}^{n-1}\right)$ is also measurable. Below, we define the equilibrium notion for the family $\{\A_\alpha\}_{\alpha \in A}$ of anonymous standard auctions.

\begin{definition} A strategy $\beta^*:\Theta\times A \to \mathbb{R}$ is called a \emph{Symmetric Equilibrium} for the family of standard auctions $\{\A_\alpha\}_{\alpha \in A}$, if $\beta^*(w, B, \alpha)$ (as a function of $\alpha$) is an optimal solution to the following optimization problem almost surely w.r.t. $(w,B) \sim G$.
\begin{align*}
    \max_{b:A \to \mathbb{R}_{\geq 0}} \quad &\mathbb{E}_{\alpha,\{\theta_i\}_{i=1}^{n-1}}\left[ w^T \alpha \ \mathds{1}\{b(\alpha) \geq \max(r(\alpha), \{\beta^*(\theta_i, \alpha)\}_i) \} - M_\alpha\left(b(\alpha), \{\beta^*(w_i,B_i,\alpha)\}_i \right) \right]\\
    \textrm{s.t.} \quad &\mathbb{E}_{\alpha,\{\theta_i\}_{i=1}^{n-1}}\left[M_\alpha\left(b(\alpha), \{\beta^*(w_i,B_i, \alpha)\}_i\right)  \right] \leq B\,.
\end{align*}
\end{definition}

Observe that the above definition reduces to Definition \ref{definition_equilibrium} if we take $\{\A_\alpha\}_{\alpha \in A}$ to be the set of first-price auctions with reserve price $r(\alpha)$. Next, we show that the equilibrium existence and characterization results of the previous sections apply to all standard auctions that satisfy the required assumptions. To do this, we first need to define value-pacing strategies for anonymous standard auctions. These are a natural generalization of the value-pacing-based strategies used for first-price auctions.

Recall that, for a pacing function $\mu: \Theta \to \R_{\geq 0}$ and $\alpha \in A$, $\lambda_\alpha^\mu$ denotes the distribution of paced values for item $\alpha$, and $H_\alpha^\mu$ denotes the distribution of the highest value for $\alpha$, among $n-1$ buyers. For ease of notation, we will use $\psi_\alpha^\mu$ to denote the single-auction equilibrium strategy for auction $\mathcal A_\alpha$ when values are drawn from $\mathcal H = \lambda^\mu_\alpha$ or more formally $\psi_\alpha^\mu:=\psi_\alpha^{\lambda^\mu_\alpha}$. For a pacing function $\mu: \Theta \to \R_{\geq 0}$, $(w,B) \in \Theta$ and $\alpha \in A$, define
\begin{align}\label{eqn:std_auction-eq_strat}
    \Psi^\mu(w, B,\alpha) \coloneqq \psi_\alpha^{\mu}\left(\frac{w^T\alpha}{1+\mu(w,B)}\right)\,,
\end{align}
to be our candidate equilibrium strategy. This strategy is well-defined because, by Lemma \ref{pacing_strategy_properties}, $\lambda_\alpha^\mu$ is atom-less almost surely w.r.t.~$\alpha$. As before, the bid $\Psi^\mu(w, B,\alpha)$ is the amount a non-budget-constrained buyer with type $(w, B)$ would bid on item $\alpha$ if her paced value was her true value, when competitors are pacing their values accordingly. In other words, bidders in the proposed equilibrium first pace their values, and then bid according to the single-auction equilibrium of auction $\mathcal A_\alpha$ in which competitors' values are also paced.

With the definition of value-pacing-based strategies in place, we can now state the main result of this section. Recall that, $C_0^*: \mathcal{X}_0 \rightrightarrows \mathcal{X}_0$ is given by $C_0^*(\mu) \coloneqq \arg \min_{\hat{\mu} \in \mathcal{X}_0} f(\mu, \hat{\mu})$ for all $\mu \in \X_0$, where $f$ is the expected dual function in the case of a first-price auction, as defined in Section~\ref{sec:fixed-point}.

\begin{theorem}[Revenue and Pacing Equivalence]\label{standard_auctions_revenue_equivalence}
	For any pacing function $\mu \in \X_0$ such that $\mu \in C^*_0(\mu)$ is an equilibrium pacing function for first-price auctions, the value-pacing-based strategy $\Psi^\mu: \Theta \times A \to \R_{\geq 0}$ is a Symmetric Equilibrium for the family of auctions $\{\A_\alpha\}_{\alpha \in A}$. Moreover, the expected payment made by buyer $\theta$ under this equilibrium strategy is equal to the expected payment made by buyer $\theta$ in first-price auctions under the equilibrium strategy $\beta^\mu: \Theta \times A \to \R_{\geq 0}$, i.e.,
	\begin{align*}	&\mathbb{E}_{\alpha,\{(\theta_i)\}_{i=1}^{n-1}}\left[M_\alpha\left(\Psi^\mu(\theta,\alpha), \{\Psi^\mu(\theta_i, \alpha)\}_i\right) \right]\\
		= & \mathbb{E}_{\alpha,\{(\theta_i)\}_{i=1}^{n-1}}\left[ \beta^\mu(\theta, \alpha)\ \mathds{1}\{\beta^\mu(\theta, \alpha) \geq \max(r(\alpha), \{\beta^\mu(\theta_i, \alpha)\}_i)\}\right]
	\end{align*}
\end{theorem}

The key step in the proof involves showing that the dual of the budget-constrained utility-optimization problem faced by a buyer is identical for all standard auctions, when the other buyers use the equilibrium strategy $\Psi^\mu$ of the standard auction under consideration. To establish this key step, we exploit the separable structure of the Lagrangian optimization problem and apply the known utility equivalence result for standard auctions in the single-auction i.i.d. setting, once for each item $\alpha \in A$. Then, we establish the analogue of Proposition~\ref{optimal_solution} for standard auctions. Combining this with $\mu \in C_0^*(\mu)$ yields Theorem~\ref{standard_auctions_revenue_equivalence}.

Our revenue equivalence relies on three critical assumptions: risk-neutrality, independence of weight vectors, and symmetry. As in the classical setting, revenue equivalence would fail if buyers are risk averse (see, e.g., \citealt{krishna2009auction}). We emphasize that, in contrast to the classical revenue equivalence result, buyers' values $w^T \alpha$ are not independent. Our result does require that weight vectors are independent across buyers. Buyers in our model are ex-ante homogeneous since buyer types are drawn from the same population. We remark, however, that buyers are heterogenous {\newsanti in the interim sense}: the buyers competing in an auction can have different budgets and weight vectors. Revenue equivalence would fail is buyers are ex-ante heterogenous, i.e., if competitors are drawn from different populations.

Before ending this section, we state some important implications of Theorem~\ref{standard_auctions_revenue_equivalence}. If the pacing function $\mu$ allows the buyers to satisfy their budget constraints in some standard auction, then the same pacing function $\mu$ allows the buyers to satisfy their budgets in every other standard auction. In other words, the equilibrium pacing functions are the same for all standard auctions. This means that in order to calculate an equilibrium pacing function $\mu$ that satisfies $\mu \in C_0^*(\mu)$, it suffices to compute it for any standard auction (in particular, one could consider a second-price auction for which bidding truthfully is a dominant-strategy equilibrium in the absence of budget constraints). This fact is especially pertinent in view of the recent shift in auction format used for selling display ads from second-price auctions to first-price auctions, because it states that, in equilibrium, the buyers can use the same pacing function even after the change.  Moreover, the same pacing function continues to work even if the family $\{\A_\alpha\}_{\alpha \in A}$ is an arbitrary collection of first-price and second-price auctions (or any other combination of standard auctions), i.e.,  Theorem~\ref{standard_auctions_revenue_equivalence} states that, not only can one pacing function be used to manage budgets in first-price and second-price auctions, the same pacing function also works in the intermediate transitions stages, in which buyers may potentially participate in some mixture of these auctions.

Another important takeaway is that all standard auctions with the same allocation rule yield the same revenue to the seller. We remark, however, that the revenue of the seller does depend on the allocation, and the seller could thus maximize her revenue by optimizing over the reserve prices. We leave the question of optimizing the auction design as a future research direction.

The revenue-equivalence in the presence of in-expectation budget constraints is driven by the invariance of the pacing function over all standard auctions and the classical revenue equivalence result for the unconstrained i.i.d.~setting, which shows that---on average---payments are the same across standard auctions. While revenue equivalence is known to hold for standard auctions without budget constraints, \cite{che1998standard} showed that, when budget constraints are hard, first-price auctions lead to higher revenue than second-price auctions. The intuition for their result is that because bids are higher in second-price auctions than first-price auctions, hard budget constraints are more likely to bind in the former, which reduces the seller's revenue. Surprisingly, Theorem~\ref{standard_auctions_revenue_equivalence} shows that when budgets constraints are in expectation (and values are feature-based), we recover revenue equivalence.
To better understand the difference between the two types of constraints, consider the following example:

\begin{example*}
Consider two buyers with values drawn uniformly from the unit interval $[0,1]$. Moreover, let the budget of the buyer with value $v$ be given by $0.5 + \epsilon v$ for some small $\epsilon > 0$. First, observe that, in the absence of budget constraints, bidding truthfully is a dominant strategy in a second-price auction and bidding half of one's value is a Bayes-Nash equilibrium in a first-price auction. Moreover, from the standard revenue-equivalence result, a buyer with value $x$ spends $x^2/2$ in expectation over the other buyer's type in both auctions. Now, since this expected expenditure is less than $1/2$ for all types, the in-expectation budget constraints are non-binding and the equilibria remain unchanged even when in-expectation budget constraints are imposed. On the other hand, consider the case when the budget constraints are hard. The first-price auction equilibrium remains unchanged because every buyer type bids less than $0.5$, so the constraint is always satisfied. But, for second-price auction, this is not the case: With hard budget constraints, the equilibrium strategy for the buyers is to bid the minimum of their value and budget, thereby leading to lower revenue compared to the truthful-bidding equilibrium.
\end{example*}

\rk{We conclude this section with a discussion of extensions and alternative models. Firstly, even though we only consider anonymous standard auctions in this work, our equilibrium existence and revenue equivalence results can be extended to other anonymous allocation rules $Q$ which (i) admit an oracle that outputs an equilibrium bidding strategy for traditional i.i.d. setting and satisfies properties (1)-(4) listed at the beginning of this section, (ii) lead to continuous non-decreasing interim-allocation rules for every buyer-item pair when other buyers follow a value-pacing-based strategy analogous to the one defined in equation~\eqref{eqn:std_auction-eq_strat}. Secondly, the argument developed in the section also implies the existence of value-pacing-based equilibria and revenue equivalence for standard auctions in the symmetric special case of the models studied in \citet{balseiro2015repeated} and \citet{balseiro2021budget}, which consider buyers with ex-ante budget constraints that hold in expectation over a buyer's own value and the values of others (see Appendix~\ref{appendix:ex-ante} for a detailed description).}

\section{Worst-Case Efficiency Guarantees}\label{sec:poa}

In this section, we use our framework to characterize the Price of Anarchy, i.e., the worst-case ratio of the efficiency of a pacing equilibrium relative to the efficiency of the best possible allocation. We measure efficiency of an allocation using the notion of \emph{liquid welfare} introduced by \citet{dobzinski2014efficiency}, which captures the maximum revenue that can be extracted by a seller who knows the values in advance. \rk{We use liquid welfare as a measure of efficiency instead of social welfare because the latter can have arbitrarily small Price of Anarchy (see Appendix~\ref{appendix:poa} for an example).} Throughout this section, we assume that the reserve price is zero for each item, i.e., $r(\alpha) = 0$ for all $\alpha \in A$.

We begin by defining the appropriate notion of liquid welfare of an allocation for our model motivated by the original definition of \citet{dobzinski2014efficiency}. Here, an allocation simply refers to a measurable function $x: A \times \Theta^n \to \Delta^n$, where $\Delta^n = \{y \in \R_+^n \mid \sum_{k=1}^n y_k = 1\}$ is the $n$-simplex, and $x_i(\alpha, \vec{\theta})$ denotes the fraction of the item $\alpha$ allocated to buyer $i$ when the buyer types are given by the profile $\vec{\theta} = (\theta_1, \dots, \theta_n)$. {\newsanti In our setting, the liquid welfare of a buyer is equal to the minimum of the value obtained by the buyer from the allocation and her budget.}

\begin{definition}
	For an allocation $x: A \times \Theta^n \to \Delta^n$,  we define its liquid welfare as
	\begin{align*}
			\operatorname{LW}(x) = \sum_{i=1}^n \E_{\theta_i} \left[ \min\left\{\E_{\alpha, \theta_{-i}}[w_i^T\alpha \cdot x_i(\alpha, \theta_i, \theta_{-i})], B_i \right\} \right]\,.
	\end{align*}
\end{definition}

Next, we define Price of Anarchy with respect to liquid welfare for pacing-based equilibria. Our definition is an instantiation of the general definition of Price of Anarchy introduced in \citet{koutsoupias1999worst}. Before proceeding with the definition, it is worth noting an important consequence of our revenue equivalence result (Theorem~\ref{standard_auctions_revenue_equivalence}): Given an equilibrium pacing function $\mu$, i.e., a fixed point of $C^*_0$, the allocation under the equilibrium parameterized by $\mu$ is the same for all standard auctions. Thus, the equilibrium allocation is determined by the pacing function and is independent of the pricing rule of the standard auction, which is reflected in the following definition. For an equilibrium pacing function $\mu$, we use $x^{\mu}$ to denote the allocation under the equilibrium parameterized by $\mu$; again, this allocation is the same for all standard auctions without reserve prices.

\begin{definition}
	The Price of Anarchy (PoA) of pacing-based equilibria (for all standard auctions) is defined as the ratio of the worst-case liquid welfare across all pacing equilibria, and the optimal liquid welfare
\begin{align*}
		\operatorname{PoA} = \frac{ \inf_{\mu: \mu \in C^*_0(\mu)} \operatorname{LW}(x^{\mu})}{\sup_x \operatorname{LW}(x)}
	\end{align*}	where the supremum in the denominator is taken over all measurable allocations $x$.
\end{definition}

Since the PoA of pacing-based equilibria does not depend on the payment rule, we can work with the most convenient standard auction to prove a lower bound on the PoA, which in this case happens to be the second-price auction. \citet{azar2017liquid} study the PoA of pure-strategy Nash equilibria of second-price auctions in a non-Bayesian multi-item setting with budgets, and provide a lower bound of 1/2 for it. Unfortunately, their result hinges on the ``no over-budgeting" assumption that requires the sum of equilibrium bids to be bounded above by the budget, which need not hold for pacing-based equilibria, thereby necessitating new proof ideas. Moreover, their bound may be vacuous for some parameter values because a pure-strategy Nash equilibrium is not guaranteed to exist in their setting. To get around this, they study mixed-strategy and Bayes-Nash equilibria, and bound their PoA, but the lower bound they obtain for these equilibria is much worse (less than $0.02$). Our model does not suffer from the problem of existence: a pure-strategy pacing-based equilibrium is always guaranteed to exist (Theorem~\ref{main_existence_result}). This makes the following lower bound  on the PoA, which provides a worst-case guarantee of $1/2$, more appealing.

\begin{theorem}\label{thm:poa}
	The PoA of pacing-based equilibria of any standard auction is greater or equal to $1/2$.	
\end{theorem}

The proof, which is in Appendix~\ref{appendix:poa}, leverages the complementary slackness condition of pacing-based equilibria to bound the PoA. Interestingly, our proof does not use a hypothetical deviation to another bidding strategy, a technique commonly found in PoA bounds (see \citealt{roughgarden2017price} for a survey); and thus may be of independent interest.

\section{Structural Properties}\label{sec:structural}

In this section, we will show that pacing-based equilibria satisfy certain monotonicity and geometric properties related to the space of value vectors. It is worth noting that, in light of the revenue equivalence result of the preceding section, the properties established in this section hold for pacing equilibria of \emph{all standard auctions}. As in Section~\ref{sec:poa}, we will assume that the reserve price for each item is zero, i.e., $r(\alpha) = 0$ for all $\alpha \in A$. Without this assumption, similar results hold, but they become less intuitively appealing and harder to state.
Moreover, we will also assume that the support of $G$, denoted by $\delta(G)$, is a convex compact subset of $\R_+^{d+1}$.
This assumption is made to avoid having to specify conditions on the pacing multipliers of types with probability zero of occurring. Moreover, we consider a pacing function $\mu: \Theta \to [0,\omega/B_{\min}]$ such that $\mu(w,B)$ is the unique optimal solution for the dual minimization problem for each $(w,B)$ in the support of $G$, i.e.,  $\mu(w,B) = \argmin_{t \in [0, \omega/ B_{\min}]} q^\mu(w, B, t)$ for all $(w,B) \in \delta(G)$.
  We remark that we are assuming that the best response is unique rather than the equilibrium being unique. The former can be shown to hold under fairly general conditions.

  First, in Lemma~\ref{monotonicity_pacing_function} we showed that the pacing function associated with an SFPE is monotone in the buyer type. In particular, when the best response is unique, this result implies that $\mu(w,B)$ is non-decreasing in each component of the weight vector $w$ and non-increasing in the budget $B$. Intuitively, if the budget decreases, a buyer needs to shade bids more aggressively to meet her constraints. Alternatively, when the weight vector increases, the advertiser's paced values increase, which would result in more auctions won and higher payments. Therefore, to meet her constraints the advertiser would need to respond by shading bids more aggressively. Furthermore, when the best response is unique, it can also be shown that $\mu$ is continuous (see Lemma~\ref{buyer_type_continuity} in the appendix).

  The next theorem further elucidates the structure imposed on $\mu$ by virtue of it corresponding to the optima of the family of dual optimization problems parameterized by $(w,B)$. In what follows, we will refer to a buyer $(w,B)$ with $\mu(w,B) = 0$ as an \emph{unpaced buyer}, and call her a \emph{paced buyer} otherwise.

\begin{proposition} \label{structural_prop}
	Consider a unit vector $\hat{w}\in \R_+^d$ and budget $B > 0$ such that $w/\|w\| = \w$, for some $(w,B) \in \delta(G)$. Then, the following statements hold,
	\begin{enumerate}[topsep = 0cm]
  \item Paced buyers with budget $B$ and weight vectors lying along the same unit vector $\w$ have identical paced feature vectors in equilibrium.
    Specifically, if $(w_1,B), (w_2, B) \in \delta(G)$, with $w_1/\|w_1\| = w_2/\|w_2\| = \w$ and $\mu(w_1, B), \mu(w_2, B) > 0$, then $w_1/(1 + \mu(w_1,B))= w_2/(1 + \mu(w_2,B))$.

\item Suppose there exists an unpaced buyer $(w,B)\in \delta(G)$ with $w/ \|w\| = \hat w$ and $\mu(w, B) =0$. Let $w_0 = \argmax\{ \|w\| \, \mid w \in \R^d;\ \mu(w, B) = 0 \text{ and } w/ \|w\| = \hat w \}$ be the largest unpaced weight vector along the direction $\hat w$. Then, all paced weight vectors get paced down to $w_0$, i.e., $w/(1 + \mu(w,B)) = w_0$ for all $w \in \delta(G)$ with  $w/\|w\| = \w$ and $\mu(w,B)>0$.

	\end{enumerate}
\end{proposition}

In combination with complementary slackness, the first part states that, in equilibrium, buyers who have the same budget, have positive pacing multipliers, and have feature vectors which are scalar multiples of each other, get paced down to the same type at which they exactly spend their budget. In other words, scaling up the feature vector of a budget-constrained buyer, while keeping her budget the same, does not affect the equilibrium outcome. The second case of Proposition~\ref{structural_prop} addresses the directions of buyers that have a mixture of paced and unpaced buyers. In this case, there is a critical buyer type who exactly spends her budget when unpaced, and all buyer types that have weight vectors with larger norm (but the same budget) get paced down to this critical buyer type, i.e., their paced weight vector equals the critical buyer type's weight vector in equilibrium. The buyer types which have a smaller norm are unpaced.


Our non-atomic model also allows us to answer the following question: Keeping the competition fixed, how should an advertiser modify her targeting criteria or ad (as captured by the weight vector) in order to maximize her utility? This result is especially important for online display ad auctions, where the weight vector is estimated with the goal of predicting the click-through-rate (CTR) and advertisers routinely modify their ads to attract more clicks. The following theorem states that the gradient w.r.t.~the weight vector of the equilibrium utility of a buyer with type $(w,B)$ is given by the expected feature vector that she wins in equilibrium. This is because strong duality (Proposition~\ref{optimal_solution}) implies that the utility of every buyer type is given by the optimal dual value $q^\mu(w, B, \mu(w,B))$. From a practical perspective, an advertiser should focus on improving the weights of those features which have the largest average among the contexts won. It is worth noting that these quantities can be easily computed using data available to an advertiser.

\begin{proposition}\label{thm:targeting}
	Assume that $A$ is compact. Let $\mu: \Theta \to \R_{\geq 0}$ be an equilibrium pacing function, i.e., $\mu: \Theta \to \R_{\geq 0}$ such that $\mu(w,B) \in \argmin_{t \geq 0} q^\mu( w, B, t)$ almost surely w.r.t. $(w,B) \sim G$. Then, for all $(w, B) \in \Theta$, we have $\nabla_w q^\mu(w, B, \mu(w,B)) = \mathbb{E}_{\alpha, \{\theta_i\}_{i=1}^{n-1}}\left[ \alpha\ \mathds{1}\left\{ \beta^\mu((w,B), \alpha) \geq \beta^\mu(\theta_i, \alpha) _i\ \forall i \right\}\right]$.
\end{proposition}

\section{Analytical Example and Numerical Experiments}

In this section, we illustrate our theory by providing a stylized example in which we can determine the equilibrium bidding strategies in closed form, and then conduct some numerical experiments to verify our theoretical results. The purpose of the analytical example is to confirm our structural results and also help validate that our numerical procedures converge to an approximate version of the equilibrium strategies proposed in our paper.

\subsection{Analytical Example} \label{section_example}

We provide an instructive (albeit stylized) example with two-dimensional feature vectors to illustrate the structural property described in Section~\ref{sec:structural}. For $1 \leq a < b$, define the set of buyer types as (see the blue region in Figure~\ref{figure_experiment_example} for a visualization of this set)
\begin{align*}
	\Theta \coloneqq \left\{ (w,B) \in \R_{\geq 0}^2 \times \R_+ \biggr\lvert a \leq \|w\| \leq b,\ B = \frac{2\|w\| - w_1 - w_2}{\pi \|w\|} \right\}.
\end{align*}
In this example weight vectors lie in the intersection of a disk with the non-negative quadrant. Observe that all buyer types whose weight vectors are co-linear (i.e., they lie along the same unit vector) have identical budgets. Let the number of buyers in the auction be $n = 2$. Moreover, define the set of item types as the two standard basis vectors $A \coloneqq \{e_1, e_2\}$. Finally, let $G$ (distribution over buyer types) and $F$ (distribution over item types) be the uniform distribution on $\Theta$ and $A$ respectively. Since $A$ is discrete and $F$ does not have a density, this example does not satisfy the assumptions we made in our model. Nonetheless, in the next claim, we show that not only does a pacing equilibrium exists, but we can also state it in closed form. The proof of the claim can be found in Appendix~\ref{appendix_example_numerics}.

\begin{claim}\label{example}
	The pacing functions $\mu: \Theta \to \R$ defined as $\mu(w,B) = \|w\| - 1$, for all $(w,B) \in \Theta$, is an equilibrium, i.e., $\beta^\mu$, as given in Definition \ref{definition_equilibrium}, is a SFPE.
\end{claim}

Since $H_\alpha^\mu(\cdot)$ is a strictly increasing function for all $\alpha \in A$, it is easy to check that $\mu(w,B)$ is the unique optimal to the dual optimization problem $\min_{t \in [0, \omega/ B_{\min}]} q(\mu, w, B, t)$ for all $(w,B) \in \delta(G)$. Therefore, this example falls under the purview of part 1 of Proposition~\ref{structural_prop}. As expected, conforming to Proposition~\ref{structural_prop}, the buyers whose weight vectors are co-linear get paced down to the same point on the unit arc, as shown in Figure \ref{figure_experiment_example}.

%
%
%


\subsection{Numerical Experiments}

We now describe the simulation-based experiments we conducted to verify our theoretical results. As is necessitated by computer simulations, we studied a discretized version of our problem in these experiments. More precisely, in our experiments, we used discrete approximations to the buyer type distribution $G$ and the item type distribution $F$. Moreover, for all item types $\alpha$, we set the reserve price $r(\alpha) = 0$. One of the primary objectives of our simulations is to demonstrate that, despite the discretization, a buyer type can obtain her optimal bidding strategy by finding the optimal solution to the dual problem, as our theory suggests. In other words, to compute an equilibrium it suffices to best-respond in the dual space which has the advantage of being much simpler than the primal space.  To do so, for each discretized instance, we run best-response dynamics in the dual space by iterating over buyer types; computing each buyer type's optimal dual solution while keeping everyone else's pacing-based strategy fixed and then using this optimal dual solution to determine her pacing-based bidding strategy. This approach is not guaranteed to converge. In fact, due to the discretization, strong duality may fail to hold and a pure strategy equilibrium may not even exist. Nevertheless, despite the lack of theoretical guarantees, our experiments demonstrate that our analytical results and the dual best-response algorithm they inspire continue to work well in discrete settings.

As a first step, and to validate our best-response dynamics, we ran the algorithm on the discrete approximation of the example discussed in Subsection~\ref{section_example}, for which we had already analytically determined a pacing equilibrium in Claim \ref{example}.
The problem was discretized by picking 320 points lying in the set of buyer types $\Theta$ defined in Subsection~\ref{section_example}.
In Figure \ref{figure_experiment_example}, we provide plots for the case when $a = 2, b = 3$.
We see that the theoretical predictions from Claim~\ref{example} are replicated almost exactly by the solution computed by best-response iteration on the discretized problem. Moreover, co-linear buyer types converge to the same paced type vector, thereby validating Proposition~\ref{structural_prop}.

\usetikzlibrary{patterns}
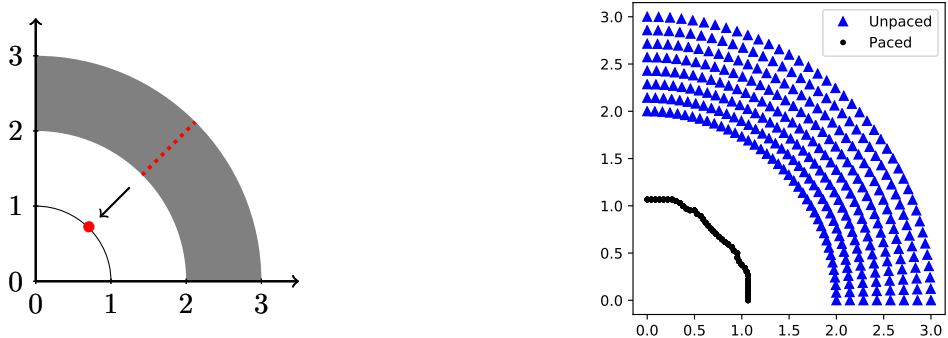
\begin{figure}[]
	\centering
	\begin{subfigure}{.49\textwidth}
  		\centering
  		\begin{tikzpicture}
			\draw [gray, line width = 1 cm] (2.5,0) arc (0:90:2.5cm);
			\draw[thick,->] (0,0) -- (3.5,0);
			\draw[thick,->] (0,0) -- (0,3.5);
			\foreach \x in {0,1,2,3}
		    \draw (\x cm,1pt) -- (\x cm,-1pt) node[anchor=north] {$\x$};
			\foreach \y in {0,1,2,3}
		    \draw (1pt,\y cm) -- (-1pt,\y cm) node[anchor=east] {$\y$};
		
		    \draw (1,0) arc (0:90:1cm);
			\draw[thick,->] (0,0) -- (3.5,0);
			\draw[thick,->] (0,0) -- (0,3.5);
			\foreach \x in {0,1,2,3}
		    \draw (\x cm,1pt) -- (\x cm,-1pt) node[anchor=north] {$\x$};
			\foreach \y in {0,1,2,3}
		    \draw (1pt,\y cm) -- (-1pt,\y cm) node[anchor=east] {$\y$};
		
		    \draw [red, dotted, line width = 0.05 cm] (1.414, 1.414) -- (2.1213, 2.1213);
		    \node [red] at (0.707, 0.707) {\textbullet};
		    \draw [thick, ->] (1.25, 1.25) -- (0.85, 0.85);
		\end{tikzpicture}
	\end{subfigure}
	\begin{subfigure}{.49\textwidth}
  		\centering
  		\includegraphics[trim=1.5cm 0cm 1.5cm 0.5cm,clip, width=.8\linewidth]{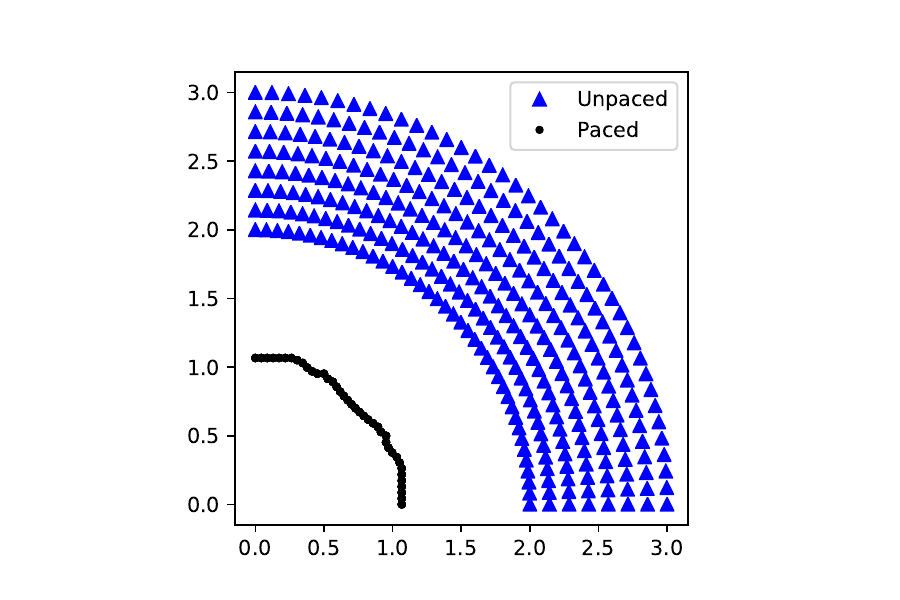}
	\end{subfigure}
	\caption{The example from Section~\ref{section_example} with $a = 2, b =3$. The unpaced and paced buyer weight vectors are uniformly distributed in the gray (triangle) and black (circle) region, respectively. Each plot shows the distribution of two-dimensional buyer weight vectors. The weight vectors before pacing are depicted in gray (triangles) and the paced weight vectors are depicted in black (circles). The left plot shows the theoretical results of Subsection~\ref{section_example}. In the left plot, the buyer weight vectors lying on the dotted line get paced down to the point. The right plot shows the results of best-response iteration on the corresponding discretized problem.}
	\label{figure_experiment_example}
\end{figure}

We conducted experiments to verify the structural properties described in Proposition~\ref{structural_prop}. Here we consider instances with $n = 3$ buyers per auction, $d = 2$ features, the buyer type distribution $G$ given by the uniform distribution on $(1,2) \times (1,2) \times \{0.6\}$ and the item type distribution $F$ given by the uniform distribution on the one-dimensional simplex $\{(x,y) \mid x,y \geq 0;\ x+y = 1\}$. These were discretized taking a uniform grid with 10 points along each dimension. The results are portrayed in Figure \ref{figure_simulation}.
\begin{figure}
	\centering
	\begin{subfigure}{.49\textwidth}
  		\centering
  		\includegraphics[trim=1cm 0.5cm 0.5cm 0.5cm,clip, width=.8\linewidth]{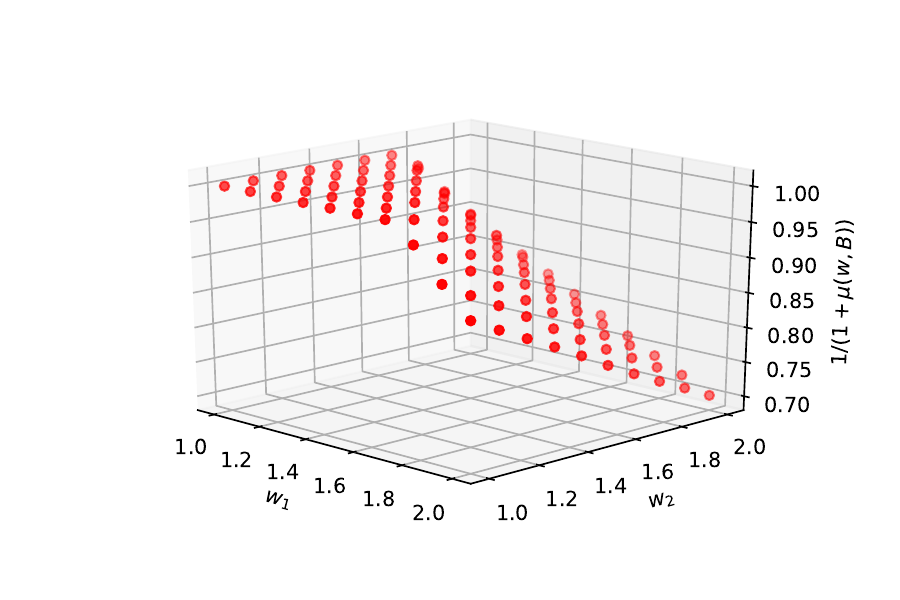}
	\end{subfigure}
	\begin{subfigure}{.49\textwidth}
  		\centering
  		\includegraphics[trim=0.5cm 0.2cm 1cm 0.5cm,clip, width=.8\linewidth]{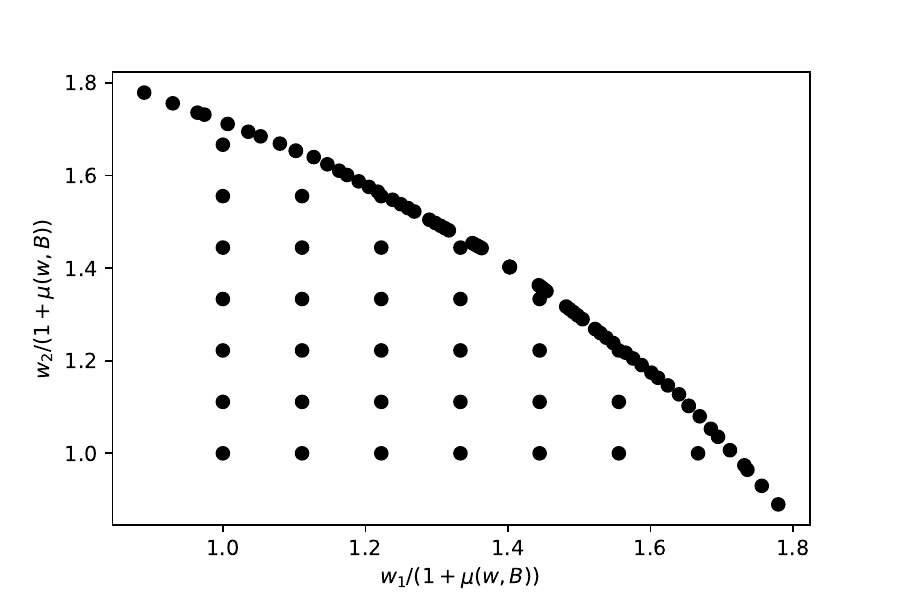}
	\end{subfigure}
	\caption{The left plot depicts how the multiplicative shading factor $1/(1 +\mu(w,B))$ varies with buyer weight vector $w$ (budget $B = 0.6$ is the same for every buyer type). On the right, we plot the paced weight vectors of the buyer types.}
	\label{figure_simulation}
\end{figure}
The structural properties discussed in Proposition~\ref{structural_prop} are clearly evident in Figure \ref{figure_simulation}. In this scenario, the buyer types are uniformly distributed on $(1,2) \times (1,2) \times \{0.6\}$ and, as a consequence, all buyers have identical budgets equal to 0.6. At equilibrium, it can be seen that the co-linear buyer types (i.e., buyers whose weight vectors $w$ are co-linear) who have a positive multiplier get paced down to the critical buyer type who exactly spends her budget. Moreover, at equilibrium, the boundary that separates the paced buyer types from the unpaced buyer types---the curve in which the critical buyer types lie---can be clearly observed in the left-hand plot in Figure \ref{figure_simulation}. Finally, we constructed random discrete instances by uniformly sampling 50 buyer weight vectors and 20 item feature vectors from the square $(1,2) \times (1,2)$, and setting the number of buyers to be $N=3$ and the budget of all buyer types to be $B = 2$. We found that our dual-based dynamics always converged within 250 iterations to pacing-based bidding strategies which on average were within 2.5\% of the utility-maximizing budget feasible bidding strategy.

\section{Conclusion and Future Work}

This paper introduces a natural contextual valuation model and characterizes the equilibrium bidding behavior of budget-constrained buyers in first-price auctions in this model. We extend this result to other standard auctions and establish revenue equivalence among them. Due to the extensive focus on second-price auctions, previous works endorse bid-pacing as the framework of choice for budget management in the presence of strategic buyers. Our results suggest that value-pacing, which coincides with bid-pacing in second-price auctions, is an appropriate framework to manage budgets across all standard auctions.

An important open question we leave unanswered is that of optimizing the reserve prices to maximize seller revenue under equilibrium bidding. In general, optimizing under equilibrium constraints is usually challenging, so it is interesting to explore whether our model possesses additional structure that allows for tractability. Another related question is that of characterizing the revenue-optimal mechanism for our model. Our contextual-value model can capture multi-item auctions with additive valuations as a special case (by interpreting each context as a different item), which is a notoriously hard setting for revenue maximization, even in the absence of budget constraints.
  Investigating dynamics in first-price auctions with strategic budget-constrained buyers is another interesting open direction worth exploring. We also leave open the question of efficient computation of the pacing-based equilibria discussed in this paper. Addressing this question will likely require choosing a suitable method of discretization and tie-breaking, without which equilibrium existence may not be guaranteed (see, e.g., \citealt{conitzer2017multiplicative,Babaioff0HIL21}). Finally, another interesting research direction is to develop conditions that guarantee uniqueness of an equilibrium. In light of recent results by \citet{conitzer2017multiplicative}, we conjecture that, without further assumptions, the equilibrium would generally not be unique.

\singlespace

\bibliographystyle{plainnat}
\bibliography{refs.bib}

\appendix

\newpage
\pagenumbering{arabic}\renewcommand{\thepage}{ec \arabic{page}}

\begin{centering}
\LARGE
Electronic Companion:\\[1em]
Contextual Standard Auctions with Budgets\\[1em]
\large
Santiago Balseiro, Christian Kroer, Rachitesh Kumar\\[1em]
\today\\
\end{centering}
\renewcommand{\theequation}{\thesection-\arabic{equation}}

\rk{
\section{Counter Example for Deterministic Context}\label{appendix:det-context-counter}

\begin{example*}Consider an auction with $n=2$ budget-constrained buyers per auction. Buyers draw their value $v$ uniformly from the interval $[0,1]$ and each with a budget of $1/8$, i.e., $T = \{(v,1/8) \in \R^2 \mid 0 \leq v \leq 1\}$ is the type space where the first component denotes the value and the second one denotes the budget. (A uniform distribution of values can be achieved by a number of fixed contexts and weight vector distributions, for example suppose the item context is $\alpha = (1)$ and the weight vectors $w$ are distributed uniformly in $[0,1]$. This would yield values $v=w^T \alpha$ that are uniformly distributed) As in our model, the buyers would like to satisfy their budget constraint in expectation at the interim stage: A buyer with value $v$ would like to spend less than $1/8$ in expectation over the value of the other buyer. Moreover, assume that the ties are broken uniformly. We will show that there does not exist a symmetric continuous non-decreasing Bayes-Nash equilibrium strategy $\beta: T \to \R_{\geq 0}$ for this example.

  Let $F$ denote the distribution of bids under $\beta$. We first show that $F$ must contain an atom. For contradiction, suppose not, i.e., $F$ is atomless. Since $F$ is atomless $\beta$ should strictly increasing. Then, the probability that a buyer with value $v$ wins the item in equilibrium is given by $v$. This follows because the bidder with the highest value wins when strategies are symmetric and strictly increasing together with the fact that values are uniformly distributed. Therefore, if the buyer with value $1$ bids $b$, her expected expenditure is given by $b$, which must be less than or equal to $1/8$ due to the budget constraint. Hence, $\beta(v)\leq 1/8$ for all $v \in [0,1]$. It is easy to see that the optimal bid for any buyer with value $v \in [1/2,1]$, in response to the other buyer using $\beta$, is $1/8$. This contradicts the assumption that $F$ is atomless.

  Hence, $F$ has an atom $b^*$. As $\beta$ is non-decreasing, there exists an interval $[x, x + \epsilon]$, where $\epsilon > 0$, such that $\beta(v) = b^*$ for all $v \in [x, x + \epsilon]$ and $\beta(v) < b^*$ for all $v < x$. If $b^* = x = 0$, then bidding infinitesimally more than $b^*$ is strictly better for a buyer with value $x + \epsilon$ because her probability of winning increases by at least $\epsilon/2$ without violating her budget constraint, thereby contradicting the fact that $\beta$ is a BNE. Hence, we have $0 < b^* = \beta(x) < x$, because if $b^* = x$, then bidding slightly less than $b^*$ would give the buyer with value $x$ a higher utility. Finally, the continuity of $\beta$ implies that, for a buyer with a value that is infinitesimally smaller than $x$, it is optimal to bid $b^*$ since it increases her probability of winning by at least $\epsilon/2$ with only an infinitesimal increase in bid. This contradicts the definition of a BNE, thereby implying that no symmetric continuous non-decreasing BNE strategy exists for this example.\qed
\end{example*}

It is worth noting that a BNE does exist for the above example if the seller employs the second-price auction. In particular, we claim that the following strategy forms a BNE for the  second-price auction:
\begin{align*}
	\beta(v) = \begin{cases}
 	1/4 \quad \text{if } v \geq 1/4\\
 	v \quad \text{if } v < 1/4
 \end{cases}
\end{align*}

First, observe that if a buyer bids $b > 1/4$, her total expected expenditure (expectation over the other buyer's value) is given by
\begin{align*}
	\frac{3}{4} \cdot \frac{1}{4} + \int_{0}^{1/4} v dv = \frac{7}{32}
\end{align*}
which is strictly greater than $1/8$. Therefore no buyer can bid strictly more than $1/4$ without violating her budget constraint. Moreover, a buyer who bids exactly $1/4$ spends
\begin{align*}
	\frac{1}{2} \cdot \frac{3}{4} \cdot \frac{1}{4} + \int_{0}^{1/4} v dv = \frac{1}{8}
\end{align*}
and satisfies her budget constraint.

Consider a buyer with value $v > 1/4$ and suppose the competing buyer bids using $\beta$. As argued above, her budget constraints her to select a bid $b \leq 1/4$. Her utility from bidding $b < 1/4$ is given by $v \cdot b - b^2/2$, which is at most $v \cdot (1/4) - (1/4)^2/2 = v/4 - 1/32$. On the other hand, the utility she receives from bidding $b = 1/4$ is given by
\begin{align*}
	v\left[ \frac{1}{2} \cdot \frac{3}{4} + \frac{1}{4} \right]	- \frac{1}{8}
\end{align*}
which is strictly greater than $v/4 - 1/32$ because $v > 1/4$. Next, consider a buyer with value $v \leq 1/4$. If she ignores her budget constraint, it is a weakly dominant strategy to bid her value. As we have shown above, bidding her value also respects her budget constraint and is therefore a best response. Hence, we have shown that, if the other buyer bids using $\beta$, it is a best response for any buyer with value $v > 1/4$ to bid $1/4$ and for any buyer with value $v \leq 1/4$ to bid $v$, as desired.
}

\section{Existence of Symmetric First-Price Equilibrium} \label{appendix_proof_existence}

\subsection{Preliminaries on Continuity}

The following lemma establishes the almost sure continuity of the CDF of the distribution of the maximum of paced values $H_\alpha^\mu$, which is used extensively in our analysis.

\begin{lemma}\label{pacing_strategy_properties}
	For every $\mu: \Theta \to \R_{\geq 0}$, the following properties hold:
	\begin{itemize}[topsep = 0cm]
		\item[a.] $\lambda_\alpha^\mu$ and $H_\alpha^\mu$ have continuous CDFs almost surely w.r.t. $\alpha \sim F$
		\item[b.] $\sigma_\alpha^\mu$ is continuous almost surely w.r.t. $\alpha \sim F$
		\item[c.] $\sigma_\alpha^\mu$ is non-decreasing. Furthermore, for $x \in [0, \omega]$ and $\alpha \in A$ such that $H_\alpha^\mu$ is continuous, the following statement holds almost surely w.r.t. $Y \sim H_\alpha^\mu$,
			\begin{align*}
        		\mathds{1}\{x \geq r(\alpha), x \geq Y \}= \mathds{1}\left\{ \sigma_\alpha^\mu(x) \geq r(\alpha),\ \sigma_\alpha^\mu(x) \geq \sigma_\alpha^\mu(Y)\right\}
    		\end{align*}
		\item[d.] Almost surely w.r.t. $\alpha \sim F$, when $x_1, x_2 \sim \lambda^\mu_\alpha$ i.i.d., the probability of $\sigma^\mu_\alpha(x_1) = \sigma_\alpha^\mu(x_2)$ is zero.
	\end{itemize}
\end{lemma}

Part (a) states that the distributions of paced values are atomless almost surely w.r.t. the items $\alpha \sim F$. This property is crucial because it allows us to leverage the known result establishing the existence of a symmetric equilibrium in the i.i.d. setting under arbitrary tie-breaking rules, which holds only if the distribution of values is atom-less.
Part (b) is a direct consequence of the definition of $\sigma_\alpha^\mu$. Part (c) follows from part (a). Part (c) says that when everyone uses the strategy $\sigma_\alpha^\mu$, with probability $1$, a buyer who has paced value $x$ for item $\alpha$ has the highest bid if and only if she has the highest paced value, which plays a key role in our analysis. Finally, part (d), says that ties are a zero probability event when players use the value-pacing-based strategy. We will need the following lemma to prove Lemma~\ref{pacing_strategy_properties}


\begin{lemma}\label{lemma_infinite_sequence}
    Consider a set $Y = \{y_\alpha\}_{\alpha \in I}$ with $y_\alpha > 0$, where $I$ is an index set. If $I$ is uncountable, then there exists a countable sequence $\{\alpha_n\}_{n\in \mathbb{N}} \subset I$ such that $\sum_{n \in \mathbb{N}} y_{\alpha_n} = \infty$.
\end{lemma}

\begin{proof}
    Rewrite $I$ as $I = \cup_{n \in \mathbb{Z}_+} \{\alpha \in I \mid y_\alpha \geq 1/n\}$. It is a well-known fact that a countable union of countable sets is countable (see Theorem 2.12 of \cite{rudin1964principles}). Therefore, in order for $I$ to be uncountable, there must exist $n_0$ such that $\{\alpha \in I \mid y_\alpha \geq 1/n_0\}$ is uncountable. It follows that we can find a countable sequence $\{y_{\alpha_n}\}_{n \in \mathbb{N}}$ such that $y_{\alpha_n} \geq 1/n_0$ for all $n \in \mathbb{N}$. For this sequence, $\sum_{n \in \mathbb{N}} y_{\alpha_n} = \infty$.
\end{proof}

We now state the proof of Lemma~\ref{pacing_strategy_properties}.

\begin{proof}[Proof of Lemma~\ref{pacing_strategy_properties}]
\
	\begin{itemize}
		\item[a.] 
%
%
	Consider a pacing function $\mu: \Theta \to \mathbb{R}_{\geq 0}$. Let $\alpha_1, \alpha_2 \in A$ be linearly independent feature vectors and $x_1, x_2 \in [0,\omega]$ be two possible item values. We consider the set of buyer types which have paced value $x_1$ for $\alpha_1$ and paced value $x_2$ for $\alpha_2$, i.e.,  define
	\begin{align*}
		S \coloneqq \left\{ (w,B) \in \Theta \bigg\lvert\ \frac{w^T\alpha_1}{1 + \mu(w,B)} = x_1;\ \frac{w^T\alpha_2}{1 + \mu(w,B)} = x_2 \right\}
	\end{align*}
	
	Observe that, for $(w,B) \in S$ and $c \coloneqq x_1/x_2$, we have $w^T\alpha_1 = c \cdot w^T\alpha_2$.
	Therefore, the set $T = \{w \in \Theta_w \mid w^T(\alpha_1 - c \alpha_2) = 0\}$ is a superset of the set $S_w$.
	Hence, $S \subset T \times (B_{\min}, U)$,
	which in combination with the assumption that $G$ has a density implies $G(S) = 0$.
		
		   Define $J = \left\{\alpha/\|\alpha\|\ \mid \exists\ x_\alpha> 0 \textrm{ s.t. } G(\{(w,B) \mid w^T\alpha/(1 + \mu(w,B) = x_\alpha\}) >0)\right\}$. Suppose $J$ is uncountable. Then, by Lemma~\ref{lemma_infinite_sequence}, there exists a countable sequence $\{\alpha_m\}_{m \in \mathbb{N}}$ and $\{x_{\alpha_m}\}_{m \in \mathbb{N}}$ such that $\alpha_i/\|\alpha_i\| \neq \alpha_j/\|\alpha_j\|$ for all $i \neq j$ and
		    \begin{align*}
		        \sum_{m} G(\{(w,B) \mid w^T\alpha_m/(1 + \mu(w,B) = x_{\alpha_m}\}) >0) = \infty.
		    \end{align*}
		
		   Set $S_m \coloneqq \{(w,B) \mid w^T\alpha_m/(1 + \mu(w,B) = x_{\alpha_m}\}$. We have shown above that $G(S_i \cap S_j) = 0$ for all $i \neq j$. Therefore, for all $m \geq 1$, we have $G(S_m \cap\ (\cup_{j < m} S_j)) = 0$, which implies $G(S_m \cap\ (\cup_{j < m} S_j)^C) = G(S_m)$. This contradicts $G(\cup_m S_m) \leq 1$ as $G(\cup_m S_m) = \sum_m G(S_m \cap\ (\cup_{j < m} S_j)^C) = \sum_{m} G(\alpha_m^T s = x_{\alpha_m}) = \infty$. Hence, $J$ is countable. Observe that
		    \begin{align*}
		        \left\{\frac{\alpha}{\|\alpha\|}\ \biggr\lvert\ \textrm{ $\lambda^\mu_\alpha$ has an atom}\right\} \subset J
		    \end{align*}
		    As $F$ has a density, we get $F(\textrm{cone}(J)) = 0$. Therefore, $\lambda^\mu_\alpha$ has no atoms almost surely w.r.t. $\alpha \in A$, i.e., $\lambda_\alpha^\mu$ has a continuous CDF almost surely w.r.t. $\alpha \sim F$. Moreover, this implies that $H_\alpha^\mu$ has a continuous CDF almost surely w.r.t. $\alpha \sim F$.
		
		   \item[b.] Follows from the fact that the integral of every bounded function is continuous.
		
		   \item[c.] Using Lemma 2.2.8 from \cite{durrett2019probability}, we can write
		   	\begin{align*}
		   		\sigma_\alpha^\mu(x) = x - \int_{r(\alpha)}^{x}  \frac{H_\alpha^\mu(s)}{H_\alpha^\mu(x)} ds = r(\alpha) + \int_{r(\alpha)}^{x}  \frac{H_\alpha^\mu(x) - H_\alpha^\mu(s)}{H_\alpha^\mu(x)} ds = \mathbb{E}_{Y \sim H_\alpha^\mu}\left[ \max\{Y, r(\alpha)\} \mid Y< x \right]
		   	\end{align*}
		   	From the last term, it can be easily seen that $\sigma^\mu_\alpha$ is non-decreasing.
		   	
		   	Observe that $\mathds{1}(x \geq r(\alpha), x \geq Y) \leq \mathds{1}\left(\sigma_\alpha^\mu(x) \geq r(\alpha), \sigma_\alpha^\mu(x) \geq \sigma_\alpha^\mu(Y)\right)$ always holds as $\sigma^\mu_\alpha$ is non-decreasing and $\sigma^\mu_\alpha(r(\alpha)) = r(\alpha)$. Moreover,
		    \begin{align*}
		        \mathds{1}(x \geq r(\alpha), x \geq Y) < \mathds{1}\left(\sigma_\alpha^\mu(x) \geq r(\alpha), \sigma_\alpha^\mu(x) \geq \sigma_\alpha^\mu(Y)\right) &\implies x \geq r(\alpha), x<Y,\ \sigma_\alpha^\mu(x) \geq \sigma_\alpha^\mu(Y)\\
		        &\implies x \geq r(\alpha), x<Y,\ \sigma_\alpha^\mu(x) = \sigma_\alpha^\mu(Y)
		    \end{align*}
		    because $\sigma_\alpha^\mu(x) \geq r(\alpha)$ if and only if $x \geq r(\alpha)$, and $\sigma_\alpha^\mu$ is non-decreasing.
		
		    Therefore, it is enough to show for $\alpha \in A$ such that $H_\alpha^\mu$ is continuous and $x \geq r(\alpha)$, we have
		    \begin{align*}
		        H^\mu_\alpha\left(\{y \in[0,\omega] \mid x < y, \sigma_\alpha^\mu(x) = \sigma_\alpha^\mu(y)\}\right) = 0
		    \end{align*}
		    Suppose the above statement doesn't hold for some $\alpha \in A$ such that $H_\alpha^\mu$ is continuous and $x \geq r(\alpha)$. Then, for $y = \sup \{t > x \mid \sigma_\alpha^\mu(t) = \sigma_\alpha^\mu(x)\}$, we have $\sigma_\alpha^\mu(y) = \sigma_\alpha^\mu(x)$ (as $\sigma_\alpha^\mu$ is continuous) and $H_\alpha^\mu((x,y]) > 0$. First, consider the case when $H_\alpha^\mu(x) > 0$. Observe that
		    \begin{align*}
		        \sigma^\mu_\alpha(y) - \sigma^\mu_\alpha(x) &= y - x  -\int_{r(\alpha)}^{y}  \frac{H^\mu_\alpha(s)}{H^\mu_\alpha(y)} ds + \int_{r(\alpha)}^{x}  \frac{H^\mu_\alpha(s)}{H^\mu_\alpha(x)} ds \\
		        &= y - x - \int_x^{y}  \frac{H^\mu_\alpha(s)}{H^\mu_\alpha(y)} ds + \left(\frac{1}{H^\mu_\alpha(x)} - \frac{1}{H^\mu_\alpha(y)}\right)\int_{r(\alpha)}^{x}  H^\mu_\alpha(s) ds \\
		        &> \left(\frac{1}{H^\mu_\alpha(x)} - \frac{1}{H^\mu_\alpha(y)}\right)\int_{r(\alpha)}^{x}  H^\mu_\alpha(s) ds
		    \end{align*}
		    where the last inequality follows from $H^\mu_\alpha(y) - H^\mu_\alpha(x) = H^\mu_\alpha((x,y]) > 0$.
		    Therefore, $\sigma^\mu_\alpha(y) > \sigma^\mu_\alpha(x)$ because $H^\mu_\alpha(y) > H^\mu_\alpha(x)$, which contradicts $\sigma^\mu_\alpha(y) = \sigma^\mu_\alpha(x)$.
		
		    Next, consider the case when $H^\mu_\alpha(x) = 0$. Then, $H^\mu_\alpha(x) = 0$ and $H^\mu_\alpha(y) = H^\mu_\alpha(x) + H^\mu_\alpha((x,y]) > 0$. Note that
		    \begin{align*}
		        \sigma_\alpha^\mu(y)H^\mu_\alpha(y) = y H^\mu_\alpha(y) - \int_{r(\alpha)}^{y}  H^\mu_\alpha(s) ds = \int_{r(\alpha)}^{y}  [H^\mu_\alpha(y) - H^\mu_\alpha(s)] ds + r(\alpha)H^\mu_\alpha(y)
		    \end{align*}
		    Hence, $\sigma_\alpha^\mu(y) = 0$ if and only if $H^\mu_\alpha(s) = H^\mu_\alpha(y)$ for all $s \in [r(\alpha),y]$ and $r(\alpha) = 0$. As $H^\mu_\alpha(0) = 0$ and $H^\mu_\alpha(y)>0$, we get $\sigma_\alpha^\mu(y) > 0$, which contradicts $\sigma_\alpha^\mu(y) = \sigma_\alpha^\mu(x)$.
    		\item[d.] Consider a $\alpha \in A$ such that $\lambda_\alpha^\mu$ has a continuous CDF and $P_{x \sim \lambda_\alpha^\mu}(\sigma_\alpha^\mu(x) = c) > 0$ for some $c \geq 0$. Then, by the definition of $\sigma^\mu_\alpha$, it must be that $c \geq r(\alpha)$. Moreover, if we let $x_0 = \inf\{x \mid \sigma_\alpha^\mu(x) = c\}$, then $P_{x \sim \lambda_\alpha^\mu}(\sigma_\alpha^\mu(x) = c) > 0$ implies $H^\mu_\alpha\left(\{y \in[0,\omega] \mid x_0 < y, \sigma_\alpha^\mu(x_0) = \sigma_\alpha^\mu(y)\}\right) > 0$. This contradicts the fact we proved as part of the proof of part (c): for $\alpha \in A$ such that $H_\alpha^\mu$ is continuous and $x \geq r(\alpha)$, we have
		    \begin{align*}
		        H^\mu_\alpha\left(\{y \in[0,\omega] \mid x < y, \sigma_\alpha^\mu(x) = \sigma_\alpha^\mu(y)\}\right) = 0
		    \end{align*}
		    	Therefore, when $x \sim \lambda_\alpha^\mu$, the CDF of $\sigma^\mu_\alpha(x)$ is continuous, and hence, if $x_1, x_2 \sim \lambda_\alpha^\mu$ i.i.d., then the probability of $\sigma^\mu_\alpha(x_1) = \sigma_\alpha^\mu(x_2)$ is zero. Part (d) follows from combining this fact with part (a).\qedhere
	\end{itemize}
\end{proof}

\subsection{Strong Duality and Characterizing an Optimal Pacing Strategy}

We begin with the proof of Lemma~\ref{lagrangian_optimal}.

\begin{proof}[Proof of Lemma~\ref{lagrangian_optimal}]
	Note that bidding more than the highest competing bid with a positive probability is not optimal, i.e., if $\mathbb{P}_\alpha(b(\alpha) > \sigma^\mu_\alpha(\omega)) > 0$, then $b$ is not optimal. Therefore, we can restrict our attention to $b$ such that $0 \leq b(\alpha) \leq \sigma^\mu_\alpha(\omega)$ a.s. w.r.t. $\alpha \sim F$. Now, consider such a $b$. As $\sigma^\mu_\alpha(0)= 0$ and $\sigma_\alpha$ is continuous a.s. w.r.t. $\alpha \sim F$, by the Intermediate Value Theorem, there exists $z(\alpha) \in [0,\omega]$ such that $\sigma^\mu_\alpha(z(\alpha)) = b(\alpha)$.
	
	Therefore, with $x(\alpha) \coloneqq w^T\alpha/(1 + t)$, we have
    \begin{align*}
        &\max_{b(.)}\ \mathbb{E}_{\alpha,\{\theta_i\}_{i=1}^{n-1}}\left[ \left(\frac{w^T \alpha}{1+t} - b(\alpha)\right)\ \mathds{1}\{b(\alpha) \geq \max(r(\alpha), \{\beta(\theta_i, \alpha)\}_i) \} \right]\\
        =\ &\max_{b(.)}\ \mathbb{E}_\alpha
        \mathbb{E}_{Y\sim H^\mu_\alpha} \left[\left(x(\alpha) - b(\alpha)\right) \mathds{1}\left\{ b(\alpha) \geq \max(r(\alpha), \sigma^\mu_\alpha(Y)) \right\}\right]\\
        =\ &\max_{z(.)}\ \mathbb{E}_\alpha
        \mathbb{E}_{Y\sim H^\mu_\alpha} \left[\left(x(\alpha) - \sigma^\mu_\alpha(z(\alpha))\right) \mathds{1}\left\{\sigma^\mu_\alpha(z(\alpha)) \geq \max(r(\alpha), \sigma^\mu_\alpha(Y)) \right\}\right]\\
        =\ &\max_{z(.)}\ \mathbb{E}_\alpha
        \mathbb{E}_{Y\sim H^\mu_\alpha} \left[\left(x(\alpha) - \sigma^\mu_\alpha(z(\alpha))\right) \mathds{1}\left\{z(\alpha) \geq \max(r(\alpha), Y) \right\}\right]\\
        =\ &\max_{z(.)}\ \mathbb{E}_\alpha\left[\left(x(\alpha) - \sigma^\mu_\alpha(z(\alpha))\right) H^\mu_\alpha\left(z(\alpha)\right))\ \mathds{1}\left\{z(\alpha) \geq r(\alpha)\right\}\right]
    \end{align*}
    where the third equality follows from part (c) of Lemma~\ref{pacing_strategy_properties}. Hence, to prove the claim, it is enough to show that for all $\alpha \in A$, we have
    \begin{align*}
        x(\alpha) \in arg\max_{z(.)}\ \left(x(\alpha) - \sigma^\mu_\alpha(z(\alpha))\right) H^\mu_\alpha\left(z(\alpha)\right)) \mathds{1}\left\{z(\alpha) \geq r(\alpha) \right\}
    \end{align*}
    The above statement holds trivially for $\alpha$ such that $x(\alpha) < r(\alpha)$, because $\sigma^\mu_\alpha(t) \geq r(\alpha)$ when $t \geq r(\alpha)$. Consider $\alpha \in A$ for which $x(\alpha)\geq r(\alpha)$. Then, for $z(\alpha) \geq r(\alpha)$,
    \begin{align*}
        \left(x(\alpha) - \sigma^\mu_\alpha(z(\alpha))\right) H^\mu_\alpha\left(z(\alpha)\right)) &= x(\alpha) H^\mu_\alpha(z(\alpha)) - z(\alpha)H^\mu_\alpha(z(\alpha)) + \int_{r(\alpha)}^{z(\alpha)}  H^\mu_\alpha(s) ds\\
        &= (x(\alpha) - z(\alpha)) H^\mu_\alpha(z(\alpha)) + \int_{r(\alpha)}^{z(\alpha)}  H^\mu_\alpha(s) ds
    \end{align*}

    Therefore, for $z(\alpha) \geq r(\alpha)$, we have
    \begin{align*}
        \left(x(\alpha) - \sigma^\mu_\alpha(x(\alpha))\right) H^\mu_\alpha(x(\alpha)) - \left(x(\alpha) - \sigma^\mu_\alpha(z(\alpha))\right) H^\mu_\alpha(z(\alpha)) = (z(\alpha) - x(\alpha))H^\mu_\alpha(z(\alpha)) - \int_{x(\alpha)}^{z(\alpha)}  H^\mu_\alpha(s) ds \geq 0
    \end{align*}
    where the inequality holds regardless of $z(\alpha) \geq x(\alpha)$ or $x(\alpha) \geq z(\alpha)$. Furthermore, for $z(\alpha) < r(\alpha)$, we have
    \begin{align*}
        \left(x(\alpha) - \sigma^\mu_\alpha(x(\alpha))\right) H^\mu_\alpha\left(x(\alpha)\right)) \mathds{1}\left\{x(\alpha) \geq r(\alpha) \right\} \geq \left(x(\alpha) - \sigma^\mu_\alpha(z(\alpha))\right) H^\mu_\alpha\left(z(\alpha)\right)) \mathds{1}\left\{z(\alpha) \geq r(\alpha) \right\} = 0
    \end{align*}
    Hence, $z(\alpha) = x(\alpha)$ is optimal, which completes the proof.
\end{proof}

In the rest of the sub-section, we build towards the proof of Proposition~\ref{optimal_solution}. Recall that the dual objective function is given by
\begin{align*}
	q^\mu(w, B, t) = (1+t)\mathbb{E}_\alpha\left[\mathds{1}\left\{\frac{w^T \alpha}{1+t} \geq r(\alpha) \right\} \int_{r(\alpha)}^{\frac{w^T\alpha}{1+t}} H^\mu_\alpha(s) ds\right] + tB
\end{align*}

We will prove Proposition~\ref{optimal_solution} by first establishing the differentiability of the dual objective function, and then invoking the first-order optimality conditions for the dual-optimal solution. Lemma~\ref{differentiability_dual} will establish the differentiability of the dual objective function. To prove it, we will need the convexity of the dual objective function (Part 1 of Lemma~\ref{compact_dual_space}), the existence of a bounded dual-optimal solution (Part 2 of Lemma~\ref{compact_dual_space}), and the differentiability of the indicator function
\begin{align*}
	\mathds{1}\left\{\frac{w^T \alpha}{1+t} \geq r(\alpha) \right\}	
\end{align*}
as a function of $t$ almost surely w.r.t. $(w,B) \sim G$, which is implied by the continuity of the CDF of $w^T\alpha/r(\alpha)$, when $\alpha \sim F$ (Lemma~\ref{continuity_K}).

\begin{lemma}\label{compact_dual_space}
    For $\mu: \Theta \to \R_{\geq 0}$ and $(w,B) \in \Theta$:
    \begin{enumerate}[topsep = 0 cm]
    	\item $q^\mu(w, B, t)$ is convex as a function of $t$.
    	\item $\min_{t \geq 0} q^\mu(w,B,t) = \min_{t \in [0,\omega/B]} q^\mu(w,B,t)$
    \end{enumerate}
\end{lemma}

\begin{proof}\
	\begin{enumerate}[topsep = 0 cm]
    	\item The objective function of the dual problem of a maximization problem is convex.
    	\item As $H^\mu_\alpha(s) \leq 1$ for all $\alpha \in A$ and $s \in \mathbb{R}$, the following inequalities hold
		    \begin{align*}
		        0 \leq \mathds{1}\left\{\frac{w^T \alpha}{1+t} \geq r(\alpha) \right\} \int_{r(\alpha)}^{\frac{w^T\alpha}{1+t}} H^\mu_\alpha(s) ds \leq \omega \qquad \forall\ t\geq 0, \alpha \in A
		    \end{align*}
		
		    If $t > \omega/B$, then, $q^\mu(w,B,t) \geq tB > \omega \geq q^\mu(w, B, 0)$. Hence, $q^\mu(w,B,t)$, as a function of $t$, has its minimum in the interval $[0, \omega/B]$.
    \end{enumerate}	
\end{proof}

Let $K$ be the distribution of $\alpha/r(\alpha)$ when $\alpha \sim F$, assuming $1/r(\alpha) = 1$ when $r(\alpha) = 0$. For $w \in \Theta_w$, let $K_w$ be the distribution of $w^T\gamma$ when $\gamma \sim K$, i.e., $K_w(\mathcal{B}) \coloneqq K(\{\gamma \mid w^T\gamma \in \mathcal{B}\})$ for all Borel sets $\mathcal{B} \subset \mathbb{R}$.

\begin{lemma}\label{continuity_K}
    $K_w$ has a continuous CDF almost surely w.r.t. $w \sim G_w$.
\end{lemma}
\begin{proof}

    Let $w_1, w_2 \in \Theta_w$ be linearly independent weight vectors and $x_1, x_2 \in \R_{\geq 0}$. We consider the set of items $\alpha$ which satisfy $w_1^T\alpha/r(\alpha) = x_1$ and $w_2^T\alpha/r(\alpha) = x_2$. Define
	\begin{align*}
		S \coloneqq \left\{ \alpha \in A \bigg\lvert\ \frac{w_1^T\alpha}{r(\alpha)} = x_1;\ \frac{w_2^T\alpha}{r(\alpha)} = x_2 \right\}
	\end{align*}
	
	Observe that, for $\alpha \in S$ and $c \coloneqq x_1/x_2$, we have $w_1^T\alpha = c \cdot w_2^T\alpha$. Therefore, the set $T = \{\alpha \in A \mid (w_1 - c w_2)^T\alpha = 0\}$ is a superset of the set $S$. Hence, since $F$ has a density, we get $F(S) = 0$.

   Define $J = \left\{w/\|w\|\ \Big \lvert\ \exists\ x_w> 0 \textrm{ s.t. } F(w^T\alpha/r(\alpha) = x_w) >0)\right\}$. Suppose $J$ is uncountable. Then, by Lemma~\ref{lemma_infinite_sequence}, there exists a countable sequence $\{w_m\}_{m \in \mathbb{N}}$ and $\{x_{w_m}\}_{m \in \mathbb{N}}$ such that $w_i/\|w_i\| \neq w_j/\|w_j\|$ for all $i \neq j$ and
    \begin{align*}
        \sum_{m} F(w_m^T\alpha/r(\alpha) = x_{w_m}) = \infty.
    \end{align*}

   Set $S_m \coloneqq \{\alpha \mid w_m^T \alpha/r(\alpha) = x_{w_m}\}$. We have shown above that $F(S_i \cap S_j) = 0$ for all $i \neq j$. Therefore, for all $m \geq 1$, we have $F(S_m \cap\ (\cup_{j < m} S_j)) = 0$, which implies $F(S_m \cap\ (\cup_{j < m} S_j)^C) = F(S_m)$. This contradicts $F(\cup_m S_m) \leq 1$ as $F(\cup_m S_m) = \sum_m F(S_m \cap\ (\cup_{j < m} S_j)^C) = \sum_{m} F(\alpha_m^T s = x_{\alpha_m}) = \infty$. Hence, $J$ is countable. Observe that
    \begin{align*}
        \left\{\frac{w}{\|w\|}\ \biggr\lvert\ \textrm{ $K_w$ has an atom}\right\} \subset J
    \end{align*}
    As $G_w$ has a density, we get $G_w(\textrm{cone}(J)) = 0$. Therefore, $K_w$ has no atoms almost surely w.r.t. $w \sim G_w$, i.e., $K_w$ has a continuous CDF almost surely w.r.t. $w \sim G_w$.
\end{proof}

\begin{definition}\label{definition_theta'}
	Define $\Theta' \subset \Theta$ to be the set of $(w,B) \in \Theta$ for which $K_w$ has a continuous CDF.
\end{definition}

The following lemma establishes differentiability of the dual objective function.

\begin{lemma}\label{differentiability_dual}
    For all  pacing functions $\mu: \Theta \to \R_{\geq 0}$ and buyer types $(w,B) \in \Theta'$, the dual objective $q^\mu(w,B,t)$ is differentiable as a function of $t$ for $t >-1/2$. Moreover,
    \begin{align*}
        \frac{\partial q^\mu(w,B,t)}{\partial t} = B - \mathbb{E}_\alpha\left[ \sigma_\alpha^\mu \left(\frac{w^T \alpha}{1+t}\right) H_\alpha^\mu \left(\frac{w^T \alpha}{1+t}\right) \mathds{1} \left\{\frac{w^T \alpha}{1+t} \geq r(\alpha) \right\} \right]
    \end{align*}
\end{lemma}

\begin{proof}
	Fix a pacing function $\mu: \Theta \to \R_{\geq 0}$ and a buyer $(w,B) \in \Theta'$. Define
    \begin{align*}
        g(t,\alpha) &\coloneqq \mathds{1}\left\{\frac{w^T \alpha}{1+t} \geq r(\alpha) \right\}\int_{r(\alpha)}^{\frac{w^T\alpha}{1+t}} H^\mu_\alpha(s) ds \qquad \forall\ t > -1/2, \alpha \in A
    \end{align*}

    Note that $x \mapsto \mathds{1}(x \geq r(\alpha)) \int_{r(\alpha)}^x H^\mu_\alpha(s) ds$ is a non-decreasing convex function because $H^\mu_\alpha$ is non-decreasing. Moreover, it is easy to verify using the second order sufficient condition that $t \mapsto \frac{w^T\alpha}{1+t}$ is convex. As $t \mapsto g(t ,\alpha)$ is a composition of these aforementioned functions, it is convex for each $\alpha$.

    Fix $t_0 > -1/2$. Using Lemma~\ref{continuity_K} and the definition of $\Theta'$, we can write
    \begin{align*}
        F\left(\left\{ \alpha \in A \mid \frac{w^T\alpha}{1+t_0} = r(\alpha) \right\}\right) &= F\left(\left\{ \alpha \in A \mid r(\alpha) > 0; \frac{w^T\alpha}{r(\alpha)} = 1+t_0 \right\}\right) + F\left(\left\{ \alpha \in A \mid r(\alpha) = 0; w^T\alpha = 0 \right\}\right)\\
         &\leq K\left(\left\{ \gamma \mid w^T\gamma = 1+t_0 \right\}\right) + F\left(\left\{ \alpha \in A \mid w^T\alpha = 0 \right\}\right) \\
         &= K_w(1+t_0) + 0 \\
         &= 0
    \end{align*}

    Using Theorem 7.46 of \cite{shapiro2009lectures}, we get that $\mathbb{E}_\alpha[g(t,\alpha)]$ is differentiable w.r.t $t$ at $t_0$ and
    \begin{align*}
        \frac{\partial}{\partial t} \mathbb{E}_\alpha[g(t_0,\alpha)] = \mathbb{E}_\alpha\left[\frac{\partial g(t_0,\alpha)}{\partial t}\right].
    \end{align*}

    Therefore, the dual objective $q^\mu(w,B,t)$ is differentiable as a function of $t$ for $t >-1/2$, and
    \begin{align*}
        \frac{\partial q^\mu(w, B, t_0)}{\partial t} &= \mathbb{E}_{\alpha}\left[ g(t_0,\alpha) \right] + (1+t) \frac{\partial}{\partial t} \mathbb{E}_\alpha[g(t_0,\alpha)] + B\\
        &= \mathbb{E}_{\alpha}\left[ g(t_0,\alpha) \right] + (1+t) \mathbb{E}_\alpha\left[\frac{\partial g(t_0,\alpha)}{\partial t}\right] + B\\
        &= \mathbb{E}_\alpha \left[\mathds{1}\left\{\frac{w^T \alpha}{1+t_0} \geq r(\alpha) \right\}\int_{r(\alpha)}^{\frac{w^T\alpha}{1+t_0}} H^\mu_\alpha(s) ds\right]\\
         &\ + (1+t)\mathbb{E}_\alpha \left[\frac{-w^T\alpha}{(1+t_0)^2} H^\mu_\alpha\left(\frac{w^T \alpha}{1+t_0}\right) \mathds{1}\left\{\frac{w^T \alpha}{1+t_0} \geq r(\alpha) \right\}\right] + B\\
        &= B - \mathbb{E}_\alpha \left[\left\{\frac{w^T\alpha}{1+t_0} H^\mu_\alpha\left(\frac{w^T\alpha}{1+t_0}\right) - \int_{r(\alpha)}^{\frac{w^T\alpha}{1+t_0}} H^\mu_\alpha(s) ds\right\}\mathds{1}\left\{\frac{w^T \alpha}{1+t_0} \geq r(\alpha) \right\} \right]
    \end{align*}
\end{proof}

\begin{corollary}\label{continuity_dual}
	For all  pacing functions $\mu: \Theta \to \R_{\geq 0}$ and buyer types $(w,B) \in \Theta'$, $q^\mu(w,B,t)$ is continuous as a function of $t$ for $t$ for $t >-1/2$.
\end{corollary}

\begin{corollary} \label{existence_dual_opt}
    For all  pacing functions $\mu: \Theta \to \R_{\geq 0}$ and buyer types $(w,B) \in \Theta'$, $arg\min_{t \in [0,\omega/B]} q^\mu(w,B,t)$ is non-empty and compact.
\end{corollary}

Corollary~\ref{continuity_dual} is a direct consequence of Lemma~\ref{differentiability_dual} and Corollary~\ref{existence_dual_opt} follows from Weierstrass Theorem. Finally, having established the required lemmas, we are ready to prove Proposition~\ref{optimal_solution}.

\begin{proof}[Proof of Proposition~\ref{optimal_solution}]
	Let $t^* \in \argmin_{t \in [0,\omega/B]} q^\mu(w,B,t)$. According to Theorem 5.1.5 from \cite{bertsekas1998nonlinear}, in order to prove Proposition~\ref{optimal_solution}, it suffices to show the following conditions:
    \begin{itemize}
        \item[(i)] Primal feasibility:
        $$\mathbb{E}_{\alpha,\{\theta_i\}_{i=1}^{n-1}}\left[ \sigma^\mu_\alpha\left(\frac{w^T \alpha}{1+t^*}\right)\ \mathds{1}\left\{\sigma^\mu_\alpha\left(\frac{w^T \alpha}{1+t^*}\right) \geq \max\left(r(\alpha), \{\beta^\mu(\theta_i, \alpha)\}_i\right) \right\}\right] \leq B$$
        \item[(ii)] Dual feasibility: $t^* \geq 0$
        \item[(iii)] Lagrangian Optimality:
        $$\sigma^\mu_\alpha\left(\frac{w^T \alpha}{1+t^*}\right) \in \argmax_{b(.)}\ \mathbb{E}_{\alpha, \{\theta_i\}_{i=1}^{n-1}}\left[ (w^T \alpha - (1+t)b(\alpha))\ \mathds{1}\{b(\alpha) \geq \max(r(\alpha), \{\beta^\mu(\theta_i, \alpha)\}_i)\}\right] + t B$$
        \item[(iv)] Complementary slackness:
        $$t^*.\left\{B - \mathbb{E}_{\alpha,\{\theta_i\}_{i=1}^{n-1}}\left[ \sigma^\mu_\alpha\left(\frac{w^T \alpha}{1+t^*}\right)\ \mathds{1}\left\{\sigma^\mu_\alpha\left(\frac{w^T \alpha}{1+t^*}\right) \geq \max\left(r(\alpha), \{\beta^\mu(\theta_i, \alpha)\}_i\right) \right\}\right]\right\} = 0$$
    \end{itemize}

    First, we simplify the expression for the expected expenditure used in the sufficient conditions (i)-(iv) stated above:
    \begin{align*}
        &\mathbb{E}_{\alpha,\{(\theta_i)\}_{i=1}^{n-1}}\left[ \sigma^\mu_\alpha\left(\frac{w^T \alpha}{1+t^*}\right)\ \mathds{1}\left\{\sigma^\mu_\alpha\left(\frac{w^T \alpha}{1+t^*}\right) \geq \max\left(r(\alpha), \{\beta^\mu(\theta_i, \alpha)\}_i\right) \right\}\right]\\
        = &\mathbb{E}_{\alpha,\{(w_i,B_i)\}_{i=1}^{n-1}}\left[ \sigma^\mu_\alpha\left(\frac{w^T \alpha}{1+t^*}\right)\ \mathds{1}\left\{\sigma^\mu_\alpha\left(\frac{w^T \alpha}{1+t^*}\right) \geq \max\left( r(\alpha),  \left\{\sigma^\mu_\alpha\left(\frac{w_i^T \alpha}{1+\mu(w_i,B_i)}\right)\right\}_i \right) \right\} \right]\\
        = &\mathbb{E}_{\alpha,\{(w_i,B_i)\}_{i=1}^{n-1}}\left[ \sigma^\mu_\alpha\left(\frac{w^T \alpha}{1+t^*}\right)\ \mathds{1}\left\{\frac{w^T \alpha}{1+t^*} \geq \max\left(r(\alpha), \left\{ \frac{w_i^T \alpha}{1+\mu(w_i,B_i)}\right\}_i \right) \right\}\right]\\
        = &\mathbb{E}_\alpha\left[\sigma^\mu_\alpha\left(\frac{w^T \alpha}{1+t^*}\right)H^\mu_\alpha\left(\frac{w^T \alpha}{1+t^*}\right) \mathds{1}\left\{\frac{w^T \alpha}{1+t^*} \geq r(\alpha) \right\}\right]
    \end{align*}

    In the rest of the proof, we establish the aforementioned sufficient conditions (i)-(iv). Note that $t^*$ satisfies the following first order conditions of optimality
		\begin{equation}\label{dual_optimality_conditions}
		    \frac{\partial q^\mu(w, B, t^*)}{\partial t} \geq 0 \qquad t^* \geq 0 \qquad t^*\cdot \frac{\partial q^\mu(w,B,t^*)}{\partial t} = 0
		\end{equation}

	Using Lemma~\ref{differentiability_dual}, we can write
		\begin{align*}
		    \frac{\partial q^\mu(w, B, t^*)}{\partial t} = B - \mathbb{E}_\alpha\left[\sigma^\mu_\alpha\left(\frac{w^T \alpha}{1+t^*}\right)H^\mu_\alpha\left(\frac{w^T \alpha}{1+t^*}\right)\mathds{1}\left\{\frac{w^T \alpha}{1+t^*} \geq r(\alpha) \right\}\right]
		\end{align*}

    To establish the sufficient conditions (i)-(iv), observe that (after simplification) conditions (i), (ii) and (iv) are the same as (\ref{dual_optimality_conditions}), and condition (iii) is a direct consequence of Lemma~\ref{lagrangian_optimal}, thereby completing the proof of Proposition~\ref{optimal_solution}.
\end{proof}

\subsection{Fixed Point Argument}

\begin{proof}[Proof of Lemma~\ref{monotonicity_pacing_function}]\
	\begin{enumerate}
		\item First, observe that
			    \begin{align*}
			        q^\mu(w, B, t) &= \mathbb{E}_\alpha \left[ \mathds{1}\left\{\frac{w^T \alpha}{1+t} \geq r(\alpha) \right\} \int_{r(\alpha)}^{\frac{w^T\alpha}{1+t}} (1+t) H_\alpha^\mu(s) ds + tB \right]\\
			        &= \mathbb{E}_\alpha \left[ \mathds{1}\left\{\frac{w^T \alpha}{1+t} \geq r(\alpha) \right\} \int_{(1+t)r(\alpha)}^{w^T\alpha} H_\alpha^\mu\left(\frac{y}{1+t}\right) dy + tB \right]
			    \end{align*}
			
			    Consider $(w^L,B), (w^H,B) \in \Theta'$ such that $w^L_i < w^H_i$ and $w^L_{-i} = w^H_{-i}$, for some $i \in [d]$. Moreover, consider $t^L, t^H \in [0, \omega/ B_{\min}]$ such that $t^L < t^H$. As $H_\alpha^\mu$ is a non-decreasing function, it is straightforward to check that $-q^\mu(w, B, t)$ has increasing differences w.r.t. $w_i$ and $t$:
			    \begin{align*}
			        q^\mu(w^H, B, t^L) - q^\mu(w^L, B, t^L) \geq q^\mu(w^H, B, t^H) - q^\mu(w^L, B, t^H)
			    \end{align*}
			
			    Theorem 10.7 of \cite{sundaram1996first} in combination with the definition of $\ell^\mu$ imply $\ell^\mu(w^H, B) \geq \ell^\mu(w^L, B)$.
		\item Consider $(w,B^L), (w,B^H) \in \Theta'$ such that $B^L < B^H$ and and $t^L, t^H \in [0, \omega/ B_{\min}]$ such that $t^L < t^H$. Then, $-q^\mu(w, B, t)$ has increasing differences w.r.t. $-B$ and $t$:
			    \begin{align*}
			        q^\mu(w, B^H, t^H) - q^\mu(w, B^L, t^H) = (B^H - B^L)t^H \geq (B^H - B^L)t^L = q^\mu(w, B^H, t^L) - q^\mu(w, B^L, t^L)
			    \end{align*}
			
			    Theorem 10.7 of \cite{sundaram1996first} and the definition of $\ell^\mu$ imply $\ell^\mu(w, B^H) \leq \ell^\mu(w, B^L)$.
	\end{enumerate}
\end{proof}

\begin{proof}[Proof of Lemma~\ref{topology_motivating_properties}]\
	\begin{enumerate}
		\item Theorem 1 of \cite{idczak1994functions} implies measurability of $\ell^\mu$. Moreover, $\ell^\mu$ is bounded by definition.
		\item Consider $\phi \in C_c^1(\Theta, \mathbb{R}^n)$ such that $\|\phi\|_\infty \leq 1$. Then,
			    \begin{align*}
			        V(\ell^\mu, \Theta) &= \int_\Theta \ell^\mu(\theta) \textrm{\textbf{div}}\phi(\theta) d\theta\\
			        &= \sum_{i = 1}^{d+1} \int_\Theta \ell^\mu(\theta) \frac{\partial \phi(\theta)}{\partial \theta_i} d\theta\\
			        &= \sum_{i = 1}^{d+1} \int_{\theta_{-i}} \int_{\theta_i} \ell^\mu(\theta) \frac{\partial \phi(\theta)}{\partial \theta_i} d\theta_i d\theta_{-i}\\
			        &= \sum_{i = 1}^{d+1} \int_{\theta_{-i}} \int_{\theta_i} -\phi(\theta_i, \theta_{-i}) d\ell^\mu(\theta_i) d\theta_{-i}\\
			        &\leq \sum_{i = 1}^{d+1} \int_{\theta_{-i}} \int_{\theta_i}  d\ell^\mu(\theta_i) d\theta_{-i}\\
			        &\leq \sum_{i = 1}^{d+1} \int_{\theta_{-i}} \frac{\omega}{B_{\min}} d\theta_{-i}\\
			        &\leq (d+1) U^{d+1} \frac{\omega}{B_{\min}}
			    \end{align*}
			    where the third equality follows from Fubini's Theorem. The sufficient conditions for Fubini's Theorem to hold are satisfied because $| \ell^\mu \textrm{\textbf{div}}\phi |$ is bounded. Moreover, the fourth equality follows from the integration by parts for Lebesgue-Stieltjes integral and the fact that $\phi$ evaluates to $0$ at the boundaries of $\Theta$ because $\phi$ is compactly supported.\qedhere
	\end{enumerate}
\end{proof}

\begin{proof}[Proof of Lemma~\ref{topology_required_properties}]\
	We start by noting that, as $G$ has a density, if a sequence converges almost surely (or in $L_1$) under the Lebesgue measure on $\Theta$, then it converges almost surely (or in $L_1$) under $G$.
	\begin{enumerate}
		\item If $\mu(\theta) = 0$ for all $\theta \in \Theta$, then $\mu \in \X_0$. Hence, $\X_0$ is non-empty. Consider $\mu_1, \mu_2 \in \X_0$ and $a \in [0,1]$. Then, $a \mu_1 + (1 - a) \mu_2 \in [0, \omega/ B_{\min}]$ and for $\phi \in C_c^1(\Omega, \mathbb{R}^n)$ s.t. $\|\phi\|_\infty \leq 1$, we have
		    \begin{align*}
		        \int_{\Omega} \{a \mu_1 + (1-a) \mu_2\}(\theta) \textrm{\textbf{div}}\phi(\theta) d\theta &= a \int_{\Omega} \mu_1(\theta) \textrm{\textbf{div}}\phi(\theta) d\theta + (1 - a) \int_{\Omega} \mu_2(\theta) \textrm{\textbf{div}}\phi(\theta) d\theta \\
		        &\leq \frac{(d+1)U^{d+1} \omega}{B_{\min}}
		    \end{align*}
		
		    Hence, $\X_0$ is convex.\\
		
		    Consider a sequence $\{\mu_n\} \subset \X_0$ and $\mu \in L_1(\Theta)$ such that $\mu_n \xrightarrow{L_1} \mu$. Then, there exists a subsequence $\{n_k\}$ such that $\mu_{n_k} \xrightarrow{\textrm{a.s.}} \mu$ as $k \to \infty$. Hence, $\text{range}(\mu) \subset [0, \omega/ B_{\min}]$. Moreover, by the semi-continuity of total variation (Remark 3.5 of \cite{ambrosio2000functions}), we have
		    \begin{align*}
		        V(\mu, \Theta) \leq \liminf_{n \to \infty} V(\mu_n, \Theta) \leq (d+1) U^{d+1} \omega/ B_{\min}
		    \end{align*}
		    Therefore, $\X_0$ is closed. To see why $\X_0$ is compact, consider a sequence $\{\mu_n\} \subset \X_0$. Then, by Theorem 3.23 of \cite{ambrosio2000functions}, there exists a subsequence $\{n_k\}$ and $\mu \in BV(\Theta)$ such that $\mu_{n_k}$ converges to $\mu$ in the weak* topology, which implies convergence in $L_1(\Theta)$ (Proposition 3.13 of \cite{ambrosio2000functions}). Combining this with the fact that $\X_0$ is closed, completes the proof of compactness of $\X_0$.
		
		\item For contradiction, suppose $f$ is not continuous.  Then, there exists $\epsilon > 0$, a sequence $\{(\mu_n,\hat{\mu}_n)\}_n \subset \mathcal{X}_0 \times \mathcal{X}_0$ and $(\mu,\hat{\mu}) \in \mathcal{X}_0 \times \mathcal{X}_0$ such that $\lim_{n \to \infty} (\mu_n,\hat{\mu}_n) = (\mu,\hat{\mu})$ and $|f(\mu_n, \hat{\mu}_n) - f(\mu, \hat{\mu})| \geq \epsilon$ for all $n \in \mathbb{N}$. As $\mu_n \xrightarrow{L_1} \mu$, there exists a subsequence $\{n_k\}_k$ such that $\mu_{n_k} \xrightarrow{a.s.} \mu$ when $k \to \infty$. Moreover, $\hat{\mu}_n \xrightarrow{L_1} \hat{\mu}$ implies $\hat{\mu}_{n_k} \xrightarrow{L_1} \hat{\mu}$. Therefore, there exists a subsequence $\{n_{k_l}\}_l$ such that $\hat{\mu}_{n_{k_l}} \xrightarrow{a.s.} \hat{\mu}$ and $\mu_{n_{k_l}} \xrightarrow{a.s.} \mu$ as $l \to \infty$. Here, we have repeatedly used the fact that $L_1$ convergence implies the existence of a subsequence that converges a.s. Hence, after relabelling for ease of notation, we can write that there exists $\epsilon > 0$, a sequence $\{(\mu_n,\hat{\mu}_n)\}_n \subset \mathcal{X}_0 \times \mathcal{X}_0$ and $(\mu,\hat{\mu}) \in \mathcal{X}_0 \times \mathcal{X}_0$ such that $\mu_n \xrightarrow{a.s.} \mu$, $\hat{\mu}_n \xrightarrow{a.s.} \hat{\mu}$ and $|f(\mu_n, \hat{\mu}_n) - f(\mu, \hat{\mu})| \geq \epsilon$ for all $n \in \mathbb{N}$.
		
			First, observe that $\mu_n \xrightarrow{a.s.} \mu$ implies $w^T\alpha/(1+\mu_n(w,B)) \xrightarrow{a.s.} w^T\alpha/(1+\mu(w,B))$ and hence, $\lambda^{\mu_n}_\alpha \xrightarrow{d} \lambda^{\mu}_\alpha$ for all $\alpha \in A$. As $\lambda^\mu_\alpha$ is continuous almost surely w.r.t. $\alpha$, by the definition of convergence in distribution, we get that $\lim_{n \to \infty} \lambda^{\mu_n}_\alpha(s) = \lambda^{\mu}_\alpha(s)$ for all $s \in \mathbb{R}$ a.s. w.r.t. $\alpha \sim F$. Therefore, $\lim_{n \to \infty} H^{\mu_n}_\alpha(s) = H^{\mu}_\alpha(s)$ for all $s \in \mathbb{R}$, a.s. w.r.t. $\alpha \sim F$.
		
			Also, note that $\lambda^{\hat{\mu}_n}_\alpha$ and $\lambda^{\hat{\mu}}_\alpha$ are atom-less almost surely w.r.t. $\alpha$. Let $\bar{A} \subset A$ be the set of $\alpha$ such that $\lim_{n \to \infty} H^{\mu_n}_\alpha(s) = H^{\mu}_\alpha(s)$ for all $s \in \mathbb{R}$ and $\{\lambda^{\hat{\mu}_n}_\alpha$, $\lambda^{\hat{\mu}}_\alpha\}$ are atom-less. Therefore, $F(\bar{A}) = 1$.  For $s \in \mathbb{R}$ and $\alpha \in \bar{A}$, we get
		    \begin{align*}
		        \lim_{n \to \infty}  \mathds{1}\left\{ \frac{w^T\alpha}{1+\hat{\mu}_n(w,B)} \geq s \geq r(\alpha) \right\} = \mathds{1}\left\{\frac{w^T\alpha}{1+\hat{\mu}(w,B)} \geq s \geq r(\alpha)\right\}
		    \end{align*}
		    a.s. w.r.t. $(w,B) \sim G$. Note that the set of measure zero on which the above equality doesn't hold may depend on $\alpha, s$.
		
		    Fix $s \in [0, \omega]$ and $\alpha \in \bar{A}$. Combining these a.s. convergence statements, we get
		    \begin{align*}
		        \lim_{n \to \infty}& (1+\hat{\mu}_n(w,B)) H_\alpha^{\mu_n}(s) \mathds{1}\left\{ \frac{w^T\alpha}{1+\hat{\mu}_n(w,B)} \geq s \geq r(\alpha) \right\}\\
		        = &(1+\hat{\mu}(w,B)) H_\alpha^\mu(s) \mathds{1}\left\{\frac{w^T\alpha}{1+\hat{\mu}(w,B)} \geq s \geq r(\alpha)\right\}
		    \end{align*}
		    a.s. w.r.t. $(w,B) \sim G$.
		
		    Furthermore, we can use the Dominated Convergence Theorem (as the sequence is bounded) to show
		    \begin{align*}
		        \lim_{n \to \infty} &\mathbb{E}_{(w,B)}\left[ (1+\hat{\mu}_n(w,B)) H_\alpha^{\mu_n}(s) \mathds{1}\left\{ \frac{w^T\alpha}{1+\hat{\mu}_n(w,B)} \geq s \geq r(\alpha) \right\} \right]\\
		        =& \mathbb{E}_{(w,B)}\left[ (1+\hat{\mu}(w,B)) H_\alpha^\mu(s) \mathds{1}\left\{ \frac{w^T\alpha}{1+\hat{\mu}(w,B)} \geq s \geq r(\alpha) \right\} \right]
		    \end{align*}

		    Keep $s \in [0,\omega]$ fixed and apply the Dominated Convergence Theorem for a second time to obtain,
		    \begin{align*}
		        \lim_{n \to \infty}& \mathbb{E}_\alpha \left[ \mathbb{E}_{(w,B)}\left[ (1+\hat{\mu}_n(w,B)) H_\alpha^{\mu_n}(s) \mathds{1}\left\{ \frac{w^T\alpha}{1+\hat{\mu}_n(w,B)} \geq s \geq r(\alpha) \right\} \right] \right]\\
		        =& \mathbb{E}_\alpha \left[ \mathbb{E}_{(w,B)}\left[ (1+\hat{\mu}(w,B)) H_\alpha^\mu(s) \mathds{1}\left\{ \frac{w^T\alpha}{1+\hat{\mu}(w,B)} \geq s \geq r(\alpha) \right\} \right] \right]
		    \end{align*}
		
		    Finally, apply the Dominated Convergence Theorem for the third time to obtain,
		    \begin{align*}
		        \lim_{n \to \infty}& \int_0^{\omega} \mathbb{E}_\alpha \left[ \mathbb{E}_{(w,B)}\left[ (1+\hat{\mu}_n(w,B)) H_\alpha^{\mu_n}(s) \mathds{1}\left\{ \frac{w^T\alpha}{1+\hat{\mu}_n(w,B)} \geq s \geq r(\alpha) \right\} \right] \right] ds\\
		        =& \int_0^{\omega} \mathbb{E}_\alpha \left[ \mathbb{E}_{(w,B)}\left[ (1+\hat{\mu}(w,B)) H_\alpha^\mu(s) \mathds{1}\left\{ \frac{w^T\alpha}{1+\hat{\mu}(w,B)} \geq s \geq r(\alpha) \right\} \right] \right] ds
		    \end{align*}
		
		    As we are dealing with non-negative random variables, we can apply Fubini's Theorem to rewrite the above statement as
		    \begin{align*}
		    	&\lim_{n \to \infty} \mathbb{E}_{(w,B)}\mathbb{E}_\alpha \left[ (1+\hat{\mu}_n(w,B)) \mathds{1}\left(\frac{w^T\alpha}{1+\hat{\mu}_n(w,B)} \geq r(\alpha) \right) \int_{r(\alpha)}^{\frac{w^T\alpha}{1+\hat{\mu}_n(w,B)}}  H_\alpha^{\mu_n}(s) ds \right]\\
		        =& \lim_{n \to \infty} \mathbb{E}_{(w,B)}\mathbb{E}_\alpha \left[ \int_0^{\omega} (1+\hat{\mu}_n(w,B)) H_\alpha^{\mu_n}(s) \mathds{1}\left\{ \frac{w^T\alpha}{1+\hat{\mu}_n(w,B)} \geq s \geq r(\alpha) \right\} ds \right]\\
		        =&\ \mathbb{E}_{(w,B)} \mathbb{E}_\alpha \left[ \int_0^{\omega} (1+\hat{\mu}(w,B)) H_\alpha^\mu(s) \mathds{1}\left\{ \frac{w^T\alpha}{1+\hat{\mu}(w,B)} \geq s \geq r(\alpha) \right\} ds \right]\\
		        =&\ \mathbb{E}_{(w,B)}\mathbb{E}_\alpha \left[ (1+\hat{\mu}(w,B)) \mathds{1}\left\{\frac{w^T\alpha}{1+\hat{\mu}(w,B)} \geq r(\alpha) \right\} \int_{r(\alpha)}^{\frac{w^T\alpha}{1+\hat{\mu}(w,B)}}  H_\alpha^{\mu}(s) ds \right]
		    \end{align*}
		
		    Moreover, applying Dominated Convergence Theorem to $\hat{\mu}_n \xrightarrow{a.s.} \hat{\mu}$ yields $\lim_{n \to \infty} E_{(w,B)}[\hat{\mu}_n(w,B)B] = E_{(w,B)}[\hat{\mu}(w,B)B]$. Together, the above statements imply
		     $\lim_{n \to \infty} f(\mu_n, \hat{\mu}_n) = f(\mu, \hat{\mu})$, which is a contradiction.
		
		\item Part (2) allows us to invoke the Berge Maximum Theorem (Theorem 17.31 of \cite{aliprantis2006infinite}), which implies that $C^*_0$ is upper hemi-continuous with non-empty and compact values. Next, we show that $C^*_0(\mu)$ is also convex. Fix $\mu \in \mathcal{X}$. Consider $\hat{\mu}_1, \hat{\mu}_2 \in C^*(\mu)$ and $\lambda \in [0,1]$. Then, by part (1) of Lemma~\ref{compact_dual_space}, we have
		    \begin{align*}
		        f\left(\mu, \lambda \hat{\mu}_1 + (1-\lambda) \hat{\mu}_2\right) &=  \mathbb{E}_{(w,B)} [q^\mu(w, B, \lambda \hat{\mu}_1(w,B) + (1-\lambda) \hat{\mu}_2(w,B))]\\
		        &\leq \lambda \mathbb{E}_{(w,B)} [q^\mu(w, B, \hat{\mu}_1(w,B))] + (1- \lambda) \mathbb{E}_{(w,B)} [q^\mu(w, B, \hat{\mu}_2(w,B))]\\
		        &= \lambda f(\mu, \hat{\mu}_1) + (1-\lambda) f(\mu, \hat{\mu}_2)
		    \end{align*}
		    Hence, $\lambda \hat{\mu}_1 + (1-\lambda) \hat{\mu}_2 \in C^*(\mu)$.\qedhere
			\end{enumerate}
\end{proof}

\begin{proof}[Proof of Lemma~\ref{combined_to_individual}]
	Recall that in, in Lemma~\ref{topology_motivating_properties}, we showed that $\ell^\mu \in \mathcal{X}_0$. Therefore, as $\mu \in C^*(\mu)$,
    \begin{align*}
        \mathbb{E}_{(w,B)} [q^\mu(w, B, \ell^\mu(w,B))] \geq \mathbb{E}_{(w,B)} [q^\mu(w, B, \mu(w,B))]
    \end{align*}

    On the other hand, by the definition of $\ell^\mu$, we get that
    \begin{align*}
        q^\mu(w, B, \ell^\mu(w,B)) \leq q^\mu(w, B, \mu(w,B)) \qquad \forall\ (w,B) \in \Theta
    \end{align*}
    Hence, combining the two statements yields $q^\mu(w, B, \ell^\mu(w,B)) = q^\mu(w, B, \mu(w,B))$ a.s. w.r.t. $(w,B) \sim G$, which completes the proof.
\end{proof}

\section{Standard Auctions and Revenue Equivalence}\label{appendix:rev-equiv}

In this section, we extend our results for first-price auctions to all anonymous standard auctions and establish revenue equivalence among them by proving Theorem \ref{standard_auctions_revenue_equivalence}. To do this, we will show that the dual of the optimization problem faced by each buyer type is identical for all anonymous standard auctions, by exploiting the structure of the Lagrangian problem and the known revenue equivalence results from the standard i.i.d. setting \citep{krishna2009auction}. This concurrence of the dual problems for all anonymous standard auctions allows us to directly apply Theorem~\ref{main_existence_result} to reduce the proof of Theorem~\ref{standard_auctions_revenue_equivalence} to showing strong duality for the optimization problem faced by the buyer types.

For buyer type $(w,B) \in \Theta$, we will use $R(w,B)$ to denote the following optimization problem:
\begin{align*}
	R^\mu(w,B) \coloneqq \max_{b:A \to \mathbb{R}_{\geq 0}} \quad &\mathbb{E}_{\alpha,\{\theta_i\}_{i=1}^{n-1}}\left[ w^T \alpha \cdot \mathds{1}\{b(\alpha) \geq \max(r(\alpha), \{\Psi^\mu(\theta_i, \alpha)\}_i) \} - M_\alpha(b(\alpha), \{\Psi^\mu(\theta_i,\alpha)\}_i)\  \right]\\
    \textrm{s.t.} \quad &\mathbb{E}_{\alpha,\{\theta_i\}_{i=1}^{n-1}}\left[M_\alpha\left(b(\alpha), \{\Psi^\mu(\theta_i, \alpha)\}_i\right) \right] \leq B	
\end{align*}

Then the dual optimization problem (or simply the dual problem) of $R^\mu(w,B)$ is given by
\begin{align*}
	\min_{t\geq 0} \max_{b:A \to \mathbb{R}_{\geq 0}}	\mathbb{E}_{\alpha,\{\theta_i\}_{i=1}^{n-1}}\left[ w^T \alpha \cdot \mathds{1}\{b(\alpha) \geq \max(r(\alpha), \{\Psi^\mu(\theta_i, \alpha)\}_i) - (1 + t)M_\alpha(b(\alpha), \{\Psi^\mu(\theta_i,\alpha)\}_i)\  \} \right] + tB
\end{align*}

The following lemma characterizes the optimal solution to the Lagrangian problem.

\begin{lemma} \label{std_auction_lagrangian_optimal}
	For all $t \geq 0$,
    \begin{align*}
        \psi^\mu_\alpha\left(\frac{w^T\alpha}{1+t}\right) \in arg\max_{b(.)}\ \mathbb{E}_{\alpha,\{\theta_i\}_{i=1}^{n-1}}\left[ \frac{w^T \alpha}{1+t} \cdot \mathds{1}(b(\alpha) \geq \max\left(r(\alpha), \{\Psi^\mu(\theta_i, \alpha)\}_i \right)) - M_\alpha(b(\alpha), \{\Psi^\mu(\theta_i,\alpha)\}_i)\ \right]
    \end{align*}
\end{lemma}
\begin{proof}
    Consider an $\alpha \in A$ such that $\lambda^\mu_\alpha$ is atom-less. Then, using the assumptions on auction $\mathcal{A}$, we can write
    \begin{align*}
        \psi^\mu_\alpha \left(\frac{w^T\alpha}{1+t}\right) \in arg \max_{t \in \mathbb{R}}\ \mathbb{E}_{\{X_i\}_{i=1}^{n-1} \sim \lambda_\alpha^\mu} \left[\left(\frac{w^T \alpha}{1+t} \cdot \mathds{1}(t \geq \max(r(\alpha), \{\psi^\mu_\alpha(X_i)\}_i) ) - M_\alpha\left(t, \{\psi^\mu_\alpha (X_i)\}_i\right)\right)\ \right]
    \end{align*}
    Combining this with the definition of $\Psi^\mu$, we get
    \begin{align*}
        \psi^\mu_\alpha\left(\frac{w^T\alpha}{1+t}\right) \in arg \max_{t\in \mathbb{R}}\ \mathbb{E}_{\{\theta_i\}_{i=1}^{n-1}}\left[ \frac{w^T \alpha}{1+t} \cdot \mathds{1}(t \geq \max (r(\alpha), \{\Psi^\mu(\theta_i, \alpha)\}_i) ) - M_\alpha(t, \{\Psi^\mu(\theta_i,\alpha)\}_i) \right]
    \end{align*}
    To complete the proof, note that $\lambda_\alpha^\mu$ is atom-less a.s. w.r.t. $\alpha$ by part (a) of Lemma \ref{pacing_strategy_properties}.
\end{proof}

We take a short interlude to state and prove a lemma which will help us simplify the expression for the dual optimization problem of $R^\mu(w,B)$.

\begin{lemma}\label{std_auction_indicator}
    For $\alpha \in A$ such that $\lambda^\mu_\alpha$ is continuous,
    \begin{align*}
        \mathds{1}\left\{\psi^\mu_\alpha(x) \geq \max(r(\alpha),  \psi^\mu_\alpha(Y))\right\} = \mathds{1}\left\{x \geq \max(r(\alpha), Y) \right\} \qquad \textrm{a.s. }  Y \sim H^\mu_\alpha,\ \forall\ x \in [0,\omega]
    \end{align*}
\end{lemma}
\begin{proof}
    As $\psi^\mu_\alpha$ is non-decreasing, $\mathds{1}\left\{\psi^\mu_\alpha(x) \geq \max(r(\alpha),  \psi^\mu_\alpha(Y))\right\} \geq \mathds{1}\left\{x \geq \max(r(\alpha), Y) \right\}$ always holds. Suppose there exists $\alpha \in A$ such that $\lambda^\mu_\alpha$ is continuous and $x \in [0,\omega]$ for which $\mathds{1}\left\{\psi^\mu_\alpha(x) \geq \max(r(\alpha),  \psi^\mu_\alpha(Y))\right\} > \mathds{1}\left\{x \geq \max(r(\alpha), Y) \right\}$ with positive probability w.r.t. $Y \sim H^\mu_\alpha$. Observe that $\psi^\mu_\alpha(x) \geq r$ implies $x \geq r$, by the assumptions made on $\psi^\mu_\alpha$. Therefore,
    \begin{align*}
        \mathds{1}\left\{\psi^\mu_\alpha(x) \geq \max(r(\alpha),  \psi^\mu_\alpha(Y))\right\} > \mathds{1}\left\{x \geq \max(r(\alpha), Y) \right\} \implies Y> x,\ x\geq r(\alpha),\ \psi^\mu_\alpha(x) \geq \psi^\mu_\alpha(Y)
    \end{align*}
    Hence, there exists $\alpha \in A$ such that $\lambda_\alpha$ is continuous and $x \in [r(\alpha),\omega]$ for which
    \begin{align*}
        H^\mu_\alpha\left(\{y \in [0,\omega] \mid y >x, \psi^\mu_\alpha(y) \leq \psi^\mu_\alpha(x)\}\right) > 0
    \end{align*}
    As $y > x$ implies $\psi^\mu_\alpha(y) \geq \psi^\mu_\alpha(x)$, we get
    \begin{align*}
        H^\mu_\alpha\left(\{y \in [0,\omega] \mid \psi^\mu_\alpha(y) = \psi^\mu_\alpha(x)\}\right) > 0
    \end{align*}
    which contradicts the assumption that $\psi^\mu_\alpha$ has a atom-less distribution. Hence, the lemma holds.
\end{proof}

Next, we proceed to prove that the dual of $R^\mu(w,B)$ is the same as the dual of the optimization problem $Q^\mu(w,B)$ associated to first-price auctions. Consider an $\alpha$ for which $\lambda^\mu_\alpha$ is continuous. Then, the expected utility $U^\mu_\alpha(x)$ of a bidder with value $x$ in auction $\mathcal{A}$, when the values of the other agents are drawn i.i.d. from $\lambda^\mu_\alpha$ and every bidder employs strategy $\psi^\mu_\alpha$, is given by
\begin{align*}
    U^\mu_\alpha(x) &\coloneqq \mathbb{E}_{\{X_i\}_{i=1}^{n-1} \sim \lambda^\mu_\alpha} \left[ x \cdot \mathds{1}\left\{\psi^\mu_\alpha(x) \geq \max(r(\alpha), \{\psi^\mu_\alpha(X_i)\}_i)\right\} - M_{\alpha}(\psi^\mu_\alpha(x), \{\psi^\mu_\alpha(X_i)\}_i)  \right]\\
    &= \mathbb{E}_{\{X_i\}_{i=1}^{n-1} \sim \lambda^\mu_\alpha} \left[x\  \mathds{1}\left\{x \geq \max(r(\alpha), \{X_i\}_i) \right\}\right] - m_\alpha(x)\\
    &= x H^\mu_\alpha(x) \mathds{1}\{x \geq r(\alpha)\} - m^\mu_\alpha(x)
\end{align*}
where $m^\mu_\alpha(x) = \mathbb{E}_{\{X_i\}_{i=1}^{n-1} \sim \lambda^\mu_\alpha} \left[M_{\alpha}(\psi^\mu_\alpha(x), \{\psi^\mu_\alpha(X_i)\}_i) \right]$ and the second equality follows from Lemma~\ref{std_auction_indicator}.

Then, from the arguments given in section 5.1.2 of Krishna, we get
\begin{align*}
    U^\mu_\alpha(x) = \int_0^x H^\mu_\alpha(s) \mathds{1}\{s \geq r(\alpha)\}  ds = \mathds{1}\{x \geq r(\alpha)\} \int_{r(\alpha)}^x H^\mu_\alpha(s)  ds
\end{align*}
which further implies
\begin{align*}
	m^\mu_\alpha(x) = x H^\mu_\alpha(x) \mathds{1}\{x \geq r(\alpha)\} - U^\mu_\alpha(x)	= \mathds{1}\{x \geq r(\alpha)\}\left(x H^\mu_\alpha(x) -  \int_{r(\alpha)}^x H^\mu_\alpha(s)  ds \right)
\end{align*}

Then, using Lemma \ref{std_auction_lagrangian_optimal} and Lemma \ref{std_auction_indicator}, the value that the objective function of the dual problem of $R^\mu(w,B)$ takes at $t\geq 0$ is given by:
\begin{align*}
	& \max_{b:A \to \mathbb{R}_{\geq 0}}\mathbb{E}_{\alpha,\{\theta_i\}_{i=1}^{n-1}}\left[ w^T \alpha \cdot \mathds{1}\{b(\alpha) \geq \max(r(\alpha), \{\Psi^\mu(\theta_i, \alpha)\}_i) \} - (1 + t)M_\alpha(b(\alpha), \{\Psi^\mu(\theta_i,\alpha)\}_i) \right] + tB\\
	=& (1 + t) \max_{b:A \to \mathbb{R}_{\geq 0}} \mathbb{E}_{\alpha,\{\theta_i\}_{i=1}^{n-1}}\left[ \frac{w^T \alpha}{1+t} \cdot \mathds{1}(b(\alpha) \geq \max\left(r(\alpha), \{\Psi^\mu(\theta_i, \alpha)\}_i \right)) - M_\alpha(b(\alpha), \{\Psi^\mu(\theta_i,\alpha)\}_i) \right] + tB\\
	=& (1 + t) \mathbb{E}_{\alpha,\{\theta_i\}_{i=1}^{n-1}}\left[ \frac{w^T \alpha}{1+t} \cdot \mathds{1}(b(\alpha) \geq \max\left(r(\alpha), \{\Psi^\mu(\theta_i, \alpha)\}_i \right)) - M_\alpha\left( \psi^\mu_\alpha\left(\frac{w^T\alpha}{1+t}\right), \{\Psi^\mu(\theta_i,\alpha)\}_i \right)) \right]+ tB\\
	 =& (1+ t) \mathbb{E}_\alpha \mathbb{E}_{\{X_i\}_{i=1}^{n-1} \sim \lambda^\mu_\alpha} \left[\frac{w^T \alpha}{1+t} \cdot \mathds{1}\left\{\psi^\mu_\alpha(x) \geq \max(r(\alpha), \{\psi^\mu_\alpha(X_i)\}_i)\right\} - M_\alpha\left( \psi^\mu_\alpha\left(\frac{w^T\alpha}{1+t}\right), \{\psi^\mu_\alpha(X_i)\}_i \right) \right] + tB\\
	 =& (1 + t) \mathbb{E}_\alpha \left[U^\mu_\alpha\left(\frac{w^T \alpha}{1+t}\right) \right] + tB\\
	 =& (1+t)\mathbb{E}_\alpha\left[\mathds{1}\left\{\frac{w^T \alpha}{1+t} \geq r(\alpha) \right\} \int_{r(\alpha)}^{\frac{w^T\alpha}{1+t}} H^\mu_\alpha(s) ds\right] + tB\\
	 =& q^\mu(w,B,t)
\end{align*}

Hence, we have shown that, for every buyer type, all anonymous standard auctions have identical dual optimization problems. In light of this, to prove  Theorem~\ref{standard_auctions_revenue_equivalence}, it suffices to prove strong duality for $R^\mu(w,B)$, where $\mu$ is a fixed-point which is guaranteed to exist by Proposition~\ref{main_fixed_point}. We give the full argument below.

\begin{proof}[Proof of Theorem~\ref{standard_auctions_revenue_equivalence}
]
	By Lemma~\ref{combined_to_individual}, we know that if $\mu \in C_0^*(\mu)$, then $\mu(w,B) \in \argmin_{t \in [0,\omega/B]} q^\mu(w,B,t)$ almost surely w.r.t. $(w,B) \sim G$. Moroeover, by part (b) of Lemma~\ref{compact_dual_space}, we have $\mu(w,B) \in \argmin_{t \in [0,\infty)} q^\mu(w,B,t)$. Consider a $\theta = (w,B) \in \Theta'$ (see Definition~\ref{definition_theta'}) for which $\mu(w,B) \in \argmin_{t \in [0,\infty)} q^\mu(w,B,t)$. Observe that such $\theta$ form a subset which has measure one under $G$. According to Theorem 5.1.5 from \cite{bertsekas1998nonlinear}, in order to prove that $\Psi^\mu(w,B, \alpha)$ (as a function of $\alpha$) is an optimal solution for the optimization problem $R^\mu(w,B)$, it suffices to show the following conditions:
    \begin{itemize}
        \item[(i)] Primal feasibility:
        $$\mathbb{E}_{\alpha,\{\theta_i\}_{i=1}^{n-1}}\left[ M_\alpha\left( \Psi^\mu(w,B,\alpha), \{\Psi^\mu(\theta_i,\alpha)\}_i \right)\right] \leq B$$
        \item[(ii)] Dual feasibility: $\mu(w,B) \geq 0$
        \item[(iii)] Lagrangian Optimality: $\Psi^\mu(w,B)$ is an optimal solution for
        	\begin{align*}
        	\max_{b:A \to \mathbb{R}_{\geq 0}}	 &\mathbb{E}_{\alpha,\{\theta_i\}_{i=1}^{n-1}}\left[ w^T \alpha \cdot \mathds{1}\{b(\alpha) \geq \max(r(\alpha), \{\Psi^\mu(\theta_i, \alpha)\}_i) \} - (1 + \mu(w,B)) M_\alpha(b(\alpha), \{\Psi^\mu(\theta_i,\alpha)\}_i) \right]\\
         	& + \mu(w,B)B
         	\end{align*}
        \item[(iv)] Complementary slackness:
        $$\mu(w,B).\left\{B - \mathbb{E}_{\alpha,\{\theta_i\}_{i=1}^{n-1}}\left[M_\alpha\left(\Psi^\mu(w,B, \alpha), \{\Psi^\mu(\theta_i, \alpha)\}_i\right) \right]\right\} = 0$$
    \end{itemize}

    First, we simplify the expression for the expected expenditure used in the sufficient conditions (i)-(iv) stated above to show that it is equal to the expected payment made by buyer type $(w,B)$ in the SFPE determined by pacing function $\mu$:
    \begin{align*}
        &\mathbb{E}_{\alpha,\{\theta_i\}_{i=1}^{n-1}}\left[ M_\alpha\left( \Psi^\mu(\theta,\alpha), \{\Psi^\mu(\theta_i,\alpha)\}_i \right) \right]\\
        = &\mathbb{E}_{\alpha} \mathbb{E}_{\{X_i\}_{i=1}^{n-1} \sim \lambda^\mu_\alpha} \left[M_{\alpha}\left(\psi^\mu_\alpha\left(\frac{w^T \alpha}{1+\mu(w,B)}\right), \{\psi^\mu_\alpha(X_i)\}_i \right)\right]\\
        = &\mathbb{E}_{\alpha} \left[m^\mu_\alpha \left(\frac{w^T \alpha}{1+\mu(w,B)}\right) \right]\\
        = &\mathbb{E}_\alpha\left[\left(\frac{w^T \alpha}{1+\mu(w,B)} H^\mu_\alpha\left(\frac{w^T \alpha}{1+\mu(w,B)}\right) -  \int_{r(\alpha)}^{\frac{w^T \alpha}{1+\mu(w,B)}} H^\mu_\alpha(s)  ds \right) \mathds{1}\left\{ \frac{w^T \alpha}{1+\mu(w,B)} \geq r(\alpha) \right\}\right]\\
        = & \mathbb{E}_\alpha\left[\sigma^\mu_\alpha\left(\frac{w^T \alpha}{1+\mu(w,B)}\right)H^\mu_\alpha\left(\frac{w^T \alpha}{1+\mu(w,B)}\right)\mathds{1}\left\{\frac{w^T \alpha}{1+\mu(w,B)} \geq r(\alpha) \right\}\right]\\
        = & \mathbb{E}_{\alpha,\{(\theta_i)\}_{i=1}^{n-1}}\left[ \beta^\mu(\theta, \alpha)\ \mathds{1}\{\beta^\mu(\theta, \alpha) \geq \max(r(\alpha), \{\beta^\mu(\theta_i, \alpha)\}_i)\}\right]
    \end{align*}
	Hence, Theorem~\ref{standard_auctions_revenue_equivalence} will follow if we establish the aforementioned sufficient conditions (i)-(iv). Note that $\mu(w,B)$ satisfies the following first order conditions of optimality
		\begin{equation}\label{std_auction_dual_optimality_conditions}
		    \frac{\partial q^\mu(w, B, \mu(w,B))}{\partial t} \geq 0 \qquad \mu(w,B) \geq 0 \qquad \mu(w,B)\cdot \frac{\partial q^\mu(w,B,\mu(w,B))}{\partial t} = 0
		\end{equation}

	Using Lemma~\ref{differentiability_dual}, we can write
		\begin{align*}
		    \frac{\partial q^\mu(w, B, \mu(w,B))}{\partial t} = B - \mathbb{E}_\alpha\left[\sigma^\mu_\alpha\left(\frac{w^T \alpha}{1+\mu(w,B)}\right)H^\mu_\alpha\left(\frac{w^T \alpha}{1+\mu(w,B)}\right)\mathds{1}\left\{\frac{w^T \alpha}{1+\mu(w,B)} \geq r(\alpha) \right\}\right]
		\end{align*}

    To establish the sufficient conditions (i)-(iv), observe that (after simplification) conditions (i), (ii) and (iv) are the same as (\ref{std_auction_dual_optimality_conditions}), and condition (iii) is a direct consequence of Lemma~\ref{std_auction_lagrangian_optimal}, thereby completing the proof of Theorem~\ref{standard_auctions_revenue_equivalence}.
\end{proof}

\rk{
\subsection{Revenue Equivalence under Ex-Ante Budget Constraints}\label{appendix:ex-ante}

The argument developed in this section also applies to the setting with non-contextual i.i.d. values and ex-ante budget constraints, which is the symmetric special case of the models studied in \citet{balseiro2015repeated} and \citet{balseiro2021budget}. More precisely, consider a single-item auction setting with $n$ buyers, and assume that the value of each buyer is drawn i.i.d. from a common atom-less distribution $\mathcal F$ over the space of all possible values $[0, \overline{V}] \subset \R_{\geq 0}$. Moreover, assume that every buyer has an ex-ante budget of $B$, i.e., she is constrained to spend at most $B$ in expectation, where the expectation is taken over her own value and the values of other buyers. Let $\A = (r, M)$ be the anonymous standard auction with reserve price $r$ and payment rule $M$ that the seller uses to sell the item.

In this simpler setting, a strategy $\beta^*: [0, \overline V] \to \R$ is a symmetric equilibrium if $\beta^*$ is the optimal bidding strategy for a buyer when all other buyers employ $\beta^*$ to bid. Concretely, $\beta^*: [\underline{V}, \overline V] \to \R$ is a symmetric equilibrium if it is an optimal solution to the following optimization problem:
\begin{align}\label{eqn:ex-ante-eq}
	\max_{b: [0, \overline V] \to \R_{\geq 0}} &\E_{v, \{v_i\}_{i=1}^{n-1}} \left[ v \cdot \mathds{1}\{b(v) \geq \max(r, \{\beta^*(v_i)\}_i)\} - M(b(v), \{\beta^*(v_i)\}_i) \right]\\
	\text{s.t.} \quad 	&\E_{v, \{v_i\}_{i=1}^{n-1}} \left[ M(b(v), \{\beta^*(v_i)\}_i) \right] \leq B \notag
\end{align}

When $\A$ is a second-price auction, the results of both \citet{balseiro2015repeated} and \citet{balseiro2017budget} imply that
strong duality holds for the optimization problem given in \eqref{eqn:ex-ante-eq}, and there exists a dual solution $\mu^* \geq 0$ such that $\beta^*(v) = v/(1 + \mu^*)$ is a symmetric equilibrium. With this existence result for second-price auctions in hand, we can leverage the argument developed earlier to establish the existence of a value-pacing-based equilibrium for all standard auctions and revenue equivalence.

Let $\mathcal{H}$ be the distribution of $v/(1 + \mu^*)$ when $v \sim \mathcal F$ and $\psi^{\mathcal H}$ be the single-auction equilibrium for distribution $\mathcal H$ and auction $\A = (r, M)$, as defined at the beginning of Section~\ref{sec:standard}. Then, we claim that the value-pacing-based strategy given by
\begin{align*}
	\Psi(v) = \psi^{\mathcal H} \left( \frac{v}{1 + \mu^*} \right)	
\end{align*}
is a symmetric equilibrium (as defined in equation~\eqref{eqn:ex-ante-eq}). To see this, first observe that, when all of the other buyers use $\beta^* = \Psi$ to bid, the dual of the optimization problem \eqref{eqn:ex-ante-eq} is given by
\begin{align*}
	&\min_{\mu \geq 0} 	\max_{b: [0, \overline V] \to \R_{\geq 0}} \E_{v, \{v_i\}_{i=1}^{n-1}} \left[ v \cdot \mathds{1}\{b(v) \geq \max(r, \{\Psi(v_i)\}_i)\} - (1 + \mu) M(b(v), \{\Psi(v_i)\}_i) \right] + \mu \cdot B\\
	= & \min_{\mu \geq 0}\ (1 + \mu)\ \E_{v} \left[ 	\max_{b \in \R_{\geq 0}} \E_{\{v_i\}_{i=1}^{n-1}} \left[ \frac{v}{1 + \mu} \cdot \mathds{1}\{b \geq \max(r, \{\Psi(v_i)\}_i) \} - M(b, \{\Psi(v_i)\}_i) \right] \right] + \mu \cdot B
\end{align*}

Next, observe that the inner optimization problem over $b \in \R_{\geq 0}$ is exactly the bidding problem faced by a buyer with value $v/(1 + \mu^*)$ who aims to maximize her utility in the single-auction setting when the values of the other buyers are drawn from the distribution $\mathcal H$. Since $\psi^{\mathcal H}(\cdot)$ is the equilibrium strategy in the single-auction setting, $\Psi(v) = \psi^{\mathcal H}(v/(1 + \mu^*))$ is an optimal solution to this bidding problem. Moreover, we know from \citet{myerson1981optimal} that the interim expected utility of a buyer under equilibrium strategies is independent of payment rule of the standard auction. Hence,  the dual optimization problem is the same for all standard auctions. In particular, $\mu^*$ is an optimal solution for this common dual problem. Finally, using a proof similar to the one we provide for Theorem~\ref{standard_auctions_revenue_equivalence} in Appendix~\ref{appendix:rev-equiv}, it is possible to show that strong duality holds for the optimization problem stated in \eqref{eqn:ex-ante-eq} when $\beta^* = \Psi$ and $\Psi(v/(1 + \mu^*))$ is an optimal solution of \eqref{eqn:ex-ante-eq} as required.

}

\section{Worst-Case Efficiency Guarantees}\label{appendix:poa}

\rk{
The following example demonstrates that the Price of Anarchy of social welfare can be arbitrarily small for value-pacing-based equilibria.

\begin{example*}
	Fix the number of buyers to $n = 2$ and consider the second-price auction format. Let the distribution of feature vectors $F$ be the uniform distribution over $A = [1,2] \times [1,2]$. Moreover, assume that the buyer weight vectors are distributed uniformly over $[1,2] \times [1,2] \cup [y^4, y^4+1/y] \times [y^4, y^4+1/y]$ for some large $y \geq 1$. Also, suppose the budget of all buyer types with weight vector $w \in [1,2] \times [1,2]$ is $10$ and the budget of all buyer types with weight vector $w \in [y^4, y^4+1/y] \times [y^4, y^4+1/y]$ is $1/y^2$. By Theorem~\ref{main_existence_result}, we get that there exists a value-pacing-based equilibrium for this instance. Let $\mu$ be the pacing function associated with it and $x^\mu$ be the corresponding allocation. First, observe that all of the buyer types with weight vectors in $[1,2] \times  [1,2]$ are not paced in equilibrium and bid their value on each item, i.e., $\mu(w, 10) = 1$ for all $w \in [1,2] \times [1,2]$. This is because their budget far exceeds their expected value: even if they win every item, their payments is as most 8, which is smaller than their budget of 10. Next, consider a buyer $i \in \{1,2\}$ with type $\theta_i = (w,1/y^2)$ for some $w \in [y^4, y^4+1/y] \times [y^4, y^4+1/y]$. Then, her expected payment (expectation over competing buyer type and item type) is at least
	\begin{align*}
		P(w_{-i} \in [1,2]^2) \cdot \E_{\alpha, \theta_{-i}}[x^\mu_i(\alpha, \theta_i, \theta_{-i}) \cdot 1 \mid w_{-i} \in [1,2]^2] = \frac{1}{1 + y^{-2}} \cdot \E_\alpha[x^\mu_i(\alpha, \theta_i, \theta_{-i})]
	\end{align*}
	because, when $w_{-i} \in [1,2] \times [1,2]$, buyer ${-i}$ bids her value on each item and her value is always at least 1. Moreover, the budget of the buyer with type $\theta_i$ is $1/y^2$. Therefore, we get
	\begin{align*}
		\frac{1}{1 + y^{-2}} \cdot \E_\alpha[x^\mu_i(\alpha, \theta_i, \theta_{-i})] \leq \frac{1}{y^2}	
	\end{align*}
Let $x$ be the allocation that always gives the item to a buyer with weight vector $w \in [y^4, y^4+1/y] \times [y^4, y^4+1/y]$ when such a buyer type is present. We partition the space of buyer-type profiles into 4 regions, and bound the expected social welfare (expectation taken only over $\alpha \sim F$) of $x^\mu$ and $x$:
	\begin{enumerate}
		\item $\theta_i = (w_i, 1)$ with $w_i \in [1,2] \times [1,2]$ for both buyers $i \in \{1,2\}$. This occurs with probability at most 1 and the expected social welfare under $x^\mu$ when $\alpha \sim F$ is bounded above by 8 for each type profile in this region.
		\item $\theta_i = (w_i, 1/y^2)$ with $w_i \in [y^4, y^4+1/y] \times [y^4, y^4+1/y]$ for both buyers $i \in \{1,2\}$. This occurs with probability at most $1/y^4$ and the expected social welfare under $x^\mu$ when $\alpha \sim F$ is bounded above by $8 y^4$ for each type profile in this region. 
		\item $\theta_1 = (w_1, 1/y^2)$ with $w_1 \in [y^4, y^4+1/y] \times [y^4, y^4+1/y]$ and $\theta_2 = (w_2, 10)$ with $w_2 \in [1,2] \times [1,2]$. This occurs with probability $y^{-2}/(1 + y^{-2})$. As we argued earlier, $\E_\alpha[x^\mu_1(\alpha, \theta_1, \theta_{2})] \leq (1 + y^{-2})/y^2$ in this case. Therefore, the expected social welfare under $x^\mu$ when $\alpha \sim F$ is bounded above by $8y^4 \cdot \{(1 + y^{-2})/y^2\} + 8 \leq 24y^2$. On the other hand, the expected social welfare under $x$ when $\alpha \sim F$ is at least $y^4$ in this region since buyer 1 always gets the item. 
		\item $\theta_2 = (w_2, 1/y^2)$ with $w_2 \in [y^4, y^4+1/y] \times [y^4, y^4+1/y]$ and $\theta_1 = (w_1, 10)$ with $w_1 \in [1,2] \times [1,2]$. This is the same as region 3 with the roles of buyer 1 and buyer 2 interchanged.
	\end{enumerate}
	
	Combining the bounds for the different regions, we get that the total expected social welfare under $x^\mu$ is bounded above by
	\begin{align*}
		8 + 8y^4 \cdot \frac{1}{y^4} + 24y^2 \cdot \frac{y^{-2}}{1+y^{-2}} + 24y^2 \cdot \frac{y^{-2}}{1+y^{-2}} \leq 64 \ ,
	\end{align*}
	and the total expected social welfare under $x$ is bounded below by
	\begin{align*}
		0 + 0 + y^4 \cdot \frac{y^{-2}}{1+y^{-2}} + y^4 \cdot \frac{y^{-2}}{1+y^{-2}} \geq y^{2}.
	\end{align*}
	Hence, the Price of Anarchy of social welfare is at most $64/y^2$, which tends to zero as $y \to \infty$.
	
\end{example*}
}

\begin{proof}[Proof of Theorem~\ref{thm:poa}]
	We will focus on second-price auctions. Consider an allocation $x$, an equilibrium pacing function $\mu$ with $\mu \in C^*_0(\mu)$ and the associated allocation $x^\mu$. Since $x$ and $\mu$ are arbitrary, it suffices to show that $\operatorname{LW}(x^\mu) \geq \operatorname{LW}(x)/2$.
	
	Let $p(\alpha, \vec{\theta})$ denote the second-highest bid on item $\alpha$ in the equilibrium parameterized by $\mu$ when the buyer-type profile is given by $\vec{\theta}$, i.e., it is the second largest element in the set $\{w_i^T\alpha/(1 + \mu(w_i, B_i)) \mid i \in [n]\}$. The following lemma is a key step in the proof of the theorem.
	
	\begin{lemma}\label{lemma:poa-inter}
		For all $i \in [n]$ and $\theta_i \in \Theta$, we have
		\begin{align*}
		&\min\left\{\E_{\alpha, \theta_{-i}}[w_i^T\alpha \cdot x^{\mu}_i(\alpha, \theta_i, \theta_{-i})], B_i \right\}\\&\quad\geq \min\left\{\E_{\alpha, \theta_{-i}}[w_i^T\alpha \cdot x_i(\alpha, \theta_i, \theta_{-i})], B_i \right\}  - \E_{\alpha, \theta_{-i}}[ p(\alpha, \theta_i, \theta_{-i}) \cdot x_i(\alpha, \theta_i, \theta_{-i})]\,.
	\end{align*}
	\end{lemma}
	
	\begin{proof}
		Fix $i \in [n]$ and $\theta_i \in \Theta$. We will prove the lemma separately for paced and unpaced buyer types. First, consider the case when $\theta_i$ is paced in equilibrium, i.e., $\mu(\theta_i) > 0$. Then, since $\mu(w_i, B_i) \in \argmin_{t \geq 0} q^\mu(w_i, B_i, t)$, complementary slackness (see proof of Theorem~\ref{standard_auctions_revenue_equivalence}) implies that:
	\begin{align*}
		\E_{\alpha, \theta_{-i}} \left[ p(\alpha, \theta_i, \theta_{-i}) \cdot x_i^\mu(\alpha, \theta_i, \theta_{-i}) \right] = B_i\,.	
	\end{align*}
	Moreover, note that $w_i^T\alpha/(1 + \mu(w_i, B_i)) \geq p(\alpha, \theta_i, \theta_{-i})$ whenever $x^{\mu}_i(\alpha, \theta_i, \theta_{-i}) > 0$ because only the highest bidder(s) win the item in a second-price auction. This allows us to establish the lemma for paced buyers:
	\begin{align*}
		&\min\left\{\E_{\alpha, \theta_{-i}}[w_i^T\alpha \cdot x^{\mu}_i(\alpha, \theta_i, \theta_{-i})], B_i \right\}\\
        &\quad\geq 	\min\left\{\E_{\alpha, \theta_{-i}} \left[ p(\alpha, \theta_i, \theta_{-i}) \cdot x_i^\mu(\alpha, \theta_i, \theta_{-i}) \right], B_i \right\}\\
		&\quad= B_i\\
		&\quad\geq \min\left\{\E_{\alpha, \theta_{-i}}[w_i^T\alpha \cdot x_i(\alpha, \theta_i, \theta_{-i})], B_i \right\}\\
		&\quad\geq \min\left\{\E_{\alpha, \theta_{-i}}[w_i^T\alpha \cdot x_i(\alpha, \theta_i, \theta_{-i})], B_i \right\} - \E_{\alpha, \theta_{-i}}[ p(\alpha, \theta_i, \theta_{-i}) \cdot x_i(\alpha, \theta_i, \theta_{-i})]\,,
	\end{align*}
where the first inequality follows because $w_i^T\alpha \ge w_i^T\alpha/(1 + \mu(w_i, B_i))$ since $\mu(w_i, B_i) \ge 0$, the first equality because budgets binds, the second inequality because $B_i \ge \min(a, B_i)$ for every $a \in \mathbb R$, and the last inequality because payments are non-negative.

	Next, consider the case when $\theta_i$ is unpaced in equilibrium, i.e., $\mu(\theta_i) = 0$. Then, by definition of a pacing-based strategy for second-price auctions, buyer type $\theta_i$ bids her value $w_i^T\alpha$ on item $\alpha$ in equilibrium, for all items $\alpha \in A$. As a consequence, if $x_i^\mu(\alpha, \theta_i, \theta_{-i}) < 1$, then we have $w_i^T\alpha \leq p(\alpha, \theta_i, \theta_{-i})$. In other words,
	\begin{align}\label{eqn:poa-inter-1}
		\E_{\alpha, \theta_{-i}}[(w_i^T\alpha - p(\alpha, \theta_i, \theta_{-i})) \cdot (1 - x^{\mu}_i(\alpha, \theta_i, \theta_{-i})) \cdot x_i(\alpha, \theta_i, \theta_{-i})] \leq 0\,.
	\end{align}
	Moreover, observe that
	\begin{align}\label{eqn:poa-inter-2}
		\E_{\alpha, \theta_{-i}}[w_i^T\alpha \cdot x^{\mu}_i(\alpha, \theta_i, \theta_{-i})] \geq \E_{\alpha, \theta_{-i}}[(w_i^T\alpha - p(\alpha, \theta_i, \theta_{-i})) \cdot x^{\mu}_i(\alpha, \theta_i, \theta_{-i}) \cdot x_i(\alpha, \theta_i, \theta_{-i})]\,,
	\end{align}
because payments are non-negative and $x_i \in [0,1]$. 	Combining \eqref{eqn:poa-inter-1} and \eqref{eqn:poa-inter-2} yields
	\begin{align}\label{eqn:poa-inter-2.1}
		&\E_{\alpha, \theta_{-i}}[w_i^T\alpha \cdot x^{\mu}_i(\alpha, \theta_i, \theta_{-i})] \nonumber\\
 &\quad \geq \E_{\alpha, \theta_{-i}}[(w_i^T\alpha - p(\alpha, \theta_i, \theta_{-i})) \cdot x_i(\alpha, \theta_i, \theta_{-i})] \nonumber\\
		&\quad\geq \min\left\{\E_{\alpha, \theta_{-i}}[w_i^T\alpha \cdot x_i(\alpha, \theta_i, \theta_{-i})], B_i \right\} - \E_{\alpha, \theta_{-i}}[ p(\alpha, \theta_i, \theta_{-i}) \cdot x_i(\alpha, \theta_i, \theta_{-i})]\,,
	\end{align}
where the last inequality follows because $a \ge \min(a, B_i)$ for every $a \in \mathbb R$. Furthermore, note the trivial inequality
	\begin{align}\label{eqn:poa-inter-2.2}
		B_i &\geq \min\left\{\E_{\alpha, \theta_{-i}}[w_i^T\alpha \cdot x_i(\alpha, \theta_i, \theta_{-i})], B_i \right\} \nonumber\\
&\ge\min\left\{\E_{\alpha, \theta_{-i}}[w_i^T\alpha \cdot x_i(\alpha, \theta_i, \theta_{-i})], B_i \right\} - \E_{\alpha, \theta_{-i}}[ p(\alpha, \theta_i, \theta_{-i}) \cdot x_i(\alpha, \theta_i, \theta_{-i})]\,,
	\end{align}
where we used again that $B_i \ge \min(a, B_i)$ for every $a \in \mathbb R$ and that payments are non-negative.
	Finally, combining \eqref{eqn:poa-inter-2.1} and \eqref{eqn:poa-inter-2.2} yields the lemma for unpaced buyers
	\begin{align}\label{eqn:poa-inter-3}
		&\min\left\{\E_{\alpha, \theta_{-i}}[w_i^T\alpha \cdot x^{\mu}_i(\alpha, \theta_i, \theta_{-i})], B_i \right\} \nonumber \\ \geq &\quad\min\left\{\E_{\alpha, \theta_{-i}}[w_i^T\alpha \cdot x_i(\alpha, \theta_i, \theta_{-i})], B_i \right\} - \E_{\alpha, \theta_{-i}}[ p(\alpha, \theta_i, \theta_{-i}) \cdot x_i(\alpha, \theta_i, \theta_{-i})],
	\end{align}
	since \eqref{eqn:poa-inter-2.1} and \eqref{eqn:poa-inter-2.2} show the inequality separately for each of the two terms in the minimum on the left-hand side of \eqref{eqn:poa-inter-3}.
	This establishes the lemma for all $i \in [n]$ and $\theta_i \in \Theta$.
	\end{proof}

Continuing the proof of Theorem~\ref{thm:poa},
next, we sum over $i \in [n]$ and take expectation w.r.t.~$\theta_i$ for the inequality in Lemma~\ref{lemma:poa-inter}. First, we study the effect of summing and taking expectations on the second term in the RHS. We have
	\begin{align}\label{eqn:poa-inter-4}
		\sum_{i=1}^n \E_{\theta_i} \left[ \E_{\alpha, \theta_{-i}}[ p(\alpha, \theta_i, \theta_{-i}) \cdot x_i(\alpha, \theta_i, \theta_{-i})]	 \right] &= \E_{\alpha, \vec{\theta}} \left[ p(\alpha, \vec{\theta}) \cdot \sum_{i=1}^n x_i(\alpha, \theta_i, \theta_{-i})\right] \nonumber \\
		&= \E_{\alpha, \vec{\theta}} \left[ p(\alpha, \vec{\theta})\right] \nonumber \\
		&= \E_{\alpha, \vec{\theta}} \left[ p(\alpha, \vec{\theta}) \cdot \sum_{i=1}^n x^\mu_i(\alpha, \theta_i, \theta_{-i})\right] \nonumber\\
		&= \sum_{i=1}^n \E_{\theta_i} \left[ \E_{\alpha, \theta_{-i}}[ p(\alpha, \theta_i, \theta_{-i}) \cdot x^\mu_i(\alpha, \theta_i, \theta_{-i})]	 \right] \nonumber\\
		&\leq \sum_{i=1}^n \E_{\theta_i} \left[ \min\left\{\E_{\alpha, \theta_{-i}}[w_i^T\alpha \cdot x_i^\mu(\alpha, \theta_i, \theta_{-i})], B_i \right\} \right]\nonumber\\
&=\operatorname{LW}(x^\mu)\,,
	\end{align}
where the first and fourth equalities follow from Fubini's theorem, the second and third because allocations sum up to one (i.e., there no reserve prices), and the last inequality follows from the budget-feasibility of the pacing-based equilibrium strategy given by $\mu$ for buyer type $\theta_i$, which implies $$\E_{\alpha, \theta_{-i}}[ p(\alpha, \theta_i, \theta_{-i}) \cdot x^\mu_i(\alpha, \theta_i, \theta_{-i})] \leq B_i$$
		and the winning criteria of second-price auctions, which implies $$p(\alpha, \theta_i, \theta_{-i}) \leq \frac{w_i^T\alpha}{1 + \mu(w_i, B_i)} \leq w_i^T\alpha$$ whenever $x_i^\mu(\alpha, \theta_i, \theta_{-i}) > 0$.

Using \eqref{eqn:poa-inter-4}, we obtain by summing over $i \in [n]$ and integrate over $\theta_i$ the inequality in the statement of Lemma~\ref{lemma:poa-inter}:
	\begin{align*}
		\operatorname{LW}(x^\mu)&=\sum_{i=1}^n \E_{\theta_i} \left[ \min\left\{\E_{\alpha, \theta_{-i}}[w_i^T\alpha \cdot x^{\mu}_i(\alpha, \theta_i, \theta_{-i})], B_i \right\} \right]\\
		&\geq   \sum_{i=1}^n \E_{\theta_i} \left[ \min\left\{\E_{\alpha, \theta_{-i}}[w_i^T\alpha \cdot x_i(\alpha, \theta_i, \theta_{-i})], B_i \right\} \right] - \sum_{i=1}^n \E_{\theta_i} \left[ \E_{\alpha, \theta_{-i}}[ p(\alpha, \theta_i, \theta_{-i}) \cdot x_i(\alpha, \theta_i, \theta_{-i})]	 \right]\\
&= \operatorname{LW}(x) - \sum_{i=1}^n \E_{\theta_i} \left[ \E_{\alpha, \theta_{-i}}[ p(\alpha, \theta_i, \theta_{-i}) \cdot x_i(\alpha, \theta_i, \theta_{-i})]	 \right]\\
		&\geq \operatorname{LW}(x) - \operatorname{LW}(x^\mu)\,.
	\end{align*}
	Therefore, we have shown that $\operatorname{LW}(x^\mu) \ge \operatorname{LW}(x)/2$	as required.
\end{proof}

\section{Structural Properties}\label{appendix:structural-properties}

Before proceeding with the proof of Proposition~\ref{structural_prop}, we establish the following Lemma, which is informative in its own right.

\begin{lemma}\label{buyer_type_continuity}
	The pacing function $\mu: \Theta \to [0, \omega/ B_{\min}]$ is continuous.
\end{lemma}

\begin{proof}
	We start by observing that the following function is continuous for all $\alpha \in A$:
    \begin{align*}
        (w,B, t) \mapsto \int_{0}^{\frac{w^T\alpha}{1+t}} H_\alpha^\mu(s) ds
    \end{align*}

    Therefore, Dominated Convergence Theorem implies $(w, B, t) \mapsto q^\mu(w, B, t)$ is continuous. Finally, applying Berge Maximum Theorem (Theorem 17.31 of \cite{aliprantis2006infinite}) yields the continuity of $(w,B) \mapsto \mu(w,B)$ because of our assumption that $\mu(w,B)$ is the unique minimizer of $q^\mu(w, B, t)$.	
\end{proof}

We now state the proof of Proposition~\ref{structural_prop}.

\begin{proof}[Proof of Proposition~\ref{structural_prop}]
	Consider a unit vector $\hat{w}\in \R_+^d$ and budget $B > 0$ such that $w/\|w\| = \w$, for some $(w,B) \in \delta(X)$. If $\mu(w,B) = 0$ for all buyers $(w,B) \in \delta(X)$ with $w/\|w\| = \hat{w}$, then the theorem statement holds trivially. So assume that there exists $x >0$ such that $x\hat{w} \in \delta(X)$ and $\mu(x\hat{w}, B) > 0$. Define $x_0 \coloneqq \inf\{x \in (0, \infty) \mid (x \w, B) \in \delta(X);\ \mu(x \w, B) > 0\}$. Then, as a consequence of the complementary slackness condition established in Proposition~\ref{optimal_solution}, for $x > x_0$, we have
		\begin{align*}
			E_\alpha \left[ \sigma^\mu_\alpha\left(\frac{x\w^T\alpha}{1 + \mu(x \w, B)}\right) H_\alpha^\mu 	\left(\frac{x\w^T\alpha}{1 + \mu(x \w, B)}\right) \right] = B.
		\end{align*}
		Recall that, in Lemma~\ref{pacing_strategy_properties}, we established the continuity of $\sigma_\alpha^\mu$ and $H_\alpha^\mu$ almost surely w.r.t. $\alpha \sim F$. Combining this with the continuity of $\mu$ established in Lemma~\ref{buyer_type_continuity}, we can apply the Dominated Convergence Theorem to establish
		\begin{align*}
			E_\alpha \left[ \sigma^\mu_\alpha\left(\frac{x_0\w^T\alpha}{1 + \mu(x_0 \w, B)}\right) H_\alpha^\mu 	\left(\frac{x_0\w^T\alpha}{1 + \mu(x_0 \w, B)}\right) \right] = B.
		\end{align*}
		As $B> 0$, we get $x_0 > 0$. Next, observe that if $t^* \geq 0$ satisfies $x_0(1 + t^*) = x(1 + \mu(x_0\w, B)$, then
		\begin{align*}
			\frac{\partial q^\mu(w, B, t^*)}{\partial t} = B - E_\alpha \left[ \sigma^\mu_\alpha\left(\frac{x\w^T\alpha}{1 + t^*}\right) H_\alpha^\mu 	\left(\frac{x\w^T\alpha}{1 + t^*}\right) \right] = 0
		\end{align*}
		Therefore, by our uniqueness assumption on $\mu$, we get $1 + \mu(x\w, B) = (x/x_0) (1 + \mu(x_0 \w, B))$ for all $x \geq x_0$. Hence, for all $x \geq x_0$, we get
		\begin{align*}
			\frac{x\w^T\alpha}{1 + \mu(x\w,B)} = \frac{x_0\w^T\alpha}{1 + \mu(x_0\w,B)}
		\end{align*}
		Part (1) of Proposition~\ref{structural_prop} follows directly. Part (2) considers the case when there exists $y \geq 0$ such that $(y\w,B) \in \delta(X)$ and $\mu(y\w, B) =0$. In this case, Lemma~\ref{buyer_type_continuity} and the connectedness of $\delta(X)$ imply that $\mu(x_0\hat{w}, B) = 0$, with part (2) of Proposition~\ref{structural_prop} following as a direct consequence.
\end{proof}

Next, we state the proof of Proposition~\ref{thm:targeting}.

\begin{proof}[Proof of Proposition~\ref{thm:targeting}]
	First, note that
	\begin{align*}
		q^\mu(w, B, \mu(w,B)) &= \min_{t \geq 0} q^\mu(w, B, t)\\
		&= \min_{t \geq 0}\ (1+t)\ \mathbb{E}_\alpha\left[ \int_0^{\frac{w^T\alpha}{1+t}} H^\mu_\alpha(s) ds\right] + tB.
	\end{align*}
	Next, define $g: \Theta \times A \times \R_{\geq 0} \to \R$ as
	\begin{align*}
		g(w, B, \alpha, t) = (1 + t)\int_0^{\frac{w^T\alpha}{1+t}} H^\mu_\alpha(s) ds + tB.
	\end{align*}
	Since $H^\mu$ is continuous (Lemma~\ref{pacing_strategy_properties}), we get that $g$ is differentiable w.r.t. $w$ and the derivative satisfies
	\begin{align*}
		\| \nabla_w g(w,B, \alpha, t) \|  =\ \biggr\| \alpha \cdot H_\alpha^\mu\left( \frac{w^T\alpha}{1+t} \right) \biggr\| \leq \max_{\alpha \in A}\ \|\alpha\|.	
	\end{align*}
	Therefore, dominated convergence theorem implies that
	\begin{align*}
		\nabla_w \mathbb{E}_\alpha[g(w, B , \alpha, t)] = \mathbb{E}_\alpha \left[ \nabla_w g(w,B, \alpha, t) \right] = \mathbb{E}_\alpha \left[ \alpha \cdot H_\alpha^\mu\left( \frac{w^T\alpha}{1+t} \right) \right].
	\end{align*}
	Note that the dual function $q^\mu$ is convex in $t$ and at least one dual optimal solution always lies in the compact set $[0, \omega/B_{\min}]$ (Lemma~\ref{compact_dual_space}). Moreover, if $t_1, t_2 \in \argmin_{t \geq 0} q^\mu(w, B, t)$, then the optimality conditions discussed in the proof of Proposition~\ref{optimal_solution} imply
	\begin{align*}
		\frac{\partial q^\mu(w,B,t_1)}{\partial t} = \frac{\partial q^\mu(w,B,t_2)}{\partial t} = 0.
	\end{align*}
	Without loss of generality, assume $ t_1 < t_2$. Moreover, suppose
	\begin{align*}
		H_\alpha^\mu\left( \frac{w^T\alpha}{1+t_1} \right) > H_\alpha^\mu\left( \frac{w^T\alpha}{1+t_2} \right).
	\end{align*}
	In the proof of Lemma~\ref{pacing_strategy_properties}, we showed that the above equation implies
	\begin{align*}
		\sigma_\alpha^\mu\left( \frac{w^T\alpha}{1+t_1} \right) H_\alpha^\mu\left( \frac{w^T\alpha}{1+t_1} \right) > \sigma_\alpha^\mu\left( \frac{w^T\alpha}{1+t_2} \right)	H_\alpha^\mu\left( \frac{w^T\alpha}{1+t_2} \right) \quad \forall\ \alpha \in A.
	\end{align*}
	This contradicts Lemma~\ref{differentiability_dual} because
	\begin{align*}
		\frac{\partial q^\mu(w,B,t)}{\partial t} = \mathbb{E}_\alpha\left[  \sigma_\alpha^\mu\left( \frac{w^T\alpha}{1+t} \right) H_\alpha^\mu\left( \frac{w^T\alpha}{1+t} \right) \right]
	\end{align*}
	Hence, we have shown that
	\begin{align*}
		H_\alpha^\mu\left( \frac{w^T\alpha}{1+t_1} \right) = H_\alpha^\mu\left( \frac{w^T\alpha}{1+t_2} \right) \quad \forall\ t_1, t_2 \in \argmin_{t \geq 0} q^\mu(w, B, t)
	\end{align*}
		
	This allows us to invoke Danskin's Theorem, which yields
	\begin{align*}
		\nabla_w q^\mu(w, B, \mu(w,B)) = B - \nabla_w \mathbb{E}_\alpha[g(w, B , \alpha, t)] \biggr|_{t = \mu(w,B)} = \mathbb{E}_\alpha \left[ \alpha \cdot H_\alpha^\mu\left( \frac{w^T\alpha}{1+ \mu(w,B)} \right) \right],
	\end{align*}
	thereby completing the proof.
\end{proof}

\section{Analytical and Numerical Examples}\label{appendix_example_numerics}

\begin{proof}[Proof of Claim~\ref{example}]
	Note that $w/(1 + \mu(w,B)) = w/\|w\|$ for all $(w,B) \in \Theta$. Therefore, 	$w/(1 + \mu(w,B))$ is distributed uniformly on the unit ring restricted to the positive quadrant $\{(x,y) \in \R_{\geq 0}^2 \mid x^2 + y^2 = 1\}$. Hence,
	\begin{align*}
		H^\mu_\alpha(s) = P_{(w,B)}\left(\frac{w^T\alpha}{1 + \mu(w,B)} \leq s\right) =  \frac{\arcsin(s)}{\pi/2} \qquad \text{for } \alpha \in A = \{e_1, e_2\}
	\end{align*}
	Observe that $H_\alpha^\mu$ is continuous for all $\alpha \in A$. This implies that, for all $(w,B) \in \Theta$, strong duality holds for the optimization problem $Q^\mu(w, B)$, because the proof of the results given in Section \ref{subsection_strong_duality} only relied on continuity of $H_\alpha^\mu$. Therefore, to prove the claim, it suffices to show that each buyer $(w,B)$ exactly spends her budget. The total payment made by buyer $(w,B) \in \Theta$, when everyone uses $\beta^\mu$, is given by
	\begin{align*}
		\mathbb{E}_\alpha\left[ \tilde{\beta}_\alpha^\mu \left(\frac{w^T \alpha}{1+\mu(w,B)}\right) H_\alpha^\mu \left(\frac{w^T \alpha}{1+\mu(w,B)}\right) \right]
		&= \sum_{i=1}^2 \frac{1}{2} \left[ \w_i H_{e_i}^\mu(\w_i)- \int_0^{\w_i} H_{e_i}^\mu(s) ds \right]\\
		&= \frac{2 - \w_1 - \w_2}{\pi}\\
		&= \frac{2\|w\| - w_1 - w_2}{\pi \|w\|}\,.
	\end{align*}
	Hence, the claim holds.
\end{proof}
\section{Extension to Non-linear Response Functions} \label{appendix:non_linear}

In this section, we discuss extensions of our results beyond linear valuation functions. Let $f: \R \to \R$ be a (potentially non-linear) monotonically increasing function. We assume that the value a buyer with weight vector $w$ has for item with feature vector $\alpha$ is given by $f(w^T\alpha)$. Moreover, we  relax the assumption that $\Theta, A \subset \R_+$ and only require that $f(w^T\alpha)$ is non-negative for all $w \in \Theta, \alpha \in A$. For example, the logistic function $f(t) = e^t/ (1 + e^t)$ is a non-linear increasing response function commonly used in practice which satisfies the above assumptions. Moreover, the linear function $f(t) = t$ yields our original linear model. Before proceeding further, we appropriately modify the terms defined earlier to accommodate this more general valuation model given by $f$.

Consider a pacing function $\mu: \Theta \to \mathbb{R}_{\geq 0}$. We define the \emph{paced value} of a buyer type $(w,B)$ for item $\alpha$ as $f(w^T\alpha)/(1 + \mu(w,B))$. For item $\alpha \in A$, let $\lambda_\alpha^\mu$ denote the distribution of paced values $f(w^T \alpha)/(1 + \mu(w,B))$ when $(w,B) \sim G$. Let $H_\alpha^\mu$ denote the distribution of the highest value $Y:=\max\{X_1,\dots, X_{n-1}\}$ among $n-1$ buyers, when each $X_i \sim \lambda_\alpha^\mu$ is drawn independently for $i \in \{1,\ldots,n-1\}$. Observe that $H_\alpha^\mu( (-\infty, x]) = \lambda_\alpha^\mu((-\infty, x])^{n-1}$ for all $\alpha \in A$ because the random variables are i.i.d.

To better understand how our results can be extended to this more general valuation model, it is important to understand how the linearity assumption was employed in our derivations. A careful analysis of the derivations would reveal that the linearity was only employed exactly once, and that was to prove part (a) of Lemma~\ref{pacing_strategy_properties}. In the following lemma, we prove the analogue of part (a) of Lemma~\ref{pacing_strategy_properties}. The analysis for the rest of our results remains the same for this more general non-linear valuation model.

\begin{lemma}\label{lemma:non_linear}
	$\lambda^\mu_\alpha$ and $H^\mu_\alpha$ (as defined above) have a continuous CDF for every pacing function $\mu: \Theta \to \mathbb{R}_{\geq 0}$.
\end{lemma}
\begin{proof}
	 Consider a pacing function $\mu: \Theta \to \mathbb{R}_{\geq 0}$. Let $\alpha_1, \alpha_2 \in A$ be linearly independent feature vectors and $x_1, x_2 \in [0,\omega]$ be two possible item values. We consider the set of buyer types which have paced value $x_1$ for $\alpha_1$ and paced value $x_2$ for $\alpha_2$. Define
	\begin{align*}
		S \coloneqq \left\{ (w,B) \in \Theta \bigg\lvert\ \frac{f(w^T\alpha_1)}{1 + \mu(w,B)} = x_1;\ \frac{f(w^T\alpha_2)}{1 + \mu(w,B)} = x_2 \right\}
	\end{align*}
	Observe that, for $(w,B) \in S$ and $c \coloneqq x_1/x_2$, we have $f(w^T\alpha_1) = c f(w^T\alpha_2)$. Therefore, the set $T = \{w \mid f(w^T\alpha_1) = c f(w^T\alpha_2)\}$ is a superset of the set $S_w$. Next, define $T(s) = \{w \mid w^T\alpha_1 = f^{-1}(cf(s));\ w^T\alpha_2 = s\}$. Then, it immediately follows that $T = \cup_{s: f(s) \geq 0} T(s)$.
	Due to their linear independence, we can find a basis that contains $\alpha_1, \alpha_2$, call it $\{\alpha_1, \alpha_2, \dots, \alpha_n\}$. Let $M$ be the invertible matrix whose rows are given by $\alpha_1, \alpha_2, \dots, \alpha_n$. Now, note that the set $U \coloneqq \{ (f^{-1}(cf(s)),s) \mid s \in \R, f(s) \geq 0\} \subseteq \R^2$ has Lebesgue measure zero because it is the graph of a monotonic continuous real-valued function. As a consequence, the set $U \times \R^{n-2}$ also has zero Lebesgue measure, which further implies that $M^{-1}(U \times \R^{n-2})$ has zero Lebesgue measure because $M$ is an invertible linear transformation.

	
	 Observe that, if $w \in T = \cup_{s: f(s) \geq 0} T(s)$, then there exists $s$ such that $f(s) \geq 0$, $\alpha_1^Tw = f^{-1}(f(s))$ and $\alpha_2^Tw = s$. Hence, the first two components of $Mw$ are $f^{-1}(f(s))$ and $s$ respectively, thereby implying $w \in M^{-1}(U \times \R^{n-2})$. Therefore, we get that $T \subset M^{-1}(U \times \R^{n-2})$ and, as a consequence, $T$ has zero Lebesgue measure. Finally, this implies that $G(S) = 0$ because $G$ has a density. The rest of the analysis is analogous to the one given in the proof of Lemma~\ref{pacing_strategy_properties}.
\end{proof}

\end{document}